\documentclass[onecolumn, 12pt]{article}

\usepackage{amsfonts}
\usepackage{epsfig}

\usepackage{graphicx}
\usepackage{float} 
\usepackage{subfigure}
\usepackage{indentfirst}
\usepackage{setspace}
\usepackage[justification=centering]{caption}
\usepackage{latexsym}
\usepackage{multirow}
\usepackage{amsmath}
\usepackage{epstopdf}
\usepackage{cases}
\usepackage{bbding}

\usepackage{bigdelim}
\usepackage{amssymb}
\usepackage{mathrsfs}
\usepackage{caption}
\usepackage{amsthm}
\usepackage{pstricks}
\usepackage[bookmarks,colorlinks,linkcolor=blue,anchorcolor=blue,citecolor=blue]{hyperref}
\usepackage[left=2cm, right=2cm, top=2cm, bottom=2cm]{geometry}
\usepackage{amsthm}
\usepackage[title]{appendix}
\usepackage{caption}
\usepackage{longtable}
\usepackage{epstopdf}

\newtheorem{definition}{Definition}
\newtheorem{lemma}{Lemma}

\newtheorem{theorem}{Theorem}
\newtheorem{corollary}{Corollary}
\newtheorem{example}{Example}
\allowdisplaybreaks[4]

\begin{document}

\title{A Characterization of Nash Equilibrium in Behavioral Strategies through Local Sequential Rationality\footnote{This work was partially supported by GRF: CityU 11215123 of Hong Kong SAR Government.}}

\author{Yiyin Cao\textsuperscript{\ref{fnote1}} and Chuangyin Dang\textsuperscript{\ref{fnote2}}\footnote{Corresponding Author}}
\date{ }

\maketitle

\footnotetext[1]{School of Management, Xi'an Jiaotong University, Xi'an, China, yiyincao2-c@my.cityu.edu.hk\label{fnote1}}

\footnotetext[2]{Department of Systems Engineering, City University of Hong Kong, Kowloon, Hong Kong, mecdang@cityu.edu.hk \label{fnote2}}

\date{ }
\maketitle

\begin{abstract}

The concept of Nash equilibrium in behavioral strategies (NashEBS) was formulated By Nash~\cite{Nash (1951)} for an extensive-form game through global rationality of nonconvex payoff functions. Kuhn's payoff equivalence theorem resolves the nonconvexity issue, but it overlooks that one Nash equilibrium of the associated normal-form game can correspond to infinitely many NashEBSs of an extensive-form game. To remedy this multiplicity,  the traditional approach as documented in Myerson~\cite{Myerson (1991)}  involves a two-step process: identifying a Nash equilibrium of the agent normal-form representation, followed by verifying whether the corresponding mixed strategy profile is a Nash equilibrium of the associated normal-form game, which often scales exponentially with the size of the extensive-form game tree. In response to these challenges, this paper develops a characterization of NashEBS through local sequential rationality of linear payoff functions, which is achieved with the introduction of an extra behavioral strategy profile and an application of self-independent beliefs. This characterization allows one to analytically determine all NashEBSs for small extensive-form games and directly prove the existence of NashEBS. Building upon this characterization, we acquire a polynomial system to serve as a necessary and sufficient condition for determining whether a behavioral strategy profile is a NashEBS. When the extra strategy profile is identical to the original strategy profile, we gain a strict refinement of NashEBS, which is named as semi-sequential equilibrium.  An application of the characterization yields differentiable path-following methods for computing such an equilibrium.
 

\end{abstract}



\section{Introduction}

An extensive-form game is a broad class of noncooperative games in game theory.  Its mathematical formulation has a tree structure and specifies the physical order of play, the actions available to a player whenever it is the player's turn to move, the rules for determining whose turn to move at any point, the information a player has whenever it is the player's turn to move, the payoffs to the players as the functions of the actions they take, and all the possible actions the chance player has. Nash equilibrium in behavioral strategies (NashEBS), as conceptualized by Nash~\cite{Nash (1951)}, stands as a pivotal and elegant concept in game theory. It delineates a standard of rational behavior within an extensive-form game, incorporating global rationality characterized by nonconvex payoff functions. This rationality condition requires each player to maximize his expected payoff by adapting his actions in response to the actions of other players. Notably, Nash's~\cite{Nash (1951)} definition is fundamentally independent of the tree structure of an extensive-form game and does not involve any players' beliefs about the probability of reaching a specific history within an information set.
Due to the nonconvexity in the rationality condition, the existence of NashEBS of an extensive-form game was established in the existing literature through the associated normal-form game. Furthermore, to address this nonconvexity in computing a NashEBS for an extensive-form game, some studies resort to finding Nash equilibria of the associated normal-form representation of the extensive-form game. However, the number of pure strategy profiles in the associated normal-form game and the reduced normal-form game increases exponentially with the size of the game tree. Besides,  the associated normal-form game and the reduced normal-form game may exhibit a vast array of equilibria, all of which are behaviorally equivalent. To overcome this limitation, Koller and Megiddo~\cite{Koller and Megiddo (1992)} and von Stengel~\cite{von Stengel (1996)} introduced a strategic description of an extensive-form game known as the sequence form. The sequence form resembles a matrix scheme similar to the normal form, but instead of pure strategies, it employs sequences of consecutive moves (referred to as realization plans) that are fewer in number. Although the sequence form offers computational advantages in computing equilibria, the practical significance of realization plans may not be as profound as behavioral strategies, which assign at each information set a probability distribution over the set of possible actions. Additionally, a realization plan in the sequence form could correspond to an infinite number of behavioral strategies and the computational advantage of the sequence form vanishes in the computation of the conditional expected payoffs on information sets.
 
Nash equilibrium places no restriction on the behavior of players at the information sets reached with probability zero. This gives rise to equilibria supported by unreasonable behavior off the equilibrium path. To alleviate this deficiency with the concept of Nash equilibrium, Selten~\cite{Selten (1965), Selten (1975)} initiated the program of refining Nash equilibrium by placing restrictions on the behavior of players at the information sets off the equilibrium path. This led to the development of the concepts of subgame perfect equilibrium and perfect equilibrium. Specifically, a Nash equilibrium is subgame perfect if the players' strategies constitute a Nash equilibrium in every subgame. Selten~\cite{Selten (1975)} introduced the concept of perfect equilibrium as a refinement of both Nash equilibrium and subgame perfect equilibrium. This notion emphasizes sequential rationality in every information set. Specifically, a perfect equilibrium is defined as a limit point of a sequence of $\varepsilon$-perfect equilibria. Another approach to address unreasonable equilibria is to impose two key requirements. Firstly, at each information set, the player with the move must possess a belief about which history within the information set has been reached. Secondly, given these beliefs, the players' actions must align with the principles of sequential rationality. By incorporating the concept of sequential rationality and formulating diverse beliefs, Fudenberg and Tirole~\cite{Fudenberg and Tirole (1991)} and Kreps and Wilson~\cite{Kreps and Wilson (1982)} introduced the concepts of perfect Bayesian equilibrium and sequential equilibrium, which serve as refinements of Nash equilibrium. Several notions of belief systems have been put forward in the existing literature to model how players' beliefs evolve upon receiving new information, including the notion of strong belief proposed by Battigalli and Siniscalchi~\cite{Battigalli and Siniscalchi (2002)}, the belief associated with the principle of ``updating previous beliefs whenever possible" in Ben-Porath~\cite{Ben-Porath (1997)}, and the belief within a sequential communication equilibrium as outlined by Myerson~\cite{Myerson (1986)}.
Given the significance of beliefs in defining equilibria in extensive-form games, Mas-Colell et al.~\cite{Mas-Colell et al. (1995)} proposed a necessary and sufficient condition in their Proposition 9.C.1 to characterize Nash equilibrium with beliefs.
While the characterization by Mas-Colell et al.~\cite{Mas-Colell et al. (1995)} establishes a connection between the concept of Nash equilibrium and the notion of a belief system, it necessitates global rationality within nonconvex payoff functions, thereby requiring the utilization of the associated normal-form game of an extensive-form game.

To bridge this gap, our paper proposes a characterization of NashEBS. This characterization is achieved by introducing an extra behavioral strategy profile and beliefs, leveraging the principles of local sequential rationality and self-independent consistency. In our characterization, local sequential rationality plays a crucial role. It is defined by linear payoff functions at every information set and necessitates that each player maximizes their expected payoff at every information set, given the actions chosen at all other information sets and the beliefs. Another key aspect is the notion of a self-independent belief system, which allows us to express a player's belief in an information set as a probability distribution. The core feature of a self-independent belief system is that each player's belief in an information set is induced by previous actions along a history in the information set, regardless of one's own actions and beliefs along that history.\footnote{The requirement of a self-independent belief system aligns with the condition known as ``same information about others implies same beliefs" in the existing literature, which necessitates that beliefs regarding other players rely solely on information about others' behavior and not on one's own previous actions.} The combination of local sequential rationality and a self-independent belief system in the characterization ensures global rationality. This characterization enables us to analytically determine all NashEBSs for small extensive-form games and establish the existence of NashEBS without relying on associated normal-form games but instead directly leveraging the extensive-form game structure.
As an application of the characterization, we acquire a polynomial system as a necessary and sufficient condition for a NashEBS, which can be exploited in the development of a general method for computing a NashEBS. Additionally, a characterization of subgame perfect equilibrium in behavioral strategies is achieved. 


To the best of our knowledge, the existing literature has not established a general method that does not rely on the associated normal form or the sequence form for computing a NashEBS in $n$-person extensive-form games. Nevertheless, several notable contributions have been made in the computation of NashEBSs, specifically in two-person extensive-form games. Wilson~\cite{Wilson (1972)} extended the Lemke-Howson pivoting method to compute NashEBSs in two-person extensive-form games. von Stengel~\cite{von Stengel (1996)} introduced the sequence form as a replacement for the associated normal-form game. The sequence form, coupled with the Lemke-Howson method, facilitates the computation of Nash equilibria for two-person extensive-form games with zero-sum payoffs. This technique was further applied in Koller et al.~\cite{Koller et al. (1996)} to compute Nash equilibria in general two-person extensive-form games. For a comprehensive review of path-following methods for computing Nash equilibria in extensive-form games, Herings and Peeters~\cite{Herings and Peeters (2010)} provided an excellent resource. Given the favorable attributes of path-following methods demonstrated in existing literature, particularly their ability to tackle global tasks by iteratively employing local approximations, we exploit our characterizations to develop differentiable path-following methods for computing Nash equilibria and subgame perfect equilibria in $n$-person extensive-form games. These methods incorporate logarithmic-barrier terms into payoff functions to constitute logarithmic-barrier extensive-form games in which each player at each of his information sets solves a strictly convex optimization problem. By applying optimality conditions to these optimization problems alongside the equilibrium condition, we derive polynomial equilibrium systems for the barrier games. These systems specify a smooth path to a Nash equilibrium and a subgame perfect equilibrium, respectively. 

The rest of this paper is organized as follows. In Section 2, we introduce notations for extensive-form games and revisit the original definition of NashEBS. We present the characterization of NashEBS in Section 3. We give the characterization of subgame perfect equilibrium in behavioral strategies in Section 4. We exemplify the application of our characterizations by presenting three illustrative examples in Section 5. 
We develop differentiable path-following methods to compute Nash equilibria and subgame perfect equilibria in Section 6. Comprehensive numerical experiments are reported in Section 7 to further confirm the efficiency of the methods. The concluding remarks are made in Section 8.

\section{Extensive-Form Games and Nash Equilibrium in Behavioral Strategies}

This paper is concerned with finite extensive-form games with perfect recall.  To describe such a game in accordance with Osborne and Rubinstein~\cite{Osborne and Rubinstein (1994)}, we need some necessary notations, which are summarized in Table~\ref{Table} for ease of reference. A selection of these notations is illustrated in Fig.~\ref{Notation}.
\begin{table}[H]\setlength{\abovedisplayskip}{1.2pt}
\setlength{\belowdisplayskip}{1.2pt}
\linespread{1.2} 
\footnotesize
\centering
\caption{Notation for Extensive-Form Games}
\label{Table}
\begin{tabular*}{\hsize}{@{}@{\extracolsep{\fill}}l|l@{}}
\hline
Notation & Terminology\\
\hline
$N=\{1,2,\ldots,n\}$ & Set of players without the chance player\\ 
$h=\langle a_1,a_2,\ldots,a_L\rangle$ &  A history, which is a sequence of actions taken by players\\ 
$H$, $\emptyset\in H$ &  Set of histories, $\langle a_1,\ldots,a_L\rangle\in H$ if $\langle a_1,\ldots,a_K\rangle\in H$ and $L<K$\\ 
$Z$ & Set of terminal histories\\ 
$A(h)=\{a:\langle h,a\rangle\in H\}$ & Set of actions after a nonterminal history $h\in H$\\ 
$P(h)$ & Player who takes an action after a history $h\in H$\\ 
${\cal I}_i$ & Collection of information partitions of player $i$\\ 
$I_i^j\in {\cal I}_i$, $j\in M_i=\{1,\ldots,m_i\}$ & $j$th information set of player $i$, $A(h)=A(h')$ whenever $h,h'\in I_i^j$\\ 
$A(I^j_i)=A(h)$ for any $h\in I^j_i$ & Set of actions of player $i$ at information set $I^j_i$\\ 
$\beta=(\beta^i_{I_i^j}(a):i\in N,I_i^j\in{\cal I}_i,a\in A(I_i^j))$ & Profile of behavioral strategies\\ 
$\beta^i=(\beta^i_{I_i^j}:j\in M_i)$ & Behavioral strategy of player $i$\\  
$\beta^{-i}=(\beta^p_{I^p_q}:p\in N\backslash\{i\},q\in M_p)$ & Profile of behavioral strategies without $\beta^i$\\ 
$\beta^i_{I_i^j}=(\beta^i_{I_i^j}(a):a\in A(I_i^j))^\top$ & Probability measure over $A(I_i^j)$ and $\beta^i_h=\beta^i_{I^j_i}$ for any $h\in I^j_i$\\
$\beta^{-I^j_i}=(\beta^p_{I^p_q}:p\in N,q\in M_p,I^q_p\ne I^j_i)$ & Profile of behavioral strategies without $\beta^i_{I^j_i}$\\ 
$f_c(\cdot|h)=(f_c(a|h):a\in A(h))^\top$ & Probability measure of the chance player $c$ over $A(h)$ \\ 
$\mu=(\mu^i_{I_i^j}(h):i\in N,I_i^j\in{\cal I}_i,h\in I_i^j)$ & Belief system, $\sum_{h\in I_i^j}\mu^i_{I_i^j}(h)=1$ and $\mu^i_{I_i^j}(h)\ge 0$ for all $h\in I_i^j$\\ 
$h\cap A(I^j_i)$; $a\in h$ & $\{a_1,\ldots,a_L\}\cap A(I^j_i)$; $a\in \{a_1,\ldots,a_L\}$ for $h=\langle a_1,\ldots,a_L\rangle$\\
$u^i:Z\rightarrow\mathbb{R}$ & Payoff function of player $i$\\ \hline
\end{tabular*}
\end{table}
\begin{figure}[ht]
    \centering
    \begin{minipage}{0.40\textwidth}
        \centering
        \includegraphics[width=0.9\textwidth, height=0.20\textheight]{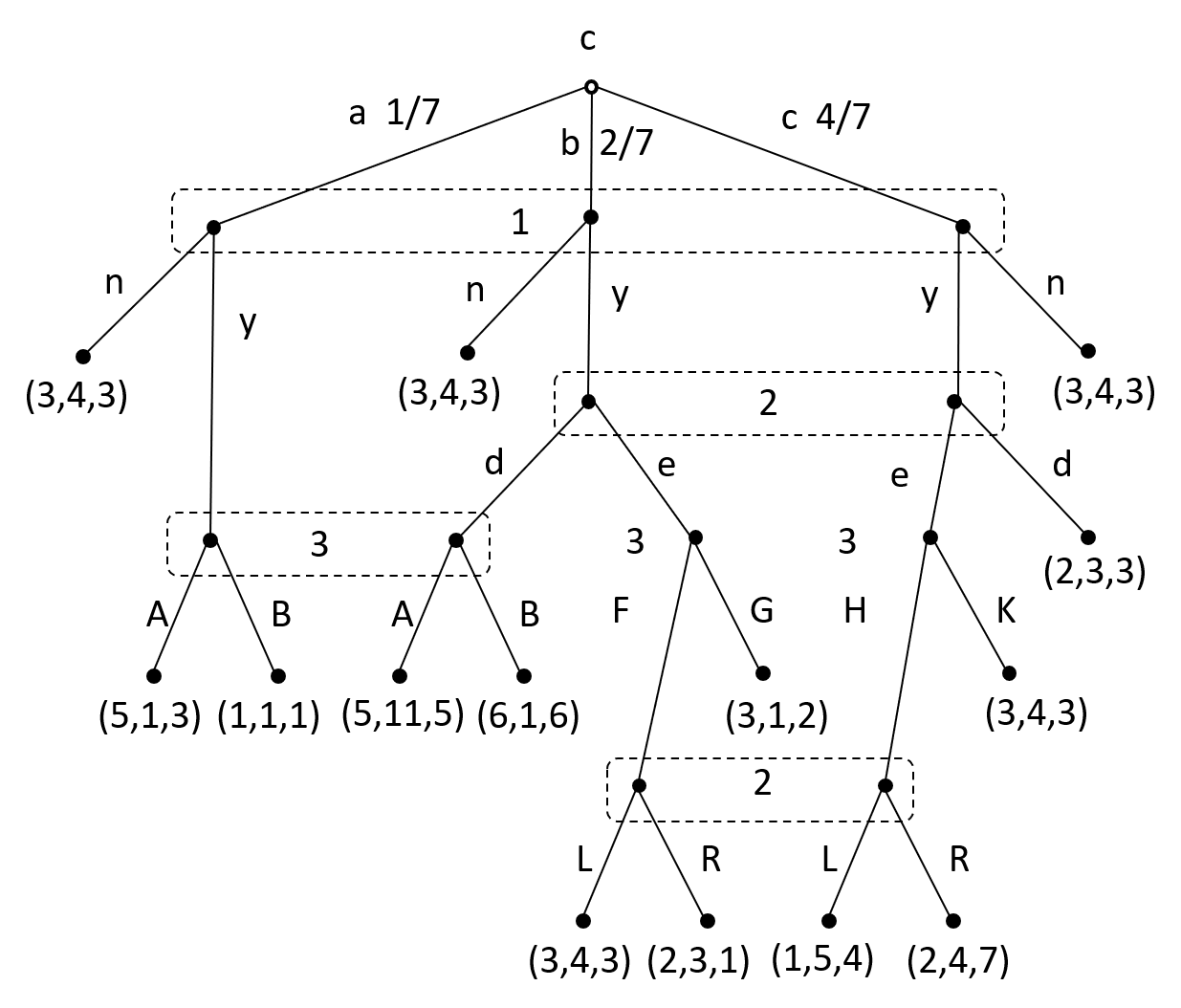}
\end{minipage}\hfill
    \begin{minipage}{0.60\textwidth}
       {\scriptsize
       Information sets: ${\cal I}_1=\{I^1_1\}$, ${\cal I}_2=\{I^1_2,I^2_2, I^3_2\}$,  ${\cal I}_3=\{I^1_3,I^2_3\}$, $I^1_1=\{\langle a\rangle,\langle b\rangle,\langle c\rangle\}$, $I^1_2=\{\langle b, y\rangle, \langle c, y\rangle\}$, $I^2_2=\{\langle b, y, e, F\rangle, \langle c, y, e, H\rangle\}$, $I^3_2=\{\langle a,  y \rangle, \langle b,  y, d \rangle\}$, $I^1_3=\{\langle b,  y, e \rangle\}$, $I^2_3=\{\langle c,  y, e \rangle\}$.\newline Player who takes an action after a history: $P(\emptyset)=c$, $P(\langle b\rangle)=1$, $P(\langle b, y\rangle)= 2$, $P(\langle c, y, e, H\rangle)=2$, $P(\langle b, y, d\rangle)=2$, $P(\langle b, y, e\rangle)=3$.\newline
       Action sets: $A(\emptyset)=\{a,b,c\}$, $A(I^1_1)=\{n, y\}$,  $A(I^2_2)=\{L, R\}$, $A(I^2_3)=\{H, K\}$. \newline
       Behavioral strategies: $f_c(\cdot|\emptyset)=(f_c(a|\emptyset), f_c(b|\emptyset),f_c(c|\emptyset))^\top=(\frac{1}{7},\frac{2}{7},\frac{4}{7})^\top$,
       $\beta^1_{I^1_1}=(\beta^1_{I^1_1}(n),\beta^1_{I^1_1}(y))^\top$, $\beta^2_{I^2_2}=(\beta^2_{I^2_2}(L),\beta^2_{I^2_2}(R))^\top$,  $\beta^3_{I^2_3}=(\beta^3_{I^2_3}(H),\beta^3_{I^2_3}(K))^\top$.\newline
       Beliefs: $\mu^2_{I^1_2}=(\mu^2_{I^1_2}(\langle b, y\rangle), \mu^2_{I^1_2}(\langle c, y\rangle))^\top$, $\mu^2_{I^3_2}=(\mu^2_{I^3_2}(\langle a, y\rangle), \mu^2_{I^3_2}(\langle b, y, d\rangle))^\top$.}
        \end{minipage}
\caption{\label{Notation}\scriptsize Illustrations of Some Notations}
\end{figure}

Given these notations, we represent  by $\Gamma=\langle N, H, P, f_c, \{{\cal I}_i\}_{i\in N}\rangle$ an extensive-form game. A finite extensive-form game means an extensive-form game with a finite number of histories. Let $X_i(h)$ be the record of player $i$'s experience along the history $h$. Then $X_i(h)$ is the sequence consisting of the information sets that player $i$ encounters in the history $h$ and the actions he takes at them in the order that these events occur. An extensive-form game has {\bf perfect recall} if, for each player $i$, we have $X_i(h')=X_i(h'')$ whenever the histories $h'$ and $h''$ are in the same information set of player $i$. 
When $P(\langle a_1,\ldots,a_k\rangle)=c$, $\beta^c_{\langle a_1,\ldots,a_k\rangle}(a_{k+1})=f_c(a_{k+1}|\langle a_1,\ldots,a_k\rangle)$.  For $i\in N$ and $j\in M_i$, we often write a behavioral strategy profile $\beta$ as $(\beta^i_{I_i^j},\beta^{-I^j_i})$. 
 The probability that the moves along a history $h=\langle a_1,\ldots,a_L\rangle\in H$ are played when $\beta$ is taken by players is represented as \begin{equation}\setlength{\abovedisplayskip}{1.2pt}
\setlength{\belowdisplayskip}{1.2pt}
\label{omega1}\omega(h|\beta)=\prod\limits_{k=0}^{L-1}\beta^{P(\langle a_1,\ldots,a_k\rangle)}_{\langle a_1,\ldots,a_k\rangle}(a_{k+1}).\end{equation} 
 Then, for $i\in N$, $j\in M_i$, $h=\langle a_1,\ldots,a_L\rangle\in H$, and $a\in A(I^j_i)$, we have \begin{equation}\label{omega1A}
 \setlength{\abovedisplayskip}{1.2pt}
\setlength{\belowdisplayskip}{1.2pt}
 \omega(h|a,\beta^{-I^j_i})=\mathop{\prod\limits_{k=0}^{L-1}}_{\langle a_1,\ldots,a_k\rangle\notin I^j_i}\beta^{P(\langle a_1,\ldots,a_k\rangle)}_{\langle a_1,\ldots,a_k\rangle}(a_{k+1}).\end{equation}
 For $i\in N$ and $j\in M_i$, let 
 \begin{equation}\label{omega1B}\setlength{\abovedisplayskip}{1.2pt}
\setlength{\belowdisplayskip}{1.2pt}
\omega(I^j_i|\beta)=\sum\limits_{h\in I^j_i}\omega(h|\beta),
\end{equation}
which equals the probability that information set $I^j_i$ is reached when $\beta$ is played.
 For the game in Fig.~\ref{Notation}, we have $\omega(\langle b,y,e,F \rangle|\beta)=\frac{2}{7}\beta^1_{\langle b\rangle}(y)$ $\beta^2_{\langle b, y\rangle}(e)\beta^3_{\langle b, y, e\rangle}(F)=\frac{2}{7}\beta^1_{I^1_1}(y)\beta^2_{I^1_2}(e)\beta^3_{I^1_3}(F)$ and $\omega(\langle b,y,e,F \rangle|y,\beta^{-I^1_1})=\frac{2}{7}\beta^2_{I^1_2}(e)\beta^3_{I^1_3}(F)$.
 The expected payoff of player $i$ at a behavioral strategy profile $\beta$ is given by
\begin{equation}\setlength{\abovedisplayskip}{1.2pt}
\setlength{\belowdisplayskip}{1.2pt}
\label{payoffs}u^i(\beta)=\sum\limits_{h\in Z}u^i(h)\omega(h|\beta).\end{equation}
With these notations, Nash~\cite{Nash (1951)} formulated the definition of Nash equilibrium in behavioral strategies as follows.
\begin{definition}[{\bf Nash Equilibrium in Behavioral Strategies (NashEBS)}, Nash~\cite{Nash (1951)}]\label{ned1} {\em A behavioral strategy profile $\beta^*$ is 
a NashEBS of an extensive-form game if $u^i(\beta^*)\ge u^i(\beta^i,\beta^{*-i})$ for any $\beta^i$.  
}\end{definition}
A profile of behavioral strategies represents a NashEBS when no player can improve their expected payoff by unilaterally deviating to a different behavioral strategy. In this paper, the requirement of $\beta^*$ as outlined in Definition~\ref{ned1} is termed as {\bf global rationality}:  

\textit{A behavioral strategy profile  $\beta^*$ possesses global rationality if, for every player $i$, we have 
\(u^i(\beta^*)\ge u^i(\beta^i,\beta^{*-i})\) for every $\beta^i$ of player $i$.}

In other words, a behavioral strategy profile $\beta^*$ exhibits global rationality when each player maximizes his expected payoff by adapting his actions given the actions of other players. To deal with global rationality in the computation of a NashEBS of an extensive-form game, some studies invoke its associated normal-form or reduced normal-form representation, given that $u^i(\beta^i,\beta^{*-i})$ is a nonconvex function of $\beta^i$. However, the number of pure strategy profiles in the associated normal-form and reduced normal-form game grows exponentially with the size of the extensive-form game. Besides, the associated normal-form and reduced normal-form game could have an enormous multiplicity of NashEBSs, all of which are payoff equivalent.  One may think that a Nash equilibrium of the agent normal-form representation of an extensive-form game yields a NashEBS. However, as revealed on pp. 161 in Myerson~\cite{Myerson (1991)} and on pp. 374 in Bonanno~\cite{Bonanno (2018)}, a Nash equilibrium of the agent normal-form representation of an extensive-form game may not be a NashEBS according to Definition~\ref{ned1}. As pointed out in Myerson~\cite{Myerson (1991)}, a NashEBS of an extensive-form game is defined to be any Nash equilibrium $\beta$ of the agent normal-form representation such that the mixed representation of $\beta$ is also a Nash equilibrium of the associated normal-form game. Additionally, Definition~\ref{ned1} is essentially detached from the tree structure of an extensive-form game and does not entail any beliefs. 

Given the significance of beliefs in defining equilibria in extensive-form games, Mas-Colell et al.~\cite{Mas-Colell et al. (1995)} proposed a characterization of Nash equilibrium with beliefs.
To aid in this characterization, the following notations are introduced.
For $i\in N$ and $j\in M_i$, when $\beta$ and $(a,\beta^{-I^j_i})$ are taken, the expected payoffs of player $i$ along terminal histories that intersect with the action set of player $i$ at the information set $I^j_i$ can be represented as  \begin{equation}\setlength{\abovedisplayskip}{1.2pt}
\setlength{\belowdisplayskip}{1.2pt}
\label{ispayoffs}u^i(\beta\land I^j_i)=\sum\limits^
{h\cap A(I^j_i)\ne\emptyset}_{h\in Z}u^i(h)\omega(h|\beta)\text{\;\;\;\;and\;\;\;\;}
 u^i((a,\beta^{-I^j_i})\land I^j_i)=\sum\limits_{a\in h\in Z}u^i(h)\omega(h|a,\beta^{-I^j_i}).\end{equation} 
In line with Kreps and Wilson~\cite{Kreps and Wilson (1982)}, we refer to $\mu$ as a belief system, which assigns to every information set a probability measure on the set of histories in the information set. A further elaboration of $\mu$ is that $\mu^i_{I^j_i}(h)$ denotes the probability assigned to the occurrence of history $h$ when it is player $i$'s turn to choose an action at $I^j_i$. To represent conditional expected payoffs at each information set using a belief system $\mu$, for $i\in N$, $j\in M_i$ and $h=\langle a_1,\ldots,a_K\rangle\in Z$, let {\small
\begin{equation}\label{cps1a}\setlength{\abovedisplayskip}{1.2pt}
\setlength{\belowdisplayskip}{1.2pt}
\nu^i_{I^j_i}(h|\beta,\mu)=\left\{\begin{array}{ll} 
\mu^i_{I^j_i}(\hat h)\prod\limits_{k=L}^{K-1}\beta^{P(\langle a_1,\ldots,a_k\rangle)}_{\langle a_1,\ldots,a_k\rangle}(a_{k+1}) & \text{if $\hat h=\langle a_1,\ldots,a_L\rangle\in I^j_i$,}\\

0 & \text{if there is no subhistory of $h$ in $I^j_i$,}
\end{array}\right.\end{equation}}
and {\small
\begin{equation}\label{cps1b}\setlength{\abovedisplayskip}{1.2pt}
\setlength{\belowdisplayskip}{1.2pt}
\nu^i_{I^j_i}(h|a,\beta^{-I^j_i},\mu)=\left\{\begin{array}{ll} 
\mu^i_{I^j_i}(\hat h)\prod\limits_{k=L+1}^{K-1}\beta^{P(\langle a_1,\ldots,a_k\rangle)}_{\langle a_1,\ldots,a_k\rangle}(a_{k+1}) & \text{if $\hat h=\langle a_1,\ldots,a_L\rangle\in I^j_i$,}\\

0 & \text{if there is no subhistory of $h$ in $I^j_i$.}
\end{array}\right.\end{equation}}$\nu^i_{I^j_i}(h|\beta,\mu)$ and $\nu^i_{I^j_i}(h|a,\beta^{-I^j_i},\mu)$ represent the probabilities that the moves along $h$ are played given that $I^j_i$ has been reached when $(\beta,\mu)$ or $(a,\beta^{-I^j_i},\mu)$ is taken, in which $\mu^i_{I^j_i}(\hat h)$ denotes the probability that $\hat h\in I^j_i$ along $h$ has occurred conditional on $I^j_i$ has been reached.
Given these notations, the conditional expected payoffs of player $i$ on information set $I^j_i$ at $(\beta,\mu)$ and $(a,\beta^{-I^j_i},\mu)$ can be denoted as  {\small\begin{equation}\setlength{\abovedisplayskip}{1.2pt}
\setlength{\belowdisplayskip}{1.2pt}
\label{cepayoffs}
u^i(\beta,\mu|I^j_i)=\sum\limits^
{h\cap A(I^j_i)\ne\emptyset}_{h\in Z}u^i(h)\nu^i_{I^j_i}(h|\beta,\mu)
\text{ and }
u^i(a,\beta^{-I^j_i},\mu|I^j_i)=\sum\limits_{a\in h\in Z}u^i(h)\nu^i_{I^j_i}(h|a,\beta^{-I^j_i},\mu).\end{equation}}
\begin{theorem}[Mas-Colell et al.~\cite{Mas-Colell et al. (1995)}]\label{ned2}{\em
 A behavioral strategy profile $\beta^*$ is 
a NashEBS of a finite extensive-form game with perfect recall if there exists a system of beliefs $\mu^*$ such that\\
(i). For $i\in N$ and $j\in M_i$ with $\omega(I^j_i|\beta^*)>0$,   $u^i(\beta^*,\mu^*|I^j_i)\ge u^i(\beta^i,\beta^{*-i},\mu^*|I^j_i)$ for all $\beta^i$.\\
(ii). $\mu^*$ is derived from $\beta^*$ through Bayes' rule whenever possible.
 }
\end{theorem}

Although the result in Theorem~\ref{ned2} establishes a connection between NashEBS and a belief system, Theorem~\ref{ned2} still demands global rationality of nonconvex conditional expected payoff functions, leading to the need to solve dynamic programming problems. To overcome the deficiencies in Definition~\ref{ned1} and Theorem~\ref{ned2} of NashEBS, we develop in this paper a characterization of NashEBS. With this characterization, one can find NashEBSs of an extensive-form game without invoking its associated (reduced) normal-form game or engaging in dynamic programming procedures.

\section{\large A Characterization of Nash Equilibrium in Behavioral Strategies}

The following notations are introduced to facilitate our characterization of NashEBS. 
To distinguish the behavioral strategies of player $i$ at the information set succeeding $I^j_i$ from other behavioral strategies,  we define
$\varrho^i_{I^j_i}(\beta,\tilde\beta)=(\varrho^i_{I^j_i}(\beta^p_{I^q_p},\tilde\beta):p\in N,q\in M_p)$ for $i\in N$ and $j\in M_i$,  where  \begin{equation}\label{varrho1}\setlength{\abovedisplayskip}{1.2pt}
\setlength{\belowdisplayskip}{1.2pt}
\varrho^i_{I^j_i}(\beta^p_{I^q_p},\tilde\beta)=\left\{\begin{array}{ll}
\tilde\beta^i_{I^q_i} & \text{if $p=i$ and $h\cap A(I^j_i)\ne\emptyset$ for some $h\in I^q_i$,}\\
\beta^p_{I^q_p} & \text{otherwise.}\\
\end{array}\right.\end{equation}
For the game in Fig.~\ref{Notation}, we have $\varrho^2_{I^1_2}(\beta^2_{I^1_2},\tilde\beta)=\beta^2_{I^1_2}$, $\varrho^2_{I^1_2}(\beta^2_{I^2_2},\tilde\beta)=\tilde\beta^2_{I^2_2}$, and $\varrho^2_{I^1_2}(\beta^3_{I^1_3},\tilde\beta)=\beta^3_{I^1_3}$.
To formulate a self-independent belief system that is employed in our characterization of NashEBS, for $h=\langle a_1,\ldots,a_L\rangle\in H$,
let \begin{equation}
\setlength{\abovedisplayskip}{1.2pt}
\setlength{\belowdisplayskip}{1.2pt}
\label{sipf1}
{\cal S}^i(h|\beta)=\mathop{\prod\limits_{k=0}^{L-1}}\limits_{P(\langle a_1,\ldots,a_k\rangle)\ne i}\beta^{P(\langle a_1,\ldots,a_k\rangle)}_{\langle a_1,\ldots,a_k\rangle}(a_{k+1}),
\end{equation}
which denotes the probability that the moves along $h$ are executed, regardless of player $i$'s actions.
We specify  \begin{equation}\setlength{\abovedisplayskip}{1.2pt}
\setlength{\belowdisplayskip}{1.2pt}\label{sipf2}
{\cal S}^i(I^j_i|\beta)=\sum\limits_{h\in I^j_i}{\cal S}^i(h|\beta).
\end{equation} For the game in Fig.~\ref{Notation}, we have ${\cal S}^2(\langle b,y,e,F\rangle |\beta)=\frac{2}{7}\beta^1_{\langle b\rangle}(y)\beta^3_{\langle b, y, e\rangle}(F)$ and ${\cal S}^2(I^2_2|\beta)={\cal S}^2(\langle b,y,e,F\rangle |\beta)+{\cal S}^2(\langle c,y,e,H\rangle |\beta)$.
 For any $h=\langle a_1,\ldots,a_L\rangle$, it holds that
\[\setlength{\abovedisplayskip}{1.2pt}
\setlength{\belowdisplayskip}{1.2pt}
\omega(h|\beta)= {\cal S}^i(h|\beta)\mathop{\prod\limits_{k=0}^{L-1}}\limits_{P(\langle a_1,\ldots,a_k\rangle) =i}\beta^{P(\langle a_1,\ldots,a_k\rangle)}_{\langle a_1,\ldots,a_k\rangle}(a_{k+1}).\] 
For $i\in N$, $j\in M_i$ and $a\in A(I^j_i)$, let 
\[\setlength{\abovedisplayskip}{1.2pt}
\setlength{\belowdisplayskip}{1.2pt}
M(a,I^j_i)=\left\{q\in M_i\left|\begin{array}{l}\text{for any
 $h=\langle a_1,\ldots,a_L\rangle\in I^q_i$, there exists $1\le \ell\le L$ such}\\
 \text{that $a_\ell=a$ and $\{a_{\ell+1},\ldots,a_L\}\cap A(I^p_i)=\emptyset$ for all $p\in M_i$}\end{array}\right.\right\},\]
and $M(I^j_i)=\mathop{\cup}\limits_{a\in A(I^j_i)}M(a,I^j_i)$. For $q\in M(a,I^j_i)$, the information set $I^q_i$ represents the first information set of player $i$ subsequent to the action $a\in A(I^j_i)$. For the game in Fig.~\ref{Notation}, we have $M(e, I^1_2)=\{ 2\}$, and $M(I^1_2)=\{ 2, 3\}$. 
 As a result of the perfect recall and the definition of ${\cal S}^i(h|\beta)$ and $\omega(h|\beta)$, one can easily derive the following conclusions.
 \begin{lemma}\label{ednelm1}{\em 
 (i) ${\cal S}^i(I^j_i|\beta)>0$ if ${\cal S}^i(I^q_i|\beta)>0$ for some $q\in M(I^j_i)$.
 (ii) ${\cal S}^i(I^q_i|\beta)=0$ for all $q\in M(I^j_i)$ if ${\cal S}^i(I^j_i|\beta)=0$.
 (iii) $\omega(I^j_i|\beta)>0$ if $\omega(I^q_i|\beta)>0$ for some $q\in M(I^j_i)$. (iv) $\omega(I^q_i|\beta)=0$ for all $q\in M(I^j_i)$ if $\omega(I^j_i|\beta)=0$.}
 \end{lemma}
 \begin{proof} {\small (i) Let $q\in M(a,I^j_i)$ such that ${\cal S}^i(I^q_i|\beta)>0$. Then there exists some $h=\langle a_1,\ldots,a_L\rangle\in I^q_i$ such that ${\cal S}^i(h|\beta)>0$ and $\{a_1,\ldots,a_L\}\cap A(I^j_i)=\{a\}$. Let $a_\ell=a$. Then, $\hat h=\langle a_1,\ldots,a_{\ell-1}\rangle\in I^j_i$. Thus, ${\cal S}^i(\hat h|\beta)>0$. Therefore, ${\cal S}^i(I^j_i|\beta)>0$. 
(ii) For any $q\in M(I^j_i)$ and any  $h=\langle a_1,\ldots,a_L\rangle\in I^q_i$, there exist $a\in A(I^j_i)$ and $1\le\ell\le L$ such that $a_\ell=a$ and $\hat h=\langle a_1,\ldots,a_{\ell-1}\rangle\in I^j_i$. Clearly, when ${\cal S}^i(I^j_i|\beta)=0$, we have ${\cal S}^i(\hat h|\beta)=0$ for any $\hat h\in I^j_i$. Therefore,  ${\cal S}^i(I^q_i|\beta)=0$ for all $q\in M(I^j_i)$ if ${\cal S}^i(I^j_i|\beta)=0$. One can prove the results of (iii) and (iv) similarly. The proof of the lemma is completed.}
 \end{proof}
 
Given these notations, we construct a self-independent belief system as follows. A belief system $\mu^*=(\mu^{*i}_{I^j_i}(h):i\in N,j\in M_i,h\in I^j_i)$ is a {\bf self-independent belief system} if, for any given $\beta^*$, $\mu^*$ is a solution to the system, 
\begin{equation}\setlength{\abovedisplayskip}{1.2pt}
\setlength{\belowdisplayskip}{1.2pt}
\label{bsne1}\begin{array}{l}
{\cal S}^i(I^j_i|\beta^*)\mu^i_{I^j_i}(h)={\cal S}^i(h|\beta^*),\;i\in N,j\in M_i,h\in I^j_i,\\
\sum\limits_{h'\in I^j_i}\mu^i_{I^j_i}(h')=1,\;0\le \mu^i_{I^j_i}(h),\;i\in N,j\in M_i,h\in I^j_i.
\end{array}\end{equation}

The key aspect of a self-independent belief system is that each player's belief in an information set is induced by previous actions along a history in the information set, independent of one's own actions along that history. By employing the notion of a self-independent belief system and the formulations provided above, we develop our characterization of NashEBS.

\begin{definition}[\bf An Equivalent Definition of NashEBS through an Extra Strategy Profile, Local Sequential Rationality, and Self-Independent Beliefs]\label{edne1}{\em A behavioral strategy profile $\beta^*$ is a NashEBS if $\beta^*$ together with an extra behavioral strategy profile $\tilde\beta^*$ and a belief system $\mu^*$ satisfies the properties:\newline
\noindent (i) $\beta^{*i}_{I^j_i}(a')=0$ for any $i\in N$, $j\in M_i$ and $a',a''\in A(I^j_i)$ with $u^i((a'',\varrho^i_{I^j_i}(\beta^{*-I^j_i},\tilde\beta^*))\land I^j_i)
>u^i((a',\varrho^i_{I^j_i}(\beta^{*-I^j_i},\tilde\beta^*))\land I^j_i)$;\newline
\noindent (ii) $\tilde\beta^{*i}_{I^j_i}(a')=0$ for any $i\in N$, $j\in M_i$ and $a',a''\in A(I^j_i)$ with $u^i(a'',\varrho^i_{I^j_i}(\beta^{*-I^j_i},\tilde\beta^*),\mu^*|I^j_i)
>u^i(a',\varrho^i_{I^j_i}(\beta^{*-I^j_i},\tilde\beta^*),\mu^*|I^j_i)$; and\newline
\noindent (iii) $\mu^*=(\mu^{*i}_{I^j_i}(h):i\in N,j\in M_i,h\in I^j_i)$ is a self-independent belief system, serving as a solution to the system~(\ref{bsne1}).}\end{definition}

We will illustrate with two examples how one can employ Definition~\ref{edne1} to analytically determine all NashEBSs for small extensive-form games in Section~\ref{examples}.  

In this paper, the requirements (i) and (ii) of $(\beta^*,\tilde\beta^*)$ as outlined in Definition~\ref{edne1} are termed as {\bf local sequential rationality}: 

\textit{A behavioral strategy profile $\beta^*$ together with $\tilde\beta^*$ possesses local rationality at an information set $I^j_i$ if
\(u^i((\beta^{*i}_{I^j_i},\varrho^i_{I^j_i}(\beta^{*-I^j_i},\tilde\beta^*))\land I^j_i)\ge u^i((\beta^i_{I^j_i},\varrho^i_{I^j_i}(\beta^{*-I^j_i},\tilde\beta^*))\land I^j_i)\)
for every $\beta^i_{I^j_i}$ of player $i$. A behavioral strategy profile $\beta^*$ together with $\tilde\beta^*$ possesses local sequential rationality if it meets the local rationality at every information set.}

\textit{
A behavioral strategy profile $\tilde\beta^*$ together with $(\beta^*, \mu^*)$ possesses local rationality at an information set $I^j_i$ if
\(u^i(\tilde\beta^{*i}_{I^j_i},\varrho^i_{I^j_i}(\beta^{*-I^j_i},\tilde\beta^*),\mu^*|I^j_i)\ge u^i(\tilde\beta^i_{I^j_i},\varrho^i_{I^j_i}(\beta^{*-I^j_i},\tilde\beta^*),\mu^*|I^j_i)\)
for every $\tilde\beta^i_{I^j_i}$ of player $i$. A behavioral strategy profile $\tilde\beta^*$ together $(\beta^*,\mu^*)$ possesses local sequential rationality if it meets the local rationality at every information set.}


The fundamental principles of local sequential rationality entail that it is defined by linear payoff functions at every information set and necessitates that each player maximizes their expected payoff at every information set, given the actions chosen at all other information sets and the beliefs. The introduction of {\bf an extra behavioral strategy profile $\tilde\beta$}  facilitates us to ensure the global rationality by invoking the local sequential rationality. 

\begin{theorem}\label{ednethm1}{\em Definition~\ref{edne1} and Definition~\ref{ned1} of NashEBS are equivalent.}
\end{theorem}
\begin{proof} {\small ($\Rightarrow$). Definition~\ref{edne1} implies Definition~\ref{ned1}. We denote by $(\beta^*, \tilde\beta^*,\mu^*)$ a triple satisfying the properties in Definition~\ref{edne1}. 
Then, for any $i\in N$ and $j\in M_i$ with $\omega(I^j_i|\beta^*)>0$, we have $\mu^{*i}_{I^j_i}(h)=\frac{{\cal S}^i(h|\beta^*)}{{\cal S}^i(I^j_i|\beta^*)}=\frac{\omega(h|\beta^*)}{\omega(I^j_i|\beta^*)}$ for any $h\in I^j_i$. Thus, for any $i\in N$ and $j\in M_i$ with $\omega(I^j_i|\beta^*)>0$, it holds that, for $a',a''\in A(I^j_i)$,  $u^i((a'',\varrho^i_{I^j_i}(\beta^{*-I^j_i},\tilde\beta^*))\land I^j_i)
>u^i((a',\varrho^i_{I^j_i}(\beta^{*-I^j_i},\tilde\beta^*))\land I^j_i)$ if and only if $u^i(a'',\varrho^i_{I^j_i}(\beta^{*-I^j_i},\tilde\beta^*),\mu^*|I^j_i)
>u^i(a',\varrho^i_{I^j_i}(\beta^{*-I^j_i},\tilde\beta^*),\mu^*|I^j_i)$. Therefore,  for any $i\in N$ and $j\in M_i$ with $\omega(I^j_i|\beta^*)>0$, it follows from Definition~\ref{edne1} that {\footnotesize
\begin{equation}\label{edneeq1}\setlength{\abovedisplayskip}{1.2pt}
\setlength{\belowdisplayskip}{1.2pt}
\begin{array}{rl}
  &  u^i((\beta^{*i}_{I^j_i}, \varrho^i_{I^j_i}(\beta^{*-I^j_i},\tilde\beta^*))\land I^j_i)\\
  = & \sum\limits_{a\in A(I^j_i)}\beta^{*i}_{I^j_i}(a)u^i((a, \varrho^i_{I^j_i}(\beta^{*-I^j_i},\tilde\beta^*))\land I^j_i)
 = \omega(I^j_i|\beta^*)\sum\limits_{a\in A(I^j_i)}\beta^{*i}_{I^j_i}(a)u^i(a, \varrho^i_{I^j_i}(\beta^{*-I^j_i},\tilde\beta^*), \mu^*|I^j_i) \\
 = & \omega(I^j_i|\beta^*)\sum\limits_{a\in A(I^j_i)}\tilde\beta^{*i}_{I^j_i}(a)u^i(a, \varrho^i_{I^j_i}(\beta^{*-I^j_i},\tilde\beta^*), \mu^*|I^j_i)
 =\sum\limits_{a\in A(I^j_i)}\tilde\beta^{*i}_{I^j_i}(a)u^i((a, \varrho^i_{I^j_i}(\beta^{*-I^j_i},\tilde\beta^*))\land I^j_i)\\
 = & u^i((\tilde\beta^{*i}_{I^j_i}, \varrho^i_{I^j_i}(\beta^{*-I^j_i},\tilde\beta^*))\land I^j_i).
 \end{array}\end{equation}}Given these results, we next acquire that $u^i(\beta^*)\ge u^i(\beta^i,\beta^{*-i})$ for any $\beta^i$. For this purpose, we first need to prove that, for any $i\in N$ and $j\in M_i$, \begin{equation}\setlength{\abovedisplayskip}{1.2pt}
\setlength{\belowdisplayskip}{1.2pt}\label{edneA} u^i(\beta^*\land I^j_i)=u^i(\varrho^i_{I^j_i}(\beta^*,\tilde\beta^*)\land I^j_i).\end{equation}
  For any $i\in N$ and $j\in M_i$ with $\omega(I^j_i|\beta^*)=0$, it follows from the definition of $\varrho^i_{I^j_i}(\beta^*,\tilde\beta^*)$ in Eq.~(\ref{varrho1}) that $u^i(\beta^*\land I^j_i)=u^i(\varrho^i_{I^j_i}(\beta^*,\tilde\beta^*)\land I^j_i)=0$. So we only need to prove that  $u^i(\beta^*\land I^j_i)=u^i(\varrho^i_{I^j_i}(\beta^*,\tilde\beta^*)\land I^j_i)$ for $i\in N$ and $j\in M_i$ with $\omega(I^j_i|\beta^*)>0$. For $i\in N$, $j\in M_i$, and $a\in A(I^j_i)$, let 
\[\setlength{\abovedisplayskip}{1.2pt}
\setlength{\belowdisplayskip}{1.2pt}
Z^0(a, I^j_i)=\left\{h=\langle a_1,\ldots,a_K\rangle\in Z\left|\begin{array}{l}
 \text{$a_\ell=a$ for some $1\le \ell\le K$ and}\\
 \text{$\{a_{\ell+1},\ldots,a_K\}\cap A(I^q_i)=\emptyset$ for all $q\in M_i$}\end{array}\right.\right\},\] and $Z^0(I^j_i)=\mathop{\cup}\limits_{a\in A(I^j_i)}Z^0(a, I^j_i)$.  For the game in Fig.~\ref{Notation}, we have $Z^0(e, I^1_2)=\{\langle  b, y, e, G\rangle, \langle c, y, e, K\rangle \}$, and $Z^0(I^1_2)=\{\langle  b, y, e, G\rangle, \langle c, y, e, K\rangle, \langle c, y, d\rangle\}$. 

\noindent {\bf Case (1)}. Consider $i\in N$ and $j\in M_i$ with $\omega(I^j_i|\beta^*)>0$ and $M(I^j_i)=\emptyset$. 
Then, $\varrho^i_{I^j_i}(\beta^*,\tilde\beta^*)=\beta^*$. 
Thus, $u^i(\beta^*\land I^j_i)=u^i(\varrho^i_{I^j_i}(\beta^*,\tilde\beta^*)\land I^j_i)$.
 
\noindent {\bf Case (2)}. Consider $i\in N$ and $j\in M_i$ with $\omega(I^j_i|\beta^*)>0$ and $\omega(I^q_i|\beta^*)=0$ for all $q\in M(I^j_i)\ne\emptyset$. Then, $u^i(\beta^*\land I^q_i)=0$ and $u^i(\varrho^i_{I^j_i}(\beta^*,\tilde\beta^*)\land I^q_i)=0$ for all $q\in M(I^j_i)$. 
Thus, {\footnotesize
\begin{equation}\label{edneB}\setlength{\abovedisplayskip}{1.2pt}
\setlength{\belowdisplayskip}{1.2pt}
\begin{array}{rl}
  & u^i(\beta^*\land I^j_i)= \sum\limits_{a\in A(I^j_i)}\beta^{*i}_{I^j_i}(a)u^i((a,\beta^{*-I^j_i})\land I^j_i)=\sum\limits_{a\in A(I^j_i)}\beta^{*i}_{I^j_i}(a)\sum\limits_{a\in h\in Z}u^i(h)\omega(h|a,\beta^{*-I^j_i}) \\
=  &
 \sum\limits_{a\in A(I^j_i)}\beta^{*i}_{I^j_i}(a)\sum\limits_{q\in M(a, I^j_i)}\sum\limits_{h\in Z}^{h\cap A(I^q_i)\ne\emptyset}u^i(h)\omega(h|a,\beta^{*-I^j_i}) +  \sum\limits_{a\in A(I^j_i)}\beta^{*i}_{I^j_i}(a)\sum\limits_{h\in Z^0(a, I^j_i)}u^i(h)\omega(h|a,\beta^{*-I^j_i})\\
 
 = & \sum\limits_{q\in M(I^j_i)}u^i(\beta^*\land I^q_i) +  \sum\limits_{a\in A(I^j_i)}\beta^{*i}_{I^j_i}(a)\sum\limits_{h\in Z^0(a, I^j_i)}u^i(h)\omega(h|a,\beta^{*-I^j_i})\\
 
 = &  \sum\limits_{a\in A(I^j_i)}\beta^{*i}_{I^j_i}(a)\sum\limits_{h\in Z^0(a, I^j_i)}u^i(h)\omega(h|a,\beta^{*-I^j_i}),
 \end{array}\end{equation}}
\noindent where the third equality comes from the perfect recall, and
{\footnotesize
\begin{equation}\label{edneB1}\setlength{\abovedisplayskip}{1.2pt}
\setlength{\belowdisplayskip}{1.2pt}
\begin{array}{rl}
  & u^i(\varrho^i_{I^j_i}(\beta^*,\tilde\beta^*)\land I^j_i)= \sum\limits_{a\in A(I^j_i)}\beta^{*i}_{I^j_i}(a)u^i((a,\varrho^i_{I^j_i}(\beta^{*-I^j_i},\tilde\beta^*))\land I^j_i)\\
  = & \sum\limits_{a\in A(I^j_i)}\beta^{*i}_{I^j_i}(a)\sum\limits_{a\in h\in Z}u^i(h)\omega(h|a,\varrho^i_{I^j_i}(\beta^{*-I^j_i},\tilde\beta^*)) \\
=  &
 \sum\limits_{a\in A(I^j_i)}\beta^{*i}_{I^j_i}(a)\sum\limits_{q\in M(a, I^j_i)}\sum\limits_{h\in Z}^{h\cap A(I^q_i)\ne\emptyset}u^i(h)\omega(h|a,\varrho^i_{I^j_i}(\beta^{*-I^j_i},\tilde\beta^*))\\
 & +  \sum\limits_{a\in A(I^j_i)}\beta^{*i}_{I^j_i}(a)\sum\limits_{h\in Z^0(a, I^j_i)}u^i(h)\omega(h|a,\varrho^i_{I^j_i}(\beta^{*-I^j_i},\tilde\beta^*))\\
 
 = & \sum\limits_{q\in M(I^j_i)}u^i(\varrho^i_{I^j_i}(\beta^*,\tilde\beta^*)\land I^q_i) +  \sum\limits_{a\in A(I^j_i)}\beta^{*i}_{I^j_i}(a)\sum\limits_{h\in Z^0(a, I^j_i)}u^i(h)\omega(h|a,\beta^{*-I^j_i})\\
 
 = &  \sum\limits_{a\in A(I^j_i)}\beta^{*i}_{I^j_i}(a)\sum\limits_{h\in Z^0(a, I^j_i)}u^i(h)\omega(h|a,\beta^{*-I^j_i}).
 \end{array}\end{equation}}

\noindent Therefore, $u^i(\beta^*\land I^j_i)=u^i(\varrho^i_{I^j_i}(\beta^*,\tilde\beta^*)\land I^j_i)$. 
 
\noindent {\bf Case (3)}. Consider $i\in N$ and $\ell\in M_i$ with $\omega(I^\ell_i|\beta^*)>0$, $M(I^\ell_i)\ne\emptyset$, and  $u^i(\beta^*\land I^j_i)=u^i(\varrho^i_{I^j_i}(\beta^*,\tilde\beta^*)\land I^j_i)$ for any $j\in M(I^\ell_i)$. We partition $M(I^\ell_i)$ into $M^1(I^\ell_i)=\{j\in M(I^\ell_i)|\omega(I^j_i|\beta^*)>0\}$ and $M^2(I^\ell_i)=\{j\in M(I^\ell_i)|\omega(I^j_i|\beta^*)=0\}$. One can see that $u^i(\beta^*\land I^j_i)=0$ and $u^i(\varrho^i_{I^\ell_i}(\beta^*,\tilde\beta^*)\land I^j_i)=0$ for all $j\in M^2(I^\ell_i)$. 
Then, {\footnotesize\begin{equation}\label{edneC}\setlength{\abovedisplayskip}{1.2pt}
\setlength{\belowdisplayskip}{1.2pt}
\begin{array}{rl}
  & u^i(\beta^*\land I^\ell_i)= \sum\limits_{a\in A(I^\ell_i)}\beta^{*i}_{I^\ell_i}(a)u^i((a,\beta^{*-I^\ell_i})\land I^\ell_i)=\sum\limits_{a\in A(I^\ell_i)}\beta^{*i}_{I^\ell_i}(a)\sum\limits_{a\in h\in Z}u^i(h)\omega(h|a,\beta^{*-I^\ell_i}) \\
=  &
 \sum\limits_{a\in A(I^\ell_i)}\beta^{*i}_{I^\ell_i}(a)\sum\limits_{j\in M(a, I^\ell_i)}\sum\limits_{h\in Z}^{h\cap A(I^j_i)\ne\emptyset}u^i(h)\omega(h|a,\beta^{*-I^\ell_i}) +  \sum\limits_{a\in A(I^\ell_i)}\beta^{*i}_{I^\ell_i}(a)\sum\limits_{h\in Z^0(a, I^\ell_i)}u^i(h)\omega(h|a,\beta^{*-I^\ell_i})\\
 
 =  &
 \sum\limits_{a\in A(I^\ell_i)}\sum\limits_{j\in M(a, I^\ell_i)}\beta^{*i}_{I^\ell_i}(a)u^i((a,\beta^{*-I^\ell_i})\land I^j_i) +  \sum\limits_{a\in A(I^\ell_i)}\beta^{*i}_{I^\ell_i}(a)\sum\limits_{h\in Z^0(a, I^\ell_i)}u^i(h)\omega(h|a,\beta^{*-I^\ell_i})\\

 = & \sum\limits_{j\in M(I^\ell_i)}u^i(\beta^*\land I^j_i) +  \sum\limits_{a\in A(I^\ell_i)}\beta^{*i}_{I^\ell_i}(a)\sum\limits_{h\in Z^0(a, I^\ell_i)}u^i(h)\omega(h|a,\beta^{*-I^\ell_i})\\
 
  = & \sum\limits_{j\in M^1(I^\ell_i)}u^i(\beta^*\land I^j_i) +  \sum\limits_{j\in M^2(I^\ell_i)}u^i(\beta^*\land I^j_i)+ \sum\limits_{a\in A(I^\ell_i)}\beta^{*i}_{I^\ell_i}(a)\sum\limits_{h\in Z^0(a, I^\ell_i)}u^i(h)\omega(h|a,\beta^{*-I^\ell_i})\\

  = &  \sum\limits_{j\in M^1(I^\ell_i)}u^i(\varrho^i_{I^j_i}(\beta^*,\tilde\beta^*)\land I^j_i) + \sum\limits_{a\in A(I^\ell_i)}\beta^{*i}_{I^\ell_i}(a)\sum\limits_{h\in Z^0(a, I^\ell_i)}u^i(h)\omega(h|a,\beta^{*-I^\ell_i}).
 \end{array}\end{equation}}
\noindent For $j\in M^1(I^\ell_i)$, it follows from Eq.~(\ref{edneeq1}) that {\footnotesize\begin{equation}\label{edneD}\setlength{\abovedisplayskip}{1.2pt}
\setlength{\belowdisplayskip}{1.2pt}
  u^i(\varrho^i_{I^j_i}(\beta^*,\tilde\beta^*)\land I^j_i)=u^i((\beta^{*i}_{I^j_i},\varrho^i_{I^j_i}(\beta^{*-I^j_i},\tilde\beta^*))\land I^j_i)
   = u^i((\tilde\beta^{*i}_{I^j_i},\varrho^i_{I^j_i}(\beta^{*-I^j_i},\tilde\beta^*))\land I^j_i)
 = u^i(\varrho^i_{I^\ell_i}(\beta^*,\tilde\beta^*)\land I^j_i).
 \end{equation}}
 \noindent Therefore, as a result of  Eq.~(\ref{edneC}) and Eq.~(\ref{edneD}), we have
{\footnotesize\begin{equation}\label{edneE}\setlength{\abovedisplayskip}{1.2pt}
\setlength{\belowdisplayskip}{1.2pt}
\begin{array}{rl}
& u^i(\beta^*\land I^\ell_i) \\
= & \sum\limits_{j\in M^1(I^\ell_i)}u^i(\varrho^i_{I^j_i}(\beta^*,\tilde\beta^*)\land I^j_i) + \sum\limits_{a\in A(I^\ell_i)}\beta^{*i}_{I^\ell_i}(a)\sum\limits_{h\in Z^0(a, I^\ell_i)}u^i(h)\omega(h|a,\beta^{*-I^\ell_i})\\
= & \sum\limits_{j\in M^1(I^\ell_i)}u^i(\varrho^i_{I^\ell_i}(\beta^*,\tilde\beta^*)\land I^j_i) + \sum\limits_{a\in A(I^\ell_i)}\beta^{*i}_{I^\ell_i}(a)\sum\limits_{h\in Z^0(a, I^\ell_i)}u^i(h)\omega(h|a,\beta^{*-I^\ell_i})\\
= & \sum\limits_{j\in M^1(I^\ell_i)}u^i(\varrho^i_{I^\ell_i}(\beta^*,\tilde\beta^*)\land I^j_i) +  \sum\limits_{j\in M^2(I^\ell_i)}u^i(\varrho^i_{I^\ell_i}(\beta^*,\tilde\beta^*)\land I^j_i) \\
& + \sum\limits_{a\in A(I^\ell_i)}\beta^{*i}_{I^\ell_i}(a)\sum\limits_{h\in Z^0(a, I^\ell_i)}u^i(h)\omega(h|a,\varrho^i_{I^\ell_i}(\beta^{*-I^\ell_i},\tilde\beta^*))\\
= & u^i(\varrho^i_{I^\ell_i}(\beta^*,\tilde\beta^*)\land I^\ell_i).\end{array}\end{equation}}

\noindent Continuing this backward induction process, one can draw the desired conclusion in Eq.~(\ref{edneA}).

We next prove that, for any $i\in N$ and $j\in M_i$ with ${\cal S}^i(I^j_i|\beta^*)>0$, \begin{equation}\label{edneF}\setlength{\abovedisplayskip}{1.2pt}
\setlength{\belowdisplayskip}{1.2pt}
\sum\limits_{a\in A(I^j_i)}\tilde \beta^{*i}_{I^j_i}(a)u^i(a,\varrho^i_{I^j_i}(\beta^{*-I^j_i}, \tilde\beta^*), \mu^*|I^j_i)=\max\limits_{\tilde\beta}\sum\limits_{a\in A(I^j_i)}\tilde \beta^i_{I^j_i}(a)u^i(a,\varrho^i_{I^j_i}(\beta^{*-I^j_i}, \tilde\beta), \mu^*|I^j_i).\end{equation}
{\bf Case (a)}. Consider $i\in N$ and $j\in M_i$ with $M(I^j_i)=\emptyset$. Then, $\varrho^i_{I^j_i}(\beta^{*-I^j_i},\tilde\beta)=\beta^{*-I^j_i}=\varrho^i_{I^j_i}(\beta^{*-I^j_i},\tilde\beta^*)$. Thus it follows from Definition~\ref{edne1} that \begin{equation}\label{edneG}\setlength{\abovedisplayskip}{1.2pt}
\setlength{\belowdisplayskip}{1.2pt}
\max\limits_{\tilde\beta}\sum\limits_{a\in A(I^j_i)}\tilde\beta^i_{I^j_i}u^i(a,\varrho(\beta^{*-I^j_i},\tilde\beta),\mu^*|I^j_i)=\sum\limits_{a\in A(I^j_i)}\tilde\beta^{*i}_{I^j_i}u^i(a,\varrho(\beta^{*-I^j_i},\tilde\beta^*),\mu^*|I^j_i)\end{equation}
since $u^i(a,\varrho(\beta^{*-I^j_i},\tilde\beta),\mu^*|I^j_i)$ is independent of $\tilde\beta$ for all $a\in A(I^j_i)$.

\noindent {\bf Case (b)}. Consider $i\in N$ and $j\in M_i$ with $M(I^j_i)\ne \emptyset$ and \begin{equation}\label{edneH}\setlength{\abovedisplayskip}{1.2pt}
\setlength{\belowdisplayskip}{1.2pt}
\max\limits_{\tilde\beta}\sum\limits_{a'\in A(I^q_i)}\tilde\beta^i_{I^q_i}(a)u^i(a',\varrho(\beta^{*-I^q_i},\tilde\beta),\mu^*|I^q_i)=\sum\limits_{a'\in A(I^q_i)}\tilde\beta^{*i}_{I^q_i}(a)u^i(a',\varrho(\beta^{*-I^q_i},\tilde\beta^*),\mu^*|I^q_i)\end{equation} for any $q\in M(I^j_i)$ with ${\cal S}^i(I^q_i|\beta^*)>0$.
For any $a\in A(I^j_i)$ and $q\in M(a, I^j_i)$, we have
{\footnotesize\begin{equation}\label{edneI}\setlength{\abovedisplayskip}{1.2pt}
\setlength{\belowdisplayskip}{1.2pt}
\begin{array}{rl}
& u^i(a,\varrho^i_{I^j_i}(\beta^{*-I^j_i},\tilde\beta),\mu^*|I^j_i)\\

= & \sum\limits_{q\in M(a, I^j_i)}\frac{{\cal S}^i(I^q_i|\beta^*)}{{\cal S}^i(I^j_i|\beta^*)}\sum\limits_{a'\in A(I^q_i)}\tilde\beta^i_{I^q_i}(a')u^i(a',\varrho^i_{I^q_i}(\beta^{*-I^q_i},\tilde\beta),\mu^*|I^q_i) +\sum\limits_{h\in Z^0(a,I^j_i)}u^i(h)\nu^i_{I^j_i}(h|a,\beta^{*-I^j_i},\mu^*).
\end{array}\end{equation}}
\noindent Then, 
{\small \begin{equation}\label{edneJ}\setlength{\abovedisplayskip}{1.2pt}
\setlength{\belowdisplayskip}{1.2pt}
\begin{array}{rl}
& \max\limits_{\tilde\beta^i}\sum\limits_{a\in A(I^j_i)}\tilde\beta^i_{I^j_i}(a)u^i(a,\varrho^i_{I^j_i}(\beta^{*-I^j_i},\tilde\beta),\mu^*|I^j_i)\\

= & \max\limits_{\tilde\beta^i_{I^j_i}}\sum\limits_{a\in A(I^j_i)}\tilde\beta^i_{I^j_i}(a) \big(\sum\limits_{q\in M(a, I^j_i)}\frac{{\cal S}^i(I^q_i|\beta^*)}{{\cal S}^i(I^j_i|\beta^*)}\max\limits_{\tilde\beta^i}\sum\limits_{a'\in A(I^q_i)}\tilde\beta^i_{I^q_i}(a')u^i(a',\varrho^i_{I^q_i}(\beta^{*-I^q_i},\tilde\beta),\mu^*|I^q_i) \\
& +\sum\limits_{h\in Z^0(a,I^j_i)}u^i(h)\nu^i_{I^j_i}(h|a,\beta^{*-I^j_i},\mu^*)\big)\\

= & \max\limits_{\tilde\beta^i_{I^j_i}}\sum\limits_{a\in A(I^j_i)}\tilde\beta^i_{I^j_i}(a) \big(\sum\limits_{q\in M(a, I^j_i)}\frac{{\cal S}^i(I^q_i|\beta^*)}{{\cal S}^i(I^j_i|\beta^*)}\sum\limits_{a'\in A(I^q_i)}\tilde\beta^{*i}_{I^q_i}(a')u^i(a',\varrho^i_{I^q_i}(\beta^{*-I^q_i},\tilde\beta^*),\mu^*|I^q_i) \\
& +\sum\limits_{h\in Z^0(a,I^j_i)}u^i(h)\nu^i_{I^j_i}(h|a,\beta^{*-I^j_i},\mu^*)\big)\\

= & \max\limits_{\tilde\beta^i_{I^j_i}}\sum\limits_{a\in A(I^j_i)}\tilde\beta^i_{I^j_i}(a) u^i(a,\varrho^i_{I^j_i}(\beta^{*-I^j_i},\tilde\beta^*),\mu^*|I^j_i)\\

= & \sum\limits_{a\in A(I^j_i)}\tilde\beta^{*i}_{I^j_i}(a) u^i(a,\varrho^i_{I^j_i}(\beta^{*-I^j_i},\tilde\beta^*),\mu^*|I^j_i),
\end{array}\end{equation}}

\noindent where the second equality comes from Eq.~(\ref{edneH}) and the fourth equality comes from Definition~\ref{edne1}.
Continuing this backward induction process, one can attain the desired result in Eq.~(\ref{edneF}).

We now are ready to prove that $u^i(\beta^*)=\max\limits_{\beta^i}u^i(\beta^i,\beta^{*-i})$. For $i\in N$ and $j\in M_i$, we denote by $d(I^j_i)$ the number of information sets of player $i$ intersecting some $h=\langle a_1,\ldots,a_L\rangle\in I^j_i$. When $I^j_i=\{\emptyset\}$,  we have $d(I^j_i)=1$. For the game in Fig.~\ref{Notation}, we have $d(I^1_2)=1$ and $d(I^2_2)=2$. Let $
Z^i=\{h\in Z|h\cap A(I^j_i)=\emptyset\text{ for all $j\in M_i$}\}$.
It follows from Eq.~(\ref{ispayoffs}) that
 \begin{equation}\label{edneK}\setlength{\abovedisplayskip}{1.2pt}
\setlength{\belowdisplayskip}{1.2pt}
u^i(\beta)=\sum\limits_{j\in M_i}^{d(I^j_i)=1}u^i(\beta\land I^j_i)+\sum\limits_{h\in Z^i}u^i(h)\omega(h|\beta).\end{equation} Then we derive from Eq.~(\ref{edneA}) that
\begin{equation}\label{edneL}\setlength{\abovedisplayskip}{1.2pt}
\setlength{\belowdisplayskip}{1.2pt}
\begin{array}{rl}
u^i(\beta^*)= & \sum\limits_{j\in M_i}^{d(I^j_i)=1}u^i(\beta^*\land I^j_i)+\sum\limits_{h\in Z^i}u^i(h)\omega(h|\beta^*)\\

= & \sum\limits_{j\in M_i}^{d(I^j_i)=1}u^i(\varrho^i_{I^j_i}(\beta^*,\tilde\beta^*)\land I^j_i)+\sum\limits_{h\in Z^i}u^i(h)\omega(h|\beta^*).
\end{array}\end{equation}
For $j\in M_i$ with $d(I^j_i)=1$, we have
\begin{equation}\label{edneM}\setlength{\abovedisplayskip}{1.2pt}
\setlength{\belowdisplayskip}{1.2pt}
\begin{array}{rl}
 u^i(\varrho^i_{I^j_i}(\beta^*,\tilde\beta^*)\land I^j_i)
= &\sum\limits_{a\in A(I^j_i)}\beta^{*i}_{I^j_i}(a)u^i((a,\varrho^i_{I^j_i}(\beta^{*-I^j_i},\tilde\beta^*))\land I^j_i)\\

= & {\cal S}^i(I^j_i|\beta^*)\sum\limits_{a\in A(I^j_i)}\beta^{*i}_{I^j_i}(a)u^i(a,\varrho^i_{I^j_i}(\beta^{*-I^j_i},\tilde\beta^*),\mu^*|I^j_i).
\end{array}\end{equation}
For $j\in M_i$ with $d(I^j_i)=1$ and ${\cal S}^i(I^j_i|\beta^*)>0$, since $\omega(I^j_i|\beta^*)={\cal S}^i(I^j_i|\beta^*)$, it follows from Definition~\ref{edne1} that \[\setlength{\abovedisplayskip}{1.2pt}
\setlength{\belowdisplayskip}{1.2pt}\sum\limits_{a\in A(I^j_i)}\beta^{*i}_{I^j_i}(a)u^i(a,\varrho^i_{I^j_i}(\beta^{*-I^j_i},\tilde\beta^*),\mu^*| I^j_i)=\sum\limits_{a\in A(I^j_i)}\tilde\beta^{*i}_{I^j_i}(a)u^i(a,\varrho^i_{I^j_i}(\beta^{*-I^j_i},\tilde\beta^*),\mu^*| I^j_i).\]  This result together with Eq.~(\ref{edneL}) leads to
\begin{equation}\label{edneN}\setlength{\abovedisplayskip}{1.2pt}
\setlength{\belowdisplayskip}{1.2pt}
u^i(\beta^*)=\sum\limits_{j\in M_i}^{d(I^j_i)=1}{\cal S}^i(I^j_i|\beta^*)\sum\limits_{a\in A(I^j_i)}\tilde\beta^{*i}_{I^j_i}(a)u^i(a,\varrho^i_{I^j_i}(\beta^{*-I^j_i},\tilde\beta^*),\mu^*|I^j_i)+\sum\limits_{h\in Z^i}u^i(h)\omega(h|\beta^*).\end{equation}
Furthermore, we have
\begin{equation}\label{edneO}\setlength{\abovedisplayskip}{1.2pt}
\setlength{\belowdisplayskip}{1.2pt}
\begin{array}{rl}  u^i(\beta^i,\beta^{*-i})
= & \sum\limits_{j\in M_i}^{d(I^j_i)=1}u^i((\beta^i,\beta^{*-i})\land I^j_i)+\sum\limits_{h\in Z^i}u^i(h)\omega(h|\beta^*)\\
= &  \sum\limits_{j\in M_i}^{d(I^j_i)=1}\sum\limits_{a\in A(I^j_i)}\beta^i_{I^j_i}(a)u^i((a,\beta^{i,-I^j_i},\beta^{*-i})\land I^j_i)+\sum\limits_{h\in Z^i}u^i(h)\omega(h|\beta^*).
\end{array}\end{equation}
For $i\in N$, $j\in M_i$ and $a\in A(I^j_i)$, let 
 \[\setlength{\abovedisplayskip}{1.2pt}
\setlength{\belowdisplayskip}{1.2pt}
\Lambda(a, I^j_i)=\left\{q\in M_i|\text{for any
 $h=\langle a_1,\ldots,a_L\rangle\in I^q_i$, there exists $1\le \ell\le L$ such that $a_\ell=a$}\right\}.\]
Because of the perfect recall, it always holds that $\Lambda(a', I^j_i)\cap\Lambda(a'', I^j_i)=\emptyset$ for any $a',a''\in A(I^j_i)$ with $a'\ne a''$.
For $i\in N$ and $j\in M_i$ with $d(I^j_i)=1$, since $\Lambda(a',I^j_i)\cap\Lambda(a'', I^j_i)=\emptyset$ for any $a',a''\in A(I^j_i)$ with $a'\ne a''$, one can deduce that $u^i((a',\beta^{i,-I^j_i},\beta^{*-i})\land I^j_i)$ is a polynomial function of  $(\beta^q_{I^q_i}: q\in\Lambda(a',I^j_i))$ and $u^i((a'',\beta^{i,-I^j_i},\beta^{*-i})\land I^j_i)$ is a polynomial function of  $(\beta^q_{I^q_i}:q\in\Lambda(a'',I^j_i))$. This result reveals that 
\begin{equation}\label{edneP}\setlength{\abovedisplayskip}{1.2pt}
\setlength{\belowdisplayskip}{1.2pt}
\max\limits_{\beta^i}\sum\limits_{a\in A(I^j_i)}\beta^i_{I^j_i}(a)u^i((a,\beta^{i,-I^j_i},\beta^{*-i})\land I^j_i)=\max\limits_{\beta^i_{I^j_i}}\sum\limits_{a\in A(I^j_i)}\beta^i_{I^j_i}(a)\max\limits_{\beta^{i,-I^j_i}} u^i((a,\beta^{i,-I^j_i},\beta^{*-i})\land I^j_i).\end{equation}
Similarly, one can derive that \begin{equation}\label{edneQ}\setlength{\abovedisplayskip}{1.2pt}
\setlength{\belowdisplayskip}{1.2pt}
\max\limits_{\tilde\beta}\sum\limits_{a\in A(I^j_i)}\tilde\beta^i_{I^j_i}(a) u^i(a,\varrho^i_{I^j_i}(\beta^*,\tilde\beta), \mu^*| I^j_i)=\max\limits_{\tilde\beta^i_{I^j_i}}\sum\limits_{a\in A(I^j_i)}\tilde\beta^i_{I^j_i}(a)\max\limits_{\tilde\beta} u^i(a,\varrho^i_{I^j_i}(\beta^*,\tilde\beta), \mu^*| I^j_i).\end{equation}
Thus,
{\footnotesize \begin{equation}\label{edneR}\setlength{\abovedisplayskip}{1.2pt}
\setlength{\belowdisplayskip}{1.2pt}
\begin{array}{rl}
 & \max\limits_{\beta^i}u^i(\beta^i,\beta^{*-i}) \\
= & \sum\limits_{j\in M_i}^{d(I^j_i)=1}\max\limits_{\beta^i}\sum\limits_{a\in A(I^j_i)}\beta^i_{I^j_i}(a)u^i((a,\beta^{i,-I^j_i},\beta^{*-i})\land I^j_i)+\sum\limits_{h\in Z^i}u^i(h)\omega(h|\beta^*)\\
= & \sum\limits_{j\in M_i}^{d(I^j_i)=1}\max\limits_{\beta^i_{I^j_i}}\sum\limits_{a\in A(I^j_i)}\beta^i_{I^j_i}(a)\max\limits_{\beta^{i,-I^j_i}} u^i((a,\beta^{i,-I^j_i},\beta^{*-i})\land I^j_i)+\sum\limits_{h\in Z^i}u^i(h)\omega(h|\beta^*)\\
= & \sum\limits_{j\in M_i}^{d(I^j_i)=1}\max\limits_{\beta^i_{I^j_i}}\sum\limits_{a\in A(I^j_i)}\beta^i_{I^j_i}(a)\max\limits_{\tilde\beta^{i,-I^j_i}} u^i((a,\tilde\beta^{i,-I^j_i},\beta^{*-i})\land I^j_i)+\sum\limits_{h\in Z^i}u^i(h)\omega(h|\beta^*)\\
= & \sum\limits_{j\in M_i}^{d(I^j_i)=1}\max\limits_{\beta^i_{I^j_i}}\sum\limits_{a\in A(I^j_i)}\beta^i_{I^j_i}(a)\max\limits_{\tilde\beta} u^i((a,\varrho^i_{I^j_i}(\beta^*,\tilde\beta))\land I^j_i)+\sum\limits_{h\in Z^i}u^i(h)\omega(h|\beta^*)\\
= & \sum\limits_{j\in M_i}^{d(I^j_i)=1}{\cal S}^i(I^j_i|\beta^*)\max\limits_{\tilde\beta^i_{I^j_i}}\sum\limits_{a\in A(I^j_i)}\tilde\beta^i_{I^j_i}(a)\max\limits_{\tilde\beta} u^i(a,\varrho^i_{I^j_i}(\beta^*,\tilde\beta), \mu^*| I^j_i)+\sum\limits_{h\in Z^i}u^i(h)\omega(h|\beta^*)\\
= & \sum\limits_{j\in M_i}^{d(I^j_i)=1}{\cal S}^i(I^j_i|\beta^*)\max\limits_{\tilde\beta}\sum\limits_{a\in A(I^j_i)}\tilde\beta^i_{I^j_i}(a) u^i(a,\varrho^i_{I^j_i}(\beta^*,\tilde\beta), \mu^*| I^j_i)+\sum\limits_{h\in Z^i}u^i(h)\omega(h|\beta^*)\\
= & \sum\limits_{j\in M_i}^{d(I^j_i)=1}{\cal S}^i(I^j_i|\beta^*)\sum\limits_{a\in A(I^j_i)}\tilde\beta^{*i}_{I^j_i}(a) u^i(a,\varrho^i_{I^j_i}(\beta^*,\tilde\beta^*), \mu^*| I^j_i)+\sum\limits_{h\in Z^i}u^i(h)\omega(h|\beta^*)\\
= & u^i(\beta^*),
\end{array}\end{equation}
}where the first equality is derived from Eq.~(\ref{edneO}), the second equality is derived from Eq.~(\ref{edneP}), the sixth equality is derived from Eq.~(\ref{edneQ}), the seventh equality is derived from Eq.~(\ref{edneF}), and the last equality is derived from Eq.~(\ref{edneN}). 
Consequently, $\beta^*$ is a NashEBS according to Definition~\ref{ned1}.

($\Leftarrow$). We next show that Definition~\ref{ned1} implies Definition~\ref{edne1}.  Let $\beta^*$ be a Nash equilibrium according to Definition~\ref{ned1}.
 We denote by $\mu^*$ a solution to the system~(\ref{bsne1}).
 Given $\mu^*$, let $\tilde\beta^*=(\tilde\beta^{*i}_{I^j_i}:i\in N,j\in M_i)$ be a behavioral strategy profile such that, for all $i\in N$ and $j\in M_i$ with $\omega(I^j_i|\beta^*)>0$, we take $\tilde\beta^{*i}_{I^j_i}=\beta^{*i}_{I^j_i}$, and for all $i\in N$ and $j\in M_i$ with $\omega(I^j_i|\beta^*)=0$, we solve through the backward induction
 a sequence of the following linear optimization problems to find $\tilde\beta^{*i}_{I^j_i}$, \begin{equation}\setlength{\abovedisplayskip}{1.2pt}
\setlength{\belowdisplayskip}{1.2pt}
\label{ednelp1}\begin{array}{rl}
\max\limits_{\tilde\beta^i_{I^j_i}} &\sum\limits_{a\in A(I^j_i)}\tilde\beta^i_{I^j_i}(a)u^i(a,\varrho(\beta^{*-I^j_i},\tilde\beta^*),\mu^*| I^j_i)\\
\text{s.t.} & \sum\limits_{a\in A(I^j_i)}\tilde\beta^i_{I^j_i}(a)=1,\;0\le\tilde\beta^i_{I^j_i}(a),\;a\in A(I^j_i).
\end{array}\end{equation}
We next prove that  $(\beta^*,\tilde\beta^*,\mu^*)$ meets the properties of Definition~\ref{edne1}. For any $i\in N$ and $j\in M_i$ with $\omega(I^j_i|\beta^*)=0$, we have $u^i(\beta^*\land I^j_i)=u^i(\varrho^i_{I^j_i}(\beta^*,\tilde\beta^*)\land I^j_i)=0$.
Then it suffices to consider only those information sets reachable  by $
\beta^*$. 
Let $I^j_i$ be an information set with $\omega(I^j_i|\beta^*)>0$. Thus, for any $q\in M(I^j_i)$ with $\omega(I^q_i|\beta^*)=0$, we have $u^i(\varrho^i_{I^j_i}(\beta^*,\tilde\beta^*)\land I^q_i)=u^i(\beta^*\land I^q_i)=0$, and for any $q\in M(I^j_i)$ with $\omega(I^q_i|\beta^*)>0$, we have $\tilde\beta^{*i}_{I^q_i}=\beta^{*i}_{I^q_i}$ by the selection of $\tilde\beta^*$. Therefore, as a result of Lemma~\ref{ednelm1}, one can draw the conclusion that
 $u^i(\varrho^i_{I^j_i}(\beta^*,\tilde\beta^*)\land I^j_i)=u^i(\beta^*\land I^j_i)$. Hence it follows from the condition of $u^i(\beta^*)\ge u^i(\beta^i,\beta^{*-i})$ for any $\beta^i$ in Definition~\ref{ned1} that $\beta^{*i}_{I^j_i}(a')=0$ for any $i\in N$, $j\in M_i$ and $a',a''\in A(I^j_i)$ with $u^i((a'', \varrho^i_{I^j_i}(\beta^{*-I^j_i},\tilde\beta^*))\land I^j_i)>u^i((a', \varrho^i_{I^j_i}(\beta^{*-I^j_i},\tilde\beta^*))\land I^j_i)$.
 
 Consider $i\in N$ and $j\in M_i$ with $\omega(I^j_i|\beta^*)>0$.
 We have $\mu^{*i}_{I^j_i}(h)=\frac{{\cal S}^i(h|\beta^*)}{{\cal S}^i(I^j_i|\beta^*)}=\frac{\omega(h|\beta^*)}{\omega(I^j_i|\beta^*)}$ for $h\in I^j_i$.  
Then, $u^i(\varrho^i_{I^j_i}(\beta^*,\tilde\beta^*), \mu^*|I^j_i)=\frac{1}{\omega(I^j_i|\beta^*)} u^i(\varrho^i_{I^j_i}(\beta^*,\tilde\beta^*)\land I^j_i)=\frac{1}{\omega(I^j_i|\beta^*)}u^i(\beta^*\land I^j_i)$. Thus, as result of $u^i(\beta^*)\ge u^i(\beta^i,\beta^{*-i})$ for any $\beta^i$ in Definition~\ref{ned1}, we acquire from $\tilde\beta^{*i}_{I^j_i}=\beta^{*i}_{I^j_i}$ that $\tilde\beta^{*i}_{I^j_i}(a')=0$ for any $i\in N$, $j\in M_i$ and $a',a''\in A(I^j_i)$ with $u^i(a'', \varrho^i_{I^j_i}(\beta^{*-I^j_i},\tilde\beta^*),\mu^*|I^j_i)>u^i(a', \varrho^i_{I^j_i}(\beta^{*-I^j_i},\tilde\beta^*),\mu^*|I^j_i)$.
Consider $i\in N$ and $j\in M_i$ with $\omega(I^j_i|\beta^*)=0$. It follows from the linear optimization problem~(\ref{ednelp1}) that $\tilde\beta^{*i}_{I^j_i}(a')=0$ for any $i\in N$, $j\in M_i$ and $a',a''\in A(I^j_i)$ with $u^i(a'', \varrho^i_{I^j_i}(\beta^{*-I^j_i},\tilde\beta^*),\mu^*|I^j_i)>u^i(a', \varrho^i_{I^j_i}(\beta^{*-I^j_i},\tilde\beta^*),\mu^*|I^j_i)$.

 These results together confirm that  $(\beta^*,\tilde\beta^*,\mu^*)$ meets the properties of Definition~\ref{edne1}. The proof is completed.}
\end{proof} 

We denote by $\text{int}(C)$ and $|C|$ the  interior of a set $C$ and the cardinality of a finite set $C$, respectively. Let
$\triangle=\mathop{\times}\limits_{i\in N,\;j\in M_i}\triangle^i_{I^j_i}$ and $\triangle^i=\mathop{\times}\limits_{j\in M_i}\triangle^i_{I^j_i}$, where $\triangle^i_{I^j_i}=\{\beta^i_{I^j_i}\in\mathbb{R}_+^{|A(I^j_i)|}|\sum\limits_{a\in A(I^j_i)}\beta^i_{I^j_i}(a)=1\}$. Let $\Xi=\mathop{\times}\limits_{i\in N,\;j\in M_i}\Xi^i_{I^j_i}$, where $\Xi^i_{I^j_i}=\{\mu^i_{I^j_i}=(\mu^i_{I^j_i}(h):h\in I^j_i)^{\top}|\sum\limits_{h\in I^j_i}\mu^i_{I^j_i}(h)=1,\;0\le\mu^i_{I^j_i}(h)\}$. 
To prove the existence of a NashEBS, a general approach is to leverage the well-known Brouwer or Kakutani fixed point theorem. However,  according to Definition~\ref{ned1},
 it would be difficult for one to establish with the general approach the existence of a NashEBS for an extensive-form game with perfect recall. The reason is as follows: Define the best response correspondence $B:\triangle\to\triangle$ such that, for all $\beta\in\triangle$, we have $B(\beta)=\mathop{\times}\limits_{i\in N} B_i(\beta)$, where $B_i(\beta)=\text{arg}\max\limits_{\hat\beta^i\in\triangle^i}u^i(\hat\beta^i,\beta^{-i})$. Clearly, $B(\beta)$ is not a convex-valued correspondence due to the nonconvexity of $u^i(\hat\beta^i,\beta^{-i})$ in $\hat\beta^i$. Thus, the conditions of Kakutani's fixed point theorem are not satisfied, and as a result, the existence of a NashEBS cannot be directly established using Definition~\ref{ned1}. Therefore one had to exploit the associated normal-form game of an extensive-form game for the existence of a NashEBS.
 We next demonstrate that our characterization allows us to establish the existence of a NashEBS for an extensive-form game directly by leveraging the extensive-form game structure, without relying on the associated normal-form game.

\begin{theorem}\label{nethm1}{\em There always exists a NashEBS for every finite extensive-form game with perfect recall.}
\end{theorem}
\begin{proof} {\small For $(\beta',\tilde\beta',\mu')\in\triangle\times\triangle\times\Xi$, let $F(\beta',\tilde\beta', \mu')=\mathop{\times}\limits_{i\in  N,\;j\in M_i}F^i_{I^j_i}(\beta',\tilde\beta',\mu')$ with
$F^i_{I^j_i}(\beta',\tilde\beta',\mu')=G^{1i}_{I^j_i}(\beta',\tilde\beta', \mu')\times G^{2i}_{I^j_i}(\beta',\tilde\beta', \mu')\times G^{3i}_{I^j_i}(\beta',\tilde\beta', \mu')$,
where, if $\omega(I^j_i|\beta')=0$, then $G^{1i}_{I^j_i}(\beta',\tilde\beta', \mu')=\triangle^i_{I^j_i}$;  if  $\omega(I^j_i|\beta')>0$, then $G^{1i}_{I^j_i}(\beta',\tilde\beta', \mu')$ is the set of solutions to the linear optimization problem, 
\begin{equation}\label{nethm1op1} \setlength{\abovedisplayskip}{1.2pt}
\setlength{\belowdisplayskip}{1.2pt}
\begin{array}{rl}
\max\limits_{\beta^i_{I^j_i}} & \sum\limits_{a\in A(I^j_i)}\beta^i_{I^j_i}(a)u^i((a,\varrho^i_{I^j_i}(\beta',\tilde\beta'))\land I^j_i)\\
\text{s.t.} & \sum\limits_{a\in A(I^j_i)}\beta^i_{I^j_i}(a)=1,\;0\le\beta^i_{I^j_i},\;a\in A(I^j_i);
\end{array}
\end{equation}
$G^{2i}_{I^j_i}(\beta',\tilde\beta', \mu')$ is the set of solutions to the linear optimization problem, 
\begin{equation}\label{nethm1op2}\setlength{\abovedisplayskip}{1.2pt}
\setlength{\belowdisplayskip}{1.2pt}
\begin{array}{rl}
\max\limits_{\beta^i_{I^j_i}} & \sum\limits_{a\in A(I^j_i)}\beta^i_{I^j_i}(a)u^i(a,\varrho^i_{I^j_i}(\beta',\tilde\beta'), \mu'| I^j_i)\\
\text{s.t.} & \sum\limits_{a\in A(I^j_i)}\beta^i_{I^j_i}(a)=1,\;0\le\beta^i_{I^j_i},\;a\in A(I^j_i);
\end{array}
\end{equation}
if 
${\cal S}^i(I^j_i|\beta')=0$, then $G^{3i}_{I^j_i}(\beta',\tilde\beta', \mu')=\Xi^i_{I^j_i}$; and if ${\cal S}^i(I^j_i|\beta')>0$, then $G^{3i}_{I^j_i}(\beta',\tilde\beta', \mu')=(\mu^i_{I^j_i}(h):h\in I^j_i)^\top$ with $\mu^i_{I^j_i}(h)=\frac{{\cal S}^i(h|\beta')}{{\cal S}^i(I^j_i|\beta')}$. Since the problem~(\ref{nethm1op1}) and the problem~(\ref{nethm1op2}) are convex optimization problems, $F(\beta,\tilde\beta,\mu)$ is a nonempty convex and compact set for any $(\beta,\tilde\beta,\mu)\in \triangle\times\triangle\times\Xi$. As a result of continuity of $u^i((a,\varrho^i_{I^j_i}(\beta,\tilde\beta))\land I^j_i)$, $u^i(a,\varrho^i_{I^j_i}(\beta,\tilde\beta), \mu| I^j_i)$, $\omega(h|\beta)$, and ${\cal S}^i(h|\beta)$ on $\triangle\times\triangle\times\Xi$, we get that  $F(\beta,\tilde\beta,\mu)$ is a semi-continuous mapping on $\triangle\times\triangle\times\Xi$. Consequently, it follows from Kakutani's fixed point theorem that
there exists $(\beta^*,\tilde\beta^*,\mu^*)\in \triangle\times\triangle\times\Xi$ such that
$(\beta^*,\tilde\beta^*,\mu^*)\in F(\beta^*,\tilde\beta^*,\mu^*)$, which meets the requirements of Definition~\ref{edne1}.   The proof is completed.}
\end{proof}
As a direct application of Definition~\ref{edne1}, we secure a polynomial system as a necessary and sufficient condition for determining whether a behavioral strategy profile is a NashEBS, which can be exploited in the development of a general method for computing such an equilibrium. 
\begin{theorem}\label{nscthm1} {\em $\beta^*$ is a Nash equilibrium in behavioral strategies if and only if there exist $(\tilde\beta^*,\mu^*,\lambda^*,\tilde\lambda^*,\zeta^*,\tilde\zeta^*)$ together with $\beta^*$
satisfying the system,
\begin{equation}\setlength{\abovedisplayskip}{1.2pt}
\setlength{\belowdisplayskip}{1.2pt}
\label{nsrne1}
\begin{array}{l}
u^i((a,\varrho^i_{I^j_i}(\beta^{-I^j_i},\tilde\beta))\land I^j_i) +\lambda^i_{I^j_i}(a)-\zeta^i_{I^j_i}=0,\;i\in N,j\in M_i,a\in A(I^j_i),\\
u^i(a,\varrho^i_{I^j_i}(\beta^{-I^j_i},\tilde\beta),\mu|I^j_i) +\tilde\lambda^i_{I^j_i}(a)-\tilde\zeta^i_{I^j_i}=0,\;i\in N,j\in M_i,a\in A(I^j_i),\\
{\cal S}^i(I^j_i|\beta)\mu^i_{I^j_i}(h)={\cal S}^i(h|\beta),\;0\le \mu^i_{I^j_i}(h),\;i\in N,j\in M_i,h\in I^j_i,\\
\sum\limits_{h\in I^j_i}\mu^i_{I^j_i}(h)=1,\;\sum\limits_{a\in A(I^j_i)}\beta^i_{I^j_i}(a)=1,\;\sum\limits_{a\in A(I^j_i)}\tilde\beta^i_{I^j_i}(a)=1,\;i\in N,j\in M_i,\\
\beta^i_{I^j_i}(a)\lambda^i_{I^j_i}(a)=0,\;\tilde\beta^i_{I^j_i}(a)\tilde\lambda^i_{I^j_i}(a)=0,\\
0\le\beta^i_{I^j_i}(a),\;0\le\lambda^i_{I^j_i}(a),\;0\le \tilde\beta^i_{I^j_i}(a),\;0\le\tilde\lambda^i_{I^j_i}(a),\;i\in N,j\in M_i,a\in A(I^j_i).
\end{array}
\end{equation}
}
\end{theorem}
\begin{proof} {\small ($\Rightarrow$). We denote by $(\beta^*,\tilde\beta^*,\mu^*)$ a triple satisfying the properties in Definition~\ref{edne1}. Let $\zeta^{*i}_{I^j_i}=\max\limits_{a\in A(I^j_i)}u^i((a,\varrho^i_{I^j_i}(\beta^{*-I^j_i},\tilde\beta^*))\land I^j_i)$, $\tilde\zeta^{*i}_{I^j_i}=\max\limits_{a\in A(I^j_i)}u^i(a,\varrho^i_{I^j_i}(\beta^{*-I^j_i},\tilde\beta^*),\mu^*| I^j_i)$, 
$\lambda^{*i}_{I^j_i}(a)= \zeta^{*i}_{I^j_i}-u^i((a,\varrho^i_{I^j_i}(\beta^{*-I^j_i},\tilde\beta^*))\land I^j_i)$, and $\tilde\lambda^{*i}_{I^j_i}(a)= \tilde\zeta^{*i}_{I^j_i}-u^i(a,\varrho^i_{I^j_i}(\beta^{*-I^j_i},\tilde\beta^*),\mu^*| I^j_i)$. Suppose that there exists some $a'\in A(I^j_i)$ such that $\lambda^{*i}_{I^j_i}(a')>0$. Then there exists $a''\in A(I^j_i)$ such that  $u^i((a'',\varrho^i_{I^j_i}(\beta^{*-I^j_i},\tilde\beta^*))\land I^j_i)>u^i((a',\varrho^i_{I^j_i}(\beta^{*-I^j_i},\tilde\beta^*))\land I^j_i)$. Thus it follows from Definition~\ref{edne1} that $\beta^{*i}_{I^j_i}(a')=0$.  One can show in a similar way that $\tilde\beta^{*i}_{I^j_i}(a')=0$ if  $\tilde\lambda^{*i}_{I^j_i}(a')>0$. Therefore, $(\beta^*,\tilde\beta^*,\mu^*,\zeta^*,\tilde\zeta^*,\lambda^*,\tilde\lambda^*)$ satisfies the system~(\ref{nsrne1}).

($\Leftarrow$). Let $(\beta^*,\tilde\beta^*,\mu^*,\lambda^*,\tilde\lambda^*,\zeta^*,\tilde\zeta^*)$ be a solution to the system~(\ref{nsrne1}). Multiplying $\beta^{*i}_{I^j_i}(a)$ to the first group of equations and $\tilde\beta^{*i}_{I^j_i}$ to the second group of equations in the system~(\ref{nsrne1}) and taking the sum over $A(I^j_i)$, we get 
$\zeta^{*i}_{I^j_i}=\sum\limits_{a\in A(I^j_i)}\beta^{*i}_{I^j_i}(a)u^i((a,\varrho^i_{I^j_i}(\beta^{*-I^j_i},\tilde\beta^*))\land I^j_i)$ and $\tilde\zeta^{*i}_{I^j_i}=\sum\limits_{a\in A(I^j_i)}\tilde\beta^{*i}_{I^j_i}(a)u^i(a,\varrho^i_{I^j_i}(\beta^{*-I^j_i}$, $\tilde\beta^*),\mu^*|I^j_i)$. Then,
as a result of $\lambda^*\ge 0$, we have $\zeta^{*i}_{I^j_i}=\max\limits_{a\in A(I^j_i)}u^i((a,\varrho^i_{I^j_i}(\beta^{*-I^j_i},\tilde\beta^*))\land I^j_i)$ and $\tilde\zeta^{*i}_{I^j_i}=\max\limits_{a\in A(I^j_i)}u^i(a,\varrho^i_{I^j_i}(\beta^{*-I^j_i},\tilde\beta^*),\mu^*|I^j_i)$. Thus, $\beta^{*i}_{I^j_i}(a')=0$ whenever $u^i((a'',\varrho^i_{I^j_i}(\beta^{*-I^j_i},\tilde\beta^*))\land I^j_i)>u^i((a'$, $\varrho^i_{I^j_i}(\beta^{*-I^j_i},\tilde\beta^*))\land I^j_i)$
since $\lambda^{*i}_{I^j_i}(a')>0$, and  $\tilde\beta^{*i}_{I^j_i}(a')=0$ whenever $u^i(a'',\varrho^i_{I^j_i}(\beta^{*-I^j_i},\tilde\beta^*),\mu^*|I^j_i)>u^i(a',\varrho^i_{I^j_i}(\beta^{*-I^j_i}$, $\tilde\beta^*),\mu^*|I^j_i)$ since $\tilde\lambda^{*i}_{I^j_i}(a')>0$.  Therefore, $\beta^*$ is a NashEBS according to Definition~\ref{edne1}. The proof is completed.}
\end{proof}

It is evident from Definition~\ref{edne1} that for any $i\in N$ and $j\in M_i$ with $\omega(I^j_i|\beta^*)>0$,  we have ${\cal S}^i(I^j_i|\beta^*)>0$ and consequently, $u^i((a, \varrho^i_{I^j_i}(\beta^{*-I^j_i},\tilde\beta^*))\land I^j_i)=\omega(I^j_i|\beta^*)u^i(a, \varrho^i_{I^j_i}(\beta^{*-I^j_i}$, $\tilde\beta^*), \mu^*| I^j_i)$ since $\mu^i_{I^j_i}(h)=\frac{{\cal S}^i(h|\beta^*)}{{\cal S}^i(I^j_i|\beta^*)}=\frac{\omega(h|\beta^*)}{\omega(I^j_i|\beta^*)}$ for $h\in I^j_i$. This result implies that Definition~\ref{edne1} 
 can be equivalently rewritten as follows.
  
\begin{definition}[\bf An Equivalent Definition of NashEBS  under Conditional Expected Payoffs]\label{edne2}{\em A behavioral strategy profile $\beta^*$ is a Nash equilibrium if $\beta^*$ together with $(\tilde\beta^*,\mu^*)$ satisfies  the properties that: \\
(i) $\beta^{*i}_{I^j_i}(a')=0$ for any $i\in N$, $j\in M_i$ and $a',a''\in A(I^j_i)$ with $\omega(I^j_i|\beta^*)>0$ and $u^i(a'',\varrho^i_{I^j_i}(\beta^{*-I^j_i},\tilde\beta^*),\mu^*|I^j_i)
>u^i(a',\varrho^i_{I^j_i}(\beta^{*-I^j_i},\tilde\beta^*),\mu^*|I^j_i)$;\\
 (ii)
 $\tilde\beta^{*i}_{I^j_i}(a')=0$ for any $i\in N$, $j\in M_i$ and $a',a''\in A(I^j_i)$ with $u^i(a'',\varrho^i_{I^j_i}(\beta^{*-I^j_i},\tilde\beta^*),\mu^*|I^j_i)
>u^i(a',\varrho^i_{I^j_i}(\beta^{*-I^j_i},\tilde\beta^*),\mu^*|I^j_i)$; \\
(iii)
 $\mu^*=(\mu^{*i}_{I^j_i}(h):i\in N,j\in M_i,h\in I^j_i)$ is a solution to the system~(\ref{bsne1}). 
}\end{definition}
Comparing Definition~\ref{edne2} with Definition~\ref{edne1}, one can see that Definition~\ref{edne2} requires only the computation of $u^i(a, \varrho^i_{I^j_i}(\beta^{-I^j_i},\tilde\beta), \mu|I^j_i)$ whereas Definition~\ref{edne1} necessitates the computation of both $u^i((a, \varrho^i_{I^j_i}(\beta^{-I^j_i},\tilde\beta))\land I^j_i)$ and $u^i(a, \varrho^i_{I^j_i}(\beta^{-I^j_i},\tilde\beta), \mu|I^j_i)$.
As a corollary of Theorem~\ref{nscthm1}, we come to the following conclusion.
\begin{corollary}\label{nscco1} {\em $\beta^*$ is a Nash equilibrium in behavioral strategies if and only if there exist $(\tilde\beta^*,\mu^*,\lambda^*,\tilde\lambda^*,\zeta^*,\tilde\zeta^*)$ together with $\beta^*$
satisfying the system,
\begin{equation}\setlength{\abovedisplayskip}{1.2pt}
\setlength{\belowdisplayskip}{1.2pt}
\label{nsrne1a}
\begin{array}{l}
\omega(I^j_i|\beta)u^i(a,\varrho^i_{I^j_i}(\beta^{-I^j_i},\tilde\beta),\mu|I^j_i) +\lambda^i_{I^j_i}(a)-\zeta^i_{I^j_i}=0,\;i\in N,j\in M_i,a\in A(I^j_i),\\
u^i(a,\varrho^i_{I^j_i}(\beta^{-I^j_i},\tilde\beta),\mu|I^j_i) +\tilde\lambda^i_{I^j_i}(a)-\tilde\zeta^i_{I^j_i}=0,\;i\in N,j\in M_i,a\in A(I^j_i),\\
{\cal S}^i(I^j_i|\beta)\mu^i_{I^j_i}(h)={\cal S}^i(h|\beta),\;0\le \mu^i_{I^j_i}(h),\;i\in N,j\in M_i,h\in I^j_i,\\
\sum\limits_{h\in I^j_i}\mu^i_{I^j_i}(h)=1,\;\sum\limits_{a\in A(I^j_i)}\beta^i_{I^j_i}(a)=1,\;\sum\limits_{a\in A(I^j_i)}\tilde\beta^i_{I^j_i}(a)=1,\;i\in N,j\in M_i,\\
\beta^i_{I^j_i}(a)\lambda^i_{I^j_i}(a)=0,\;\tilde\beta^i_{I^j_i}(a)\tilde\lambda^i_{I^j_i}(a)=0,\\
0\le\beta^i_{I^j_i}(a),\;0\le\lambda^i_{I^j_i}(a),\;0\le \tilde\beta^i_{I^j_i}(a),\;0\le\tilde\lambda^i_{I^j_i}(a),\;i\in N,j\in M_i,a\in A(I^j_i).
\end{array}
\end{equation}
}
\end{corollary}

To further boost the applications of NashEBS, we will utilize Theorem~\ref{nscthm1} or Corollary~\ref{nscco1} to develop differentiable path-following methods to compute such an equilibrium in Section~\ref{dpm}. 

\section{\large A Characterization of Subgame Perfect Equilibrium in Behavioral Strategies}

As a strict refinement of Nash equilibrium in behavioral strategies, Selten~\cite{Selten (1965)} introduced the notion of subgame perfect equilibrium in behavioral strategies.
\begin{definition}[{\bf Subgame Perfect Equilibrium in Behavioral Strategies (SGPEBS)}, Selten~\cite{Selten (1965)}]\label{sgped1} {\em A behavioral strategy profile $\beta^*$ is 
a subgame perfect equilibrium of an extensive-form game if its restriction on every subgame remains to be a Nash equilibrium of the subgame. 
}\end{definition}
From Definition~\ref{sgped1}, a characterization of subgame perfect equilibrium, based on the principles of local sequential rationality and self-independent beliefs, can be obtained by restricting our characterization of Nash equilibrium to every subgame.
For $h=\langle a_1,\ldots,a_L\rangle\in H$, let 
\begin{equation}\setlength{\abovedisplayskip}{1.2pt}
\setlength{\belowdisplayskip}{1.2pt}{\cal C}^i_{I^j_i}(h|\beta)=\prod\limits_{k=g(h)}^{L-1}
\beta^{P(\langle a_1,\ldots,a_k\rangle)}_{\langle a_1,\ldots,a_k\rangle}(a_{k+1}),\end{equation}
where $g(h)$ is the smallest index such that the subhistory $\hat h=\langle a_{g(h)+1},\ldots,a_L\rangle$ of $h$ entirely belongs to the smallest subgame containing $A(I^j_i)$.  For $i\in N$, $j\in M_i$, $h=\langle a_1,\ldots,a_L\rangle\in H$, and $a\in A(I^j_i)$,  we have \({\cal C}^i_{I^j_i}(h|a,\beta^{-I^j_i})=\mathop{\prod\limits_{k=g(h)}^{L-1}}\limits_{\langle a_1,\ldots,a_k\rangle\notin I^j_i}\beta^{P(\langle a_1,\ldots,a_k\rangle)}_{\langle a_1,\ldots,a_k\rangle}(a_{k+1}).\) 
For $i\in N$ and $j\in M_i$,
let {\small\begin{equation}\label{sgispayoffs} u^i(\beta\Diamond I^j_i)=\sum\limits^
{h\cap A(I^j_i)\ne\emptyset}_{h\in Z}u^i(h){\cal C}^i_{I^j_i}(h|\beta)\text{
and  }
u^i((a,\beta^{-I^j_i})\Diamond I^j_i)=\sum\limits_{a\in h\in Z}u^i(h){\cal C}^i_{I^j_i}(h|a,\beta^{-I^j_i}).\end{equation}}

\noindent For $h=\langle a_1,\ldots,a_L\rangle\in H$, let {\small
\begin{equation}\setlength{\abovedisplayskip}{1.2pt}
\setlength{\belowdisplayskip}{1.2pt}
\label{edsgpeeq1}{\cal Y}^i(h|\beta)=\mathop{\prod\limits_{k=q(h)}^{L-1}}\limits_{P(\langle a_1,\ldots,a_k\rangle)\ne i}
\beta^{P(\langle a_1,\ldots,a_k\rangle)}_{\langle a_1,\ldots,a_k\rangle}(a_{k+1}),\end{equation}} 

\noindent where $q(h)$ is the smallest index such that the subhistory $\hat h=\langle a_{q(h)+1},\ldots,a_L\rangle$ of $h$ entirely belongs to the smallest subgame containing $a_L$. For $i\in N$ and $j\in M_i$, we have ${\cal Y}^i(I^j_i|\beta)=\sum\limits_{h\in I^j_i}{\cal Y}^i(h|\beta)$. Clearly, ${\cal S}^i(h|\beta)={\cal Y}^i(h|\beta)$ if $\Gamma$ contains no subgame. 
\begin{definition}[\bf A Characterization of Subgame Perfect Equilibrium]\label{edsgpe1}{\em A behavioral strategy profile $\beta^*$ is a subgame perfect equilibrium if $\beta^*$ satisfies together with $(\tilde\beta^*,\mu^*)$ the properties: \newline
\noindent (i) $\beta^{*i}_{I^j_i}(a')=0$ for any $i\in N$, $j\in M_i$ and $a',a''\in A(I^j_i)$ with
$u^i((a'',\varrho^i_{I^j_i}(\beta^{*-I^j_i},\tilde\beta^*))\Diamond I^j_i)>u^i((a',\varrho^i_{I^j_i}(\beta^{*-I^j_i},\tilde\beta^*))\Diamond I^j_i)$;\newline 
\noindent (ii) $\tilde\beta^{*i}_{I^j_i}(a')=0$ for any $i\in N$, $j\in M_i$ and $a',a''\in A(I^j_i)$ with
$u^i(a'',\varrho^i_{I^j_i}(\beta^{*-I^j_i},\tilde\beta^*),\mu^*|I^j_i)> u^i(a',\varrho^i_{I^j_i}(\beta^{*-I^j_i},\tilde\beta^*),\mu^*|I^j_i)$; and\newline 
\noindent (iii) $\mu^*=(\mu^{*i}_{I^j_i}(h):i\in N,j\in M_i,h\in I^j_i)$ is a solution to the system,  
{\small\begin{equation}\setlength{\abovedisplayskip}{1.2pt}
\setlength{\belowdisplayskip}{1.2pt}
\label{sgbseq1}\begin{array}{l}
{\cal Y}^i(I^j_i|\beta^*)\mu^i_{I^j_i}(h)={\cal Y}^i(h|\beta^*),\;i\in N,j\in M_i,h\in I^j_i,\\
\sum\limits_{h'\in I^j_i}\mu^i_{I^j_i}(h')=1,\;0\le \mu^i_{I^j_i}(h),\;i\in N,j\in M_i,h\in I^j_i.
\end{array}\end{equation}}
}\end{definition}
We will illustrate with one example how one can employ Definition~\ref{edsgpe1} to analytically determine all subgame perfect equilibria for small extensive-form games in Section~\ref{examples}.  

One can prove in a similar way to Theorem~\ref{ednethm1} the following conclusion.
\begin{theorem}{\em Definition~\ref{edsgpe1} and Definition~\ref{sgped1} of subgame perfect equilibrium are equivalent.}
\end{theorem}
As a direct result of Definition~\ref{edsgpe1}, we acquire a polynomial system as
a necessary and sufficient condition for a subgame perfect equilibrium, which can be leveraged in the development of a general method for computing such an  equilibrium.
\begin{theorem}\label{sgpethm1}
{\em $\beta^*$ is a subgame perfect equilibrium if and only if there exists $(\tilde\beta^*,\mu^*,\lambda^*,\tilde\lambda^*,\zeta^*,\tilde\zeta^*)$ together with $\beta^*$ satisfying the system, {\small\begin{equation}\setlength{\abovedisplayskip}{1.2pt}
\setlength{\belowdisplayskip}{1.2pt}
\label{nscsgpe1}
\begin{array}{l}
u^i((a,\varrho^i_{I^j_i}(\beta^{-I^j_i},\tilde\beta))\Diamond I^j_i) +\lambda^i_{I^j_i}(a)-\zeta^i_{I^j_i}=0,\;i\in N,j\in M_i,a\in A(I^j_i),\\
u^i(a,\varrho^i_{I^j_i}(\beta^{-I^j_i},\tilde\beta),\mu|I^j_i) +\tilde\lambda^i_{I^j_i}(a)-\tilde\zeta^i_{I^j_i}=0,\;i\in N,j\in M_i,a\in A(I^j_i),\\
{\cal Y}^i(I^j_i|\beta)\mu^i_{I^j_i}(h)={\cal Y}^i(h|\beta),\;0\le \mu^i_{I^j_i}(h),\;i\in N,j\in M_i,h\in I^j_i,\\
\sum\limits_{h\in I^j_i}\mu^i_{I^j_i}(h)=1,\;\sum\limits_{a\in A(I^j_i)}\beta^i_{I^j_i}(a)=1,\;\sum\limits_{a\in A(I^j_i)}\tilde\beta^i_{I^j_i}(a)=1,\;i\in N,j\in M_i,\\
\beta^i_{I^j_i}(a)\lambda^i_{I^j_i}(a)=0,\;\tilde\beta^i_{I^j_i}(a)\tilde\lambda^i_{I^j_i}(a)=0,\\
0\le \beta^i_{I^j_i}(a),\;0\le \lambda^i_{I^j_i}(a),\;0\le \tilde\beta^i_{I^j_i}(a),\;0\le\tilde\lambda^i_{I^j_i}(a),\;i\in N,j\in M_i,a\in A(I^j_i).
\end{array}
\end{equation}}}
\end{theorem}
The proof of Theorem~\ref{sgpethm1} is essentially the same as that of Theorem~\ref{nscthm1}. 
Let ${\cal C}^i_{I^j_i}(I^j_i|\beta)=\sum\limits_{h\in I^j_i}{\cal C}^i_{I^j_i}(h|\beta)$. 
When ${\cal C}^i_{I^j_i}(I^j_i|\beta^*)>0$, we have ${\cal Y}^i(I^j_i|\beta^*)>0$ and consequently, $u^i((a,\varrho^i_{I^j_i}(\beta^{*-I^j_i},\tilde\beta^*))\Diamond I^j_i)={\cal C}^i_{I^j_i}(I^j_i|\beta^*) u^i(a,\varrho^i_{I^j_i}(\beta^{*-I^j_i},\tilde\beta^*),\mu^*|I^j_i)$ since $\mu^{*i}_{I^j_i}(h)=\frac{{\cal Y}^i(h|\beta^*)}{{\cal Y}^i(I^j_i|\beta^*)}=\frac{{\cal C}^i_{I^j_i}(h|\beta^*)}{{\cal C}^i_{I^j_i}(I^j_i|\beta^*)}$. Therefore, Definition~\ref{edsgpe1} can be equivalently rewritten as follows.
\begin{definition}[\bf A Characterization of Subgame Perfect Equilibrium under Conditional Expected Payoffs]\label{edsgpe2}{\em A behavioral strategy profile $\beta^*$ is a subgame perfect equilibrium if $\beta^*$ together with $(\tilde\beta^*,\mu^*)$ satisfies the properties: \\ (i) $\beta^{*i}_{I^j_i}(a')=0$ for any $i\in N$, $j\in M_i$ and $a',a''\in A(I^j_i)$ with ${\cal C}^i_{I^j_i}(I^j_i|\beta^*)>0$ and
$u^i(a'',\varrho^i_{I^j_i}(\beta^{*-I^j_i},\tilde\beta^*),\mu^*|I^j_i)> u^i(a',\varrho^i_{I^j_i}(\beta^{*-I^j_i},\tilde\beta^*),\mu^*|I^j_i)$;\\
 (ii) $\tilde\beta^{*i}_{I^j_i}(a')=0$ for any $i\in N$, $j\in M_i$ and $a',a''\in A(I^j_i)$ with
$u^i(a'',\varrho^i_{I^j_i}(\beta^{*-I^j_i},\tilde\beta^*),\mu^*|I^j_i)> u^i(a',\varrho^i_{I^j_i}(\beta^{*-I^j_i},\tilde\beta^*),\mu^*|I^j_i)$; and \\
(iii)
 $\mu^*=(\mu^{*i}_{I^j_i}(h):i\in N,j\in M_i,h\in I^j_i)$ is a solution to the system~(\ref{sgbseq1}).}
\end{definition}
Comparing Definition~\ref{edsgpe2} with Definition~\ref{edsgpe1}, one can see that Definition~\ref{edsgpe2} requires only the computation of $u^i(a, \varrho^i_{I^j_i}(\beta^{-I^j_i},\tilde\beta), \mu|I^j_i)$ whereas Definition~\ref{edsgpe1} requires the computation of both $u^i((a, \varrho^i_{I^j_i}(\beta^{-I^j_i},\tilde\beta))\Diamond I^j_i)$ and $u^i(a, \varrho^i_{I^j_i}(\beta^{-I^j_i},\tilde\beta), \mu|I^j_i)$.
As a corollary of Theorem~\ref{sgpethm1}, we come to the following conclusion.
\begin{corollary}\label{sgpeco1}
{\em $\beta^*$ is a subgame perfect equilibrium if and only if there exists $(\tilde\beta^*,\mu^*,\lambda^*$, $\tilde\lambda^*,\zeta^*,\tilde\zeta^*)$ together with $\beta^*$ satisfying the system, {\small \begin{equation}\setlength{\abovedisplayskip}{1.2pt}
\setlength{\belowdisplayskip}{1.2pt}
\label{nscsgpe2}
\begin{array}{l}
{\cal C}^i_{I^j_i}(I^j_i|\beta)u^i(a,\varrho^i_{I^j_i}(\beta^{-I^j_i},\tilde\beta),\mu|I^j_i) +\lambda^i_{I^j_i}(a)-\zeta^i_{I^j_i}=0,\;i\in N,j\in M_i,a\in A(I^j_i),\\
u^i(a,\varrho^i_{I^j_i}(\beta^{-I^j_i},\tilde\beta),\mu|I^j_i) +\tilde\lambda^i_{I^j_i}(a)-\tilde\zeta^i_{I^j_i}=0,\;i\in N,j\in M_i,a\in A(I^j_i),\\
{\cal Y}^i(I^j_i|\beta)\mu^i_{I^j_i}(h)={\cal Y}^i(h|\beta),\;0\le \mu^i_{I^j_i}(h),\;i\in N,j\in M_i,h\in I^j_i,\\
\sum\limits_{h\in I^j_i}\mu^i_{I^j_i}(h)=1,\;\sum\limits_{a\in A(I^j_i)}\beta^i_{I^j_i}(a)=1,\;\sum\limits_{a\in A(I^j_i)}\tilde\beta^i_{I^j_i}(a)=1,\;i\in N,j\in M_i,\\
\beta^i_{I^j_i}(a)\lambda^i_{I^j_i}(a)=0,\;\tilde\beta^i_{I^j_i}(a)\tilde\lambda^i_{I^j_i}(a)=0,\\
0\le \beta^i_{I^j_i}(a),\;0\le \lambda^i_{I^j_i}(a),\;0\le \tilde\beta^i_{I^j_i}(a),\;0\le\tilde\lambda^i_{I^j_i}(a),\;i\in N,j\in M_i,a\in A(I^j_i).
\end{array}
\end{equation}}}
\end{corollary}

To further boost the applications of subgame perfect equilibrium, we will utilize Theorem~\ref{sgpethm1} or Corollary~\ref{sgpeco1} to develop differentiable path-following methods to compute such an equilibrium in Section~\ref{dpm}. 

\section{Illustrative Examples\label{examples}}

This section illustrates through examples how one can employ Definition~\ref{edne1}, Definition~\ref{edne2}, Definition~\ref{edsgpe1} or Definition~\ref{edsgpe2} to analytically find all NashEBSs and SGPEBSs. 
These games are deliberately chosen to highlight the indispensable role of the extra behavioral strategy profile $\tilde\beta$ in the equivalent definitions of NashEBS  for achieving global rationality in Definition~\ref{ned1} through local sequential rationality.
 
 \begin{figure}[H]
    \centering
    \begin{minipage}{0.49\textwidth}
        \centering
        \includegraphics[width=0.8\textwidth, height=0.15\textheight]{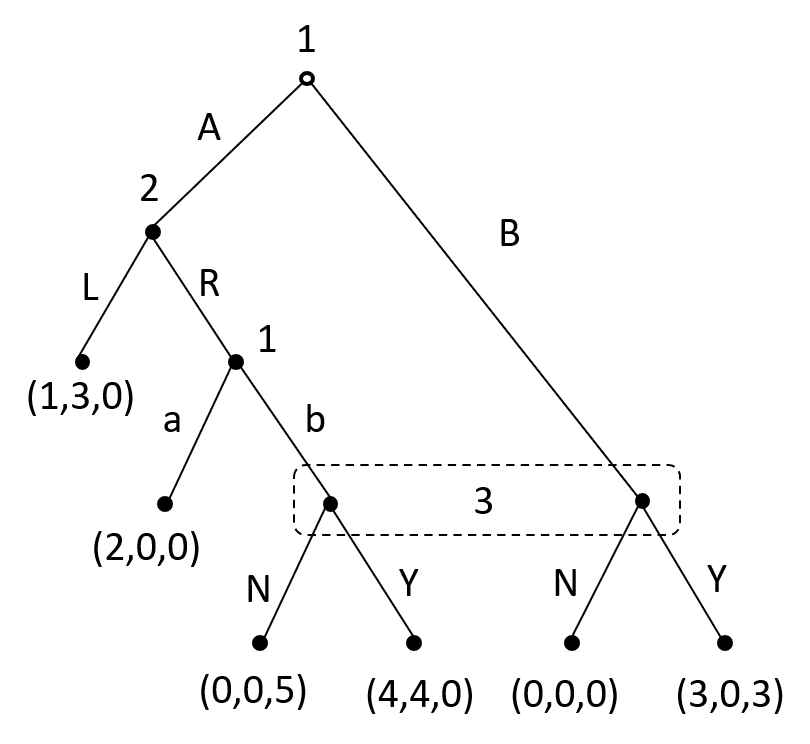}
                \caption{\label{fexm1}\scriptsize An Extensive-Form Game}
\end{minipage}\hfill
    \begin{minipage}{0.49\textwidth}
        \centering
        \includegraphics[width=0.8\textwidth, height=0.15\textheight]{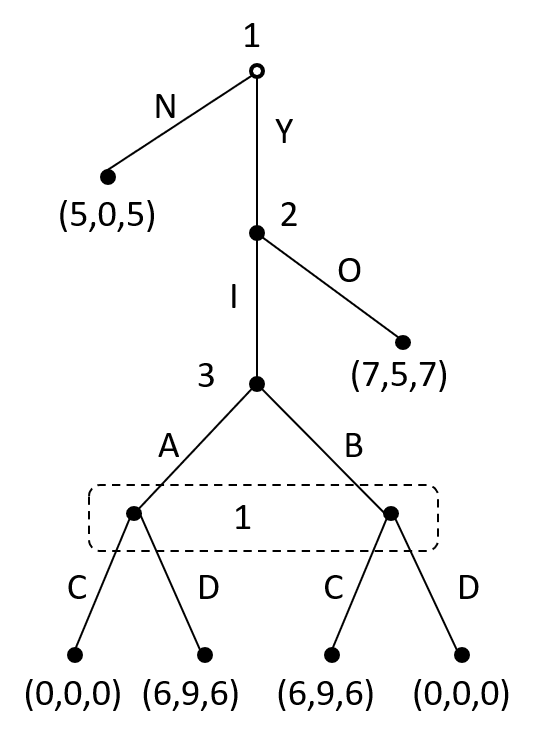}
\caption{\label{fexm2}\scriptsize  An Extensive-Form Game}\end{minipage}
 \end{figure}

 {\small
\begin{example} {\em Consider the game in Fig.~\ref{fexm1}. 
The information sets consist of $I^1_1=\{\emptyset\}$, $I^2_1=\{\langle A, R \rangle\}$, $I^1_2=\{\langle A\rangle\}$, and $I^1_3=\{\langle A, R,b
\rangle, \langle B
\rangle\}$. We denote by $(\beta,\tilde\beta,\mu)$ a triple satisfying Definition~\ref{edne1}. Each NashEBS is represented in the form $\big((\beta^1_{I^1_1}(A),\beta^1_{I^1_1}(B)), (\beta^1_{I^2_1}(a), \beta^1_{I^2_1}(b))$, $(\beta^2_{I^1_2}(L),\beta^2_{I^1_2}(R))$, $(\beta^3_{I^1_3}(N),\beta^3_{I^1_3}(Y))\big)$. The expected payoffs and conditional expected payoffs at $(\beta,\tilde\beta,\mu)$ are given by
{\footnotesize
\[\setlength{\abovedisplayskip}{1.2pt}
\setlength{\belowdisplayskip}{1.2pt}
\begin{array}{l}
u^1((A,\varrho^1_{I^1_1}(\beta^{-I^1_1},\tilde\beta))\land I^1_1)=\beta^2_{I^1_2}(L)+2\beta^2_{I^1_2}(R)\tilde\beta^1_{I^2_1}(a)+4\beta^2_{I^1_2}(R)\tilde\beta^1_{I^2_1}(b)\beta^3_{I^1_3}(Y),\\

u^1((B,\varrho^1_{I^1_1}(\beta^{-I^1_1},\tilde\beta))\land I^1_1)=3\beta^3_{I^1_3}(Y),\\

u^1((a,\varrho^1_{I^2_1}(\beta^{-I^2_1},\tilde\beta))\land I^2_1) = 2\beta^1_{I^1_1}(A)\beta^2_{I^1_2}(R),\;

u^1((b,\varrho^1_{I^2_1}(\beta^{-I^2_1},\tilde\beta))\land I^2_1) = 4\beta^1_{I^1_1}(A)\beta^2_{I^1_2}(R)\beta^3_{I^1_3}(Y),\\

u^2((L,\varrho^2_{I^1_2}(\beta^{-I^1_2},\tilde\beta))\land I^1_2)=3\beta^1_{I^1_1}(A),\;

u^2((R,\varrho^2_{I^1_2}(\beta^{-I^1_2},\tilde\beta))\land I^1_2)=4\beta^1_{I^1_1}(A)\beta^1_{I^2_1}(b)\beta^3_{I^1_3}(Y),\\

u^3((N,\varrho^3_{I^1_3}(\beta^{-I^1_3},\tilde\beta))\land I^1_3)=5\beta^1_{I^1_1}(A)\beta^2_{I^1_2}(R)\beta^1_{I^2_1}(b),\;

u^3((Y,\varrho^3_{I^1_3}(\beta^{-I^1_3},\tilde\beta))\land I^1_3)=3\beta^1_{I^1_1}(B),\\

u^1(A,\varrho^1_{I^1_1}(\beta^{-I^1_1},\tilde\beta),\mu|I^1_1)=\beta^2_{I^1_2}(L)+2\beta^2_{I^1_2}(R)\tilde\beta^1_{I^2_1}(a)+4\beta^2_{I^1_2}(R)\tilde\beta^1_{I^2_1}(b)\beta^3_{I^1_3}(Y),\\

u^1(B,\varrho^1_{I^1_1}(\beta^{-I^1_1},\tilde\beta),\mu| I^1_1)=3\beta^3_{I^1_3}(Y),\\

u^1(a,\varrho^1_{I^2_1}(\beta^{-I^2_1},\tilde\beta),\mu| I^2_1) = 2,\;

u^1(b,\varrho^1_{I^2_1}(\beta^{-I^2_1},\tilde\beta)| I^2_1) = 4\beta^3_{I^1_3}(Y),\\

u^2(L,\varrho^2_{I^1_2}(\beta^{-I^1_2},\tilde\beta),\mu| I^1_2)=3,\;

u^2(R,\varrho^2_{I^1_2}(\beta^{-I^1_2},\tilde\beta),\mu| I^1_2)=4\beta^1_{I^2_1}(b)\beta^3_{I^1_3}(Y),\\

u^3(N,\varrho^3_{I^1_3}(\beta^{-I^1_3},\tilde\beta),\mu| I^1_3)=5\mu^3_{I^1_3}(\langle A,R,b\rangle|\beta),\;

u^3(Y,\varrho^3_{I^1_3}(\beta^{-I^1_3},\tilde\beta),\mu| I^1_3)=3\mu^3_{I^1_3}(\langle B\rangle|\beta).
\end{array}\]
}

\noindent
{\bf Case (1)}. Suppose that $u^3((N,\varrho^3_{I^1_3}(\beta^{-I^1_3},\tilde\beta))\land I^1_3)>u^3((Y,\varrho^3_{I^1_3}(\beta^{-I^1_3},\tilde\beta))\land I^1_3)$. Then, $\beta^3_{I^1_3}(Y)=0$, $\beta^1_{I^1_1}(A)>0$, $\beta^2_{I^1_2}(R)>0$, and $\beta^1_{I^2_1}(b)>0$. Thus, $u^2((L,\varrho^2_{I^1_2}(\beta^{-I^1_2},\tilde\beta))\land I^1_2)>u^2((R,\varrho^2_{I^1_2}(\beta^{-I^1_2},\tilde\beta))\land I^1_2)$ and consequently,  $\beta^2_{I^1_2}(R)=0$. A contradiction occurs. The case is excluded.
\newline
{\bf Case (2)}. Suppose that $u^3((Y,\varrho^3_{I^1_3}(\beta^{-I^1_3},\tilde\beta))\land I^1_3)>u^3((N,\varrho^3_{I^1_3}(\beta^{-I^1_3},\tilde\beta))\land I^1_3)$. Then, $\beta^3_{I^1_3}(N)=0$ and  $\beta^1_{I^1_1}(B)>0$. Thus, $u^1(b,\varrho^1_{I^2_1}(\beta^{-I^2_1},\tilde\beta),\mu|I^2_1)>u^1(a,\varrho^1_{I^2_1}(\beta^{-I^2_1},\tilde\beta),\mu| I^2_1)$ and $u^1((B,\varrho^1_{I^1_1}(\beta^{-I^1_1},\tilde\beta))\land I^1_1)\ge u^1((A,\varrho^1_{I^1_1}(\beta^{-I^1_1},\tilde\beta))\land I^1_1)$. Therefore, $\tilde\beta^1_{I^2_1}(a)=0$ and $\beta^2_{I^1_2}(R)\le\frac{2}{3}$. \newline
{\bf (a)}. Assume that $\beta^2_{I^1_2}(R)<\frac{2}{3}$. Then, $\beta^1_{I^1_1}(A)=0$. The game has a class of Nash equilibria given by $(B, (1-\beta^1_{I^2_1}(b),\beta^1_{I^2_1}(b)), (1-\beta^2_{I^1_2}(R),\beta^2_{I^1_2}(R)), Y)$ with $\beta^2_{I^1_2}(R)<\frac{2}{3}$.
\newline
{\bf (b)}. Assume that $\beta^2_{I^1_2}(R)=\frac{2}{3}$.
\newline
{\bf (i)}. Consider the situation that $\beta^1_{I^1_1}(A)>0$. Then, $u^1((b,\varrho^1_{I^2_1}(\beta^{-I^2_1},\tilde\beta))\land I^2_1)>u^1((a,\varrho^1_{I^2_1}(\beta^{-I^2_1},\tilde\beta))\land I^2_1)$ and consequently,  $\beta^1_{I^2_1}(a)=0$. Thus, $u^2((R,\varrho^2_{I^1_2}(\beta^{-I^1_2}, \tilde\beta))\land I^1_2)>u^2((L,\varrho^2_{I^1_2}(\beta^{-I^1_2},\tilde\beta))\land I^1_2)$ and accordingly, $\beta^2_{I^1_2}(L)=0$. A contradiction occurs. The situation cannot arise.
\newline
{\bf (ii)}. Consider the situation that $\beta^1_{I^1_1}(A)=0$. The game has a class of Nash equilibria given by $(B, (1-\beta^1_{I^2_1}(b),\beta^1_{I^2_1}(b)), (\frac{1}{3},\frac{2}{3}), Y)$.
\newline
{\bf Case (3)}. Suppose that $u^3((Y,\varrho^3_{I^1_3}(\beta^{-I^1_3},\tilde\beta))\land I^1_3)=u^3((N,\varrho^3_{I^1_3}(\beta^{-I^1_3},\tilde\beta))\land I^1_3)$. Then, $\beta^1_{I^1_1}(A)>0$, which implies $u^1((A,\varrho^1_{I^1_1}(\beta^{-I^1_1},\tilde\beta))\land I^1_1)\ge u^1((B,\varrho^1_{I^1_1}(\beta^{-I^1_1},\tilde\beta))\land I^1_1)$.
\newline
{\bf (a)}. Suppose that $u^1(a,\varrho^1_{I^2_1}(\beta^{-I^2_1},\tilde\beta),\mu| I^2_1)>u^1(b,\varrho^1_{I^2_1}(\beta^{-I^2_1},\tilde\beta),\mu| I^2_1)$. Then, $\tilde\beta^1_{I^2_1}(b)=0$ and $\beta^3_{I^1_3}(Y)<\frac{1}{2}$. Thus, $u^2((L,\varrho^2_{I^1_2}(\beta^{-I^1_2},\tilde\beta))\land I^1_2)>u^2((R,\varrho^2_{I^1_2}(\beta^{-I^1_2},\tilde\beta))\land I^1_2)$ and accordingly, $\beta^2_{I^1_2}(R)=0$, which implies $\beta^1_{I^1_1}(B)=0$. 
\newline
{\bf (i)}. Assume that $\beta^3_{I^1_3}(Y)<\frac{1}{3}$.
Then, $u^1((A,\varrho^1_{I^1_1}(\beta^{-I^1_1},\tilde\beta))\land I^1_1)>u^1((B,\varrho^1_{I^1_1}(\beta^{-I^1_1},\tilde\beta))\land I^1_1)$ and consequently,  $\beta^1_{I^1_1}(B)=0$. The game has a class of Nash equilibria given by $(A, (\beta^1_{I^2_1}(b),1-\beta^1_{I^2_1}(b)), L, (1-\beta^3_{I^1_3}(Y), \beta^3_{I^1_3}(Y)))$ with $\beta^3_{I^1_3}(Y)<\frac{1}{3}$.\newline
{\bf (ii)}. Assume that $\beta^3_{I^1_3}(Y)=\frac{1}{3}$.
Then, $u^1((A,\varrho^1_{I^1_1}(\beta^{-I^1_1},\tilde\beta))\land I^1_1)=u^1((B,\varrho^1_{I^1_1}(\beta^{-I^1_1}, \tilde\beta))\land I^1_1)$. The game has a class of Nash equilibria given by $(A, (\beta^1_{I^2_1}(b),1-\beta^1_{I^2_1}(b)), L, (\frac{2}{3},\frac{1}{3}))$.
\newline
{\bf (b)}. Suppose that $u^1(b,\varrho^1_{I^2_1}(\beta^{-I^2_1},\tilde\beta)| I^2_1)>u^1(a,\varrho^1_{I^2_1}(\beta^{-I^2_1},\tilde\beta)| I^2_1)$. Then, $\tilde\beta^1_{I^2_1}(a)=0$ and $\beta^3_{I^1_3}(Y)>\frac{1}{2}$. 
\newline
{\bf (i)}. Assume that $u^2((L,\varrho^2_{I^1_2}(\beta^{-I^1_2},\tilde\beta))\land I^1_2)>u^2((R,\varrho^2_{I^1_2}(\beta^{-I^1_2},\tilde\beta))\land I^1_2)$. Then, $\beta^2_{I^1_2}(R)=0$. Thus,
$u^1((B,\varrho^1_{I^1_1}(\beta^{-I^1_1},\tilde\beta))\land I^1_1)>u^1((A,\varrho^1_{I^1_1}(\beta^{-I^1_1},\tilde\beta))\land I^1_1)$. A contradiction occurs. The assumption cannot arise.
\newline
{\bf (ii)}. Assume that $u^2((R,\varrho^2_{I^1_2}(\beta^{-I^1_2},\tilde\beta))\land I^1_2)>u^2((L,\varrho^2_{I^1_2}(\beta^{-I^1_2},\tilde\beta))\land I^1_2)$. Then, $\beta^2_{I^1_2}(L)=0$. Thus,
$u^1((b,\varrho^1_{I^2_1}(\beta^{-I^2_1},\tilde\beta))\land I^2_1)>u^1((a,\varrho^1_{I^2_1}(\beta^{-I^2_1},\tilde\beta))\land I^2_1)$ and consequently,  $\beta^1_{I^2_1}(a)=0$. Therefore, $\beta^3_{I^1_3}(Y)>\frac{3}{4}$ and $5\beta^1_{I^1_1}(A)-3\beta^1_{I^1_1}(B)=0$. Hence, $u^1((A,\varrho^1_{I^1_1}(\beta^{-I^1_1},\tilde\beta))\land I^1_1)=u^1((B,\varrho^1_{I^1_1}(\beta^{-I^1_1},\tilde\beta))\land I^1_1)$, which implies $4\beta^3_{I^1_3}(Y)-3\beta^3_{I^1_3}(Y)=0$. A contradiction occurs and the assumption cannot be sustained.
\newline
{\bf (iii)}. Assume that $u^2((R,\varrho^2_{I^1_2}(\beta^{-I^1_2},\tilde\beta))\land I^1_2)=u^2((L,\varrho^2_{I^1_2}(\beta^{-I^1_2},\tilde\beta))\land I^1_2)$. Then, $\beta^1_{I^2_1}(b)\beta^3_{I^1_3}(Y)=\frac{3}{4}$ and consequently, $\beta^1_{I^2_1}(b)>0$. This together with $u^1((A,\varrho^1_{I^1_1}(\beta^{-I^1_1},\tilde\beta))\land I^1_1)\ge u^1((B,\varrho^1_{I^1_1}(\beta^{-I^1_1},\tilde\beta))\land I^1_1)$ deduces that $\beta^2_{I^1_2}(R)>0$. Therefore, $u^1((b,\varrho^1_{I^2_1}(\beta^{-I^2_1},\tilde\beta))\land I^2_1)>u^1((a,\varrho^1_{I^2_1}(\beta^{-I^2_1},\tilde\beta))\land I^2_1)$ and consequently,  $\beta^1_{I^2_1}(a)=0$. 
The game has a Nash equilibrium given by the unique solution to the system,
$\beta^2_{I^1_2}(L)+4\beta^2_{I^1_2}(R)\beta^3_{I^1_3}(Y)-3\beta^3_{I^1_3}(Y)=0$,
$4\beta^1_{I^1_1}(A)\beta^3_{I^1_3}(Y)-3\beta^1_{I^1_1}(A)=0$,
$5\beta^1_{I^1_1}(A)\beta^2_{I^1_2}(R)-3\beta^1_{I^1_1}(B)=0$,
that is, 
$((\frac{24}{49},\frac{25}{49}), b, (\frac{3}{8},\frac{5}{8}), (\frac{1}{4},\frac{3}{4}))$. 
\newline
{\bf (c)}. Suppose that $u^1(a,\varrho^1_{I^2_1}(\beta^{-I^2_1},\tilde\beta)| I^2_1)=u^1(b,\varrho^1_{I^2_1}(\beta^{-I^2_1},\tilde\beta)| I^2_1)$.
Then, $\beta^3_{I^1_3}(Y)=\frac{1}{2}$. Thus, $u^2((L,\varrho^2_{I^1_2}(\beta^{-I^1_2},\tilde\beta))\land I^1_2)>u^2((R,\varrho^2_{I^1_2}(\beta^{-I^1_2},\tilde\beta))\land I^1_2)$ and accordingly, $\beta^2_{I^1_2}(R)=0$. 
Therefore, $u^1((B,\varrho^1_{I^1_1}(\beta^{-I^1_1},\tilde\beta))\land I^1_1)>u^1((A,\varrho^1_{I^1_1}(\beta^{-I^1_1},\tilde\beta))\land I^1_1)$ and consequently,  $\beta^1_{I^1_1}(A)=0$. A contradiction occurs and the assumption is excluded.

The cases (1)-(3) together show that the game has three types of Nash equilibria given by\newline
Type 1: $(B, (\beta^1_{I^2_1}(a),1-\beta^1_{I^2_1}(a)), (1-\beta^2_{I^1_2}(R),\beta^2_{I^1_2}(R)), Y)$ with $\beta^2_{I^1_2}(R)\le\frac{2}{3}$.\newline
Type 2: $(A, (\beta^1_{I^2_1}(a),1-\beta^1_{I^2_1}(a)), L, (1-\beta^3_{I^1_3}(Y), \beta^3_{I^1_3}(Y)))$ with $\beta^3_{I^1_3}(Y)\le\frac{1}{3}$.
\newline
Type 3: $((\frac{24}{49}, \frac{25}{49}), b, (\frac{3}{8},\frac{5}{8}), (\frac{1}{4}, \frac{3}{4}))$.
}
\end{example}

\begin{example} {\em Consider the game in Fig.~\ref{fexm2}. The information sets consist of $I^1_1=\{\emptyset\}$, $I^2_1=\{\langle Y, I, A\rangle, \langle Y, I, B\rangle\}$, $I^1_2=\{\langle Y\rangle\}$, and $I^1_3=\{\langle Y, I\rangle\}$. 
We denote by $(\beta, \mu,\tilde\beta)$ a triple meeting the properties in Definition~\ref{edne2}. Each NashEBS is presented in the form of $\beta=((\beta^1_{I^1_1}(N), \beta^1_{I^1_1}(Y)), (\beta^1_{I^2_1}(C), \beta^1_{I^2_1}(D)), (\beta^2_{I^1_2}(I), \beta^2_{I^1_2}(O)), (\beta^3_{I^1_3}(A), \beta^3_{I^1_3}(B)))$. The conditional expected payoffs at $(\beta, \mu,\tilde\beta)$ on $I^j_i$ are given by
{\small \[\setlength{\abovedisplayskip}{1.2pt}
\setlength{\belowdisplayskip}{1.2pt}\begin{array}{l}
\omega(I^1_1|\beta)=1,\; \omega(I^2_1|\beta)=\beta^1_{I^1_1}(Y)\beta^2_{I^1_2}(I),\;\omega(I^1_2|\beta)=\beta^1_{I^1_1}(Y),\; \omega(I^1_3|\beta)=\omega(I^2_1|\beta),\;
 {\cal S}^1(I^2_1|\beta)=\beta^2_{I^1_2}(I),\\

u^1(N,\varrho^1_{I^1_1}(\beta^{-I^1_1},\tilde\beta),\mu|I^1_1)=5,\; u^1(Y,\varrho^1_{I^1_1}(\beta^{-I^1_1},\tilde\beta),\mu|I^1_1)=6\beta^2_{I^1_2}(I)(\beta^3_{I^1_3}(A)\tilde\beta^1_{I^2_1}(D)+\beta^3_{I^1_3}(B)\tilde\beta^1_{I^2_1}(C)) +7\beta^2_{I^1_2}(O),\\

u^1(C,\varrho^1_{I^2_1}(\beta^{-I^2_1},\tilde\beta),\mu|I^2_1)=6\mu^1_{I^2_1}(\langle Y, I, B\rangle),\; u^1(D,\varrho^1_{I^2_1}(\beta^{-I^2_1},\tilde\beta),\mu|I^2_1)=6\mu^1_{I^2_1}(\langle Y, I, A\rangle),\\

u^2(I,\varrho^2_{I^1_2}(\beta^{-I^1_2},\tilde\beta),\mu|I^1_2)=9(\beta^3_{I^1_3}(A)\beta^1_{I^2_1}(D) + \beta^3_{I^1_3}(B)\beta^1_{I^2_1}(A)),\; u^2(O,\varrho^2_{I^1_2}(\beta^{-I^1_2},\tilde\beta),\mu|I^1_2)=5,\\

u^3(A,\varrho^3_{I^1_3}(\beta^{-I^1_3},\tilde\beta),\mu|I^1_3)= 6\beta^1_{I^2_1}(D),\; u^3(B,\varrho^3_{I^1_3}(\beta^{-I^1_3},\tilde\beta),\mu|I^1_3)=6\beta^1_{I^2_1}(C),
\end{array}\]}where ${\cal S}^1(I^2_1|\beta)\mu^1_{I^2_1}(\langle Y, I, B\rangle)= {\cal S}^1(\langle Y, I, B\rangle|\beta)$ and ${\cal S}^1(I^2_1|\beta)\mu^1_{I^2_1}(\langle Y, I, A\rangle)= {\cal S}^1(\langle Y, I, A\rangle|\beta)$.

\noindent
{\bf Case (1)}. Suppose that $\omega(I^2_1|\beta)(u^1(C,\varrho^1_{I^2_1}(\beta^{-I^2_1},\tilde\beta),\mu|I^2_1)-u^1(D,\varrho^1_{I^2_1}(\beta^{-I^2_1},\tilde\beta),\mu|I^2_1))>0$. Then, $\beta^1_{I^2_1}(D)=0$, $\tilde\beta^1_{I^2_1}(D)=0$, $\omega(I^2_1|\beta)>0$, and $\frac{1}{2}-\mu^1_{I^2_1}(\langle Y, I, A\rangle)>0$. Thus, $u^3(B,\varrho^3_{I^1_3}(\beta^{-I^1_3},\tilde\beta),\mu|I^1_3)>u^3(A,\varrho^3_{I^1_3}(\beta^{-I^1_3},\tilde\beta),\mu|I^1_3)$ and consequently, $\beta^3_{I^1_3}(A)=0$. Therefore,
  $u^1(Y,\varrho^1_{I^1_1}(\beta^{-I^1_1},\tilde\beta),\mu|I^1_1)>u^1(N,\varrho^1_{I^1_1}(\beta^{-I^1_1},\tilde\beta),\mu|I^1_1)$ and $\omega(I^1_2|\beta)(u^2(I,\varrho^2_{I^1_2}(\beta^{-I^1_2},\tilde\beta),\mu|I^1_2) -u^2(O,\varrho^2_{I^1_2}(\beta^{-I^1_2},\tilde\beta),\mu|I^1_2))>0$. Accordingly, $\beta^1_{I^1_1}(N)=0$ and $\beta^2_{I^1_2}(O)=0$. The game has a NashEBS given by $(Y, C, I, B)$. 
  
 \noindent {\bf Case (2)}. Suppose that $\omega(I^2_1|\beta)(u^1(D,\varrho^1_{I^2_1}(\beta^{-I^2_1},\tilde\beta),\mu|I^2_1)-u^1(C,\varrho^1_{I^2_1}(\beta^{-I^2_1},\tilde\beta),\mu|I^2_1))>0$. Then, $\beta^1_{I^2_1}(C)=0$, $\tilde\beta^1_{I^2_1}(C)=0$, $\omega(I^2_1|\beta)>0$, and $\frac{1}{2}-\mu^1_{I^2_1}(\langle Y, I, B\rangle)>0$. Thus, $u^3(A,\varrho^3_{I^1_3}(\beta^{-I^1_3},\tilde\beta),\mu|I^1_3)>u^3(B,\varrho^3_{I^1_3}(\beta^{-I^1_3},\tilde\beta),\mu|I^1_3)$ and consequently, $\beta^3_{I^1_3}(B)=0$. Therefore,
  $u^1(Y,\varrho^1_{I^1_1}(\beta^{-I^1_1},\tilde\beta),\mu|I^1_1)>u^1(N,\varrho^1_{I^1_1}(\beta^{-I^1_1},\tilde\beta),\mu|I^1_1)$ and $\omega(I^1_2|\beta)(u^2(I,\varrho^2_{I^1_2}(\beta^{-I^1_2},\tilde\beta),\mu|I^1_2) -u^2(O,\varrho^2_{I^1_2}(\beta^{-I^1_2},\tilde\beta),\mu|I^1_2))>0$. Accordingly, $\beta^1_{I^1_1}(N)=0$ and $\beta^2_{I^1_2}(O)=0$. The game has a NashEBS given by $(Y, D, I, A)$. 
  
  \noindent {\bf Case (3)}. Suppose that $\omega(I^2_1|\beta)(u^1(D,\varrho^1_{I^2_1}(\beta^{-I^2_1},\tilde\beta),\mu|I^2_1)-u^1(C,\varrho^1_{I^2_1}(\beta^{-I^2_1},\tilde\beta),\mu|I^2_1))=0$. Then at least one of $\omega(I^2_1|\beta)=0$ 
  and $\frac{1}{2}-\mu^1_{I^2_1}(\langle Y, I, B\rangle)=0$ holds. 
  
  \noindent {\bf (a)}. Assume that  $\omega(I^2_1|\beta)=0$. Then at least one of $\beta^1_{I^1_1}(Y)=0$ and $\beta^2_{I^1_2}(I)=0$ holds. 
  
  \noindent {\bf (i)}. Consider the scenario that $\beta^1_{I^1_1}(Y)=0$. Then, $u^1(N,\varrho^1_{I^1_1}(\beta^{-I^1_1},\tilde\beta),\mu|I^1_1)\ge u^1(Y,\varrho^1_{I^1_1}(\beta^{-I^1_1},\tilde\beta),\mu|I^1_1)$ and consequently,  $\beta^2_{I^1_2}(O)\le\frac{5}{7}$. Thus, ${\cal S}^1(I^2_1|\beta)=\beta^2_{I^1_2}(I)>0$.
    
  \noindent $\bullet$ Postulate that $\beta^3_{I^1_3}(A)>\frac{1}{2}$. Then, $u^1(D,\varrho^1_{I^2_1}(\beta^{-I^2_1},\tilde\beta),\mu|I^2_1)>u^1(C,\varrho^1_{I^2_1}(\beta^{-I^2_1},\tilde\beta),\mu|I^2_1)$ and consequently, $\tilde\beta^1_{I^2_1}(C)=0$. Thus, $\beta^2_{I^1_2}(I)(1+6\beta^3_{I^1_3}(B))\ge 2$.
  The game has a class of NashEBSs given by $(N, (\beta^1_{I^2_1}(C),1-\beta^1_{I^2_1}(C)), (\beta^2_{I^1_2}(I), 1-\beta^2_{I^1_2}(I)), (\beta^3_{I^1_3}(A),1-\beta^3_{I^1_3}(A)))$ with $\beta^2_{I^1_2}(I)(1+6\beta^3_{I^1_3}(B))\ge 2$ and $\beta^3_{I^1_3}(B)<\frac{1}{2}$.

\noindent $\bullet$ Postulate that $\beta^3_{I^1_3}(A)<\frac{1}{2}$. Then, $u^1(C,\varrho^1_{I^2_1}(\beta^{-I^2_1},\tilde\beta),\mu|I^2_1)>u^1(D,\varrho^1_{I^2_1}(\beta^{-I^2_1},\tilde\beta),\mu|I^2_1)$ and consequently, $\tilde\beta^1_{I^2_1}(D)=0$. Thus, $\beta^2_{I^1_2}(I)(1+6\beta^3_{I^1_3}(A))\ge 2$.
  The game has a class of NashEBSs given by $(N, (\beta^1_{I^2_1}(C),1-\beta^1_{I^2_1}(C)), (\beta^2_{I^1_2}(I), 1-\beta^2_{I^1_2}(I)), (\beta^3_{I^1_3}(A),1-\beta^3_{I^1_3}(A)))$ with $\beta^2_{I^1_2}(I)(1+6\beta^3_{I^1_3}(A))\ge 2$ and $\beta^3_{I^1_3}(A)<\frac{1}{2}$.
    
    \noindent $\bullet$ Postulate that $\beta^3_{I^1_3}(A)=\frac{1}{2}$.  Then, $\beta^2_{I^1_2}(I)\ge \frac{1}{2}$.
    The game has a class of NashEBSs given by $(N, (\beta^1_{I^2_1}(C),1-\beta^1_{I^2_1}(C)), (\beta^2_{I^1_2}(I), 1-\beta^2_{I^1_2}(I)), (\frac{1}{2},\frac{1}{2}))$ with $\beta^2_{I^1_2}(I)\ge \frac{1}{2}$.
     
  \noindent {\bf (ii)}. Consider the scenario that $\beta^1_{I^1_1}(Y)>0$. Then, $\beta^2_{I^1_2}(I)=0$. Thus, $u^1(Y,\varrho^1_{I^1_1}(\beta^{-I^1_1},\tilde\beta),\mu|I^1_1)>u^1(N,\varrho^1_{I^1_1}(\beta^{-I^1_1},\tilde\beta),\mu|I^1_1)$ and  $u^2(O,\varrho^2_{I^1_2}(\beta^{-I^1_2},\tilde\beta),\mu|I^1_2)-u^2(I,\varrho^2_{I^1_2}(\beta^{-I^1_2},\tilde\beta),\mu|I^1_2)\ge 0$. Consequently, $\beta^1_{I^1_1}(N)=0$ and $\beta^3_{I^1_3}(A)\beta^1_{I^2_1}(D) + \beta^3_{I^1_3}(B)\beta^1_{I^2_1}(C)\le\frac{5}{9}$. The game has a class of NashEBSs given by $(Y, (\beta^1_{I^2_1}(C),1-\beta^1_{I^2_1}(C)), O, (\beta^3_{I^1_3}(A),1-\beta^3_{I^1_3}(A)))$ with $\beta^3_{I^1_3}(A)\beta^1_{I^2_1}(D) + \beta^3_{I^1_3}(B)\beta^1_{I^2_1}(C)\le\frac{5}{9}$.
  
  \noindent {\bf (b)}.  Assume that  $\omega(I^2_1|\beta)>0$.  Then, $\frac{1}{2}-\mu^1_{I^2_1}(\langle Y, I, B\rangle)=0$ and consequently, $\beta^3_{I^1_3}(A)=\frac{1}{2}$. Thus, $\omega(I^1_2|\beta)(u^2(O,\varrho^2_{I^1_2}(\beta^{-I^1_2},\tilde\beta),\mu|I^1_2) -u^2(I,\varrho^2_{I^1_2}(\beta^{-I^1_2},\tilde\beta),\mu|I^1_2))>0$ and accordingly, $\beta^2_{I^1_2}(I)=0$. A contradiction occurs. The assumption cannot be sustained.
    
  The cases (1)-(3) together show that the game has six classes of NashEBSs given by
  \\
  (1). $(Y, C, I, B)$.\\
  (2). $(Y, D, I, A)$.\\
  (3). $(N, (\beta^1_{I^2_1}(C),1-\beta^1_{I^2_1}(C)), (\beta^2_{I^1_2}(I), 1-\beta^2_{I^1_2}(I)), (\beta^3_{I^1_3}(A),1-\beta^3_{I^1_3}(A)))$ with $\beta^2_{I^1_2}(I)(1+6\beta^3_{I^1_3}(B))\ge 2$ and $\beta^3_{I^1_3}(B)<\frac{1}{2}$.\\
  (4). $(N, (\beta^1_{I^2_1}(C),1-\beta^1_{I^2_1}(C)), (\beta^2_{I^1_2}(I), 1-\beta^2_{I^1_2}(I)), (\beta^3_{I^1_3}(A),1-\beta^3_{I^1_3}(A)))$ with $\beta^2_{I^1_2}(I)(1+6\beta^3_{I^1_3}(A))\ge 2$ and $\beta^3_{I^1_3}(A)<\frac{1}{2}$.\\
  (5). $(N, (\beta^1_{I^2_1}(C),1-\beta^1_{I^2_1}(C)), (\beta^2_{I^1_2}(I), 1-\beta^2_{I^1_2}(I)), (\frac{1}{2},\frac{1}{2}))$ with $\beta^2_{I^1_2}(I)\ge \frac{1}{2}$.\\
  (6). $(Y, (\beta^1_{I^2_1}(C),1-\beta^1_{I^2_1}(C)), O, (\beta^3_{I^1_3}(A),1-\beta^3_{I^1_3}(A)))$ with $\beta^3_{I^1_3}(A)\beta^1_{I^2_1}(D) + \beta^3_{I^1_3}(B)\beta^1_{I^2_1}(C)\le\frac{5}{9}$.
     }
\end{example}

\begin{example}{\em Consider the game in Fig.~\ref{fexm2}. The information sets consist of $I^1_1=\{\emptyset\}$, $I^2_1=\{\langle Y, I, A \rangle, \langle Y, I, B \rangle\}$, $I^1_2=\{\langle Y\rangle\}$, and $I^1_3=\{\langle Y, I
\rangle\}$. We denote by $(\beta,\tilde\beta,\mu)$ a triple satisfying the properties in Definition~\ref{edsgpe1}. Each subgame perfect equilibrium is represented in the form $\big((\beta^1_{I^1_1}(N),\beta^1_{I^1_1}(Y)), (\beta^1_{I^2_1}(C)$, $\beta^1_{I^2_1}(D)),  (\beta^2_{I^1_2}(I),\beta^2_{I^1_2}(O)), (\beta^3_{I^1_3}(A),\beta^3_{I^1_3}(B))\big)$.  The expected payoffs and conditional expected payoffs at $(\beta,\tilde\beta,\mu)$ on $I^j_i$ are given by
{\small \[\setlength{\abovedisplayskip}{1.2pt}
\setlength{\belowdisplayskip}{1.2pt}
\begin{array}{l}
u^1((N, \varrho^1_{I^1_1}(\beta^{-I^1_1},\tilde\beta))\Diamond I^1_1)=5,\\

 u^1((Y, \varrho^1_{I^1_1}(\beta^{-I^1_1},\tilde\beta))\Diamond I^1_1)=6\beta^2_{I^1_2}(I)(\beta^3_{I^1_3}(A)\tilde\beta^1_{I^2_1}(D)
+ \beta^3_{I^1_3}(B)\tilde\beta^1_{I^2_1}(C))+7\beta^2_{I^1_2}(O),\\

u^1((C, \varrho^1_{I^2_1}(\beta^{-I^2_1},\tilde\beta))\Diamond I^2_1)=6\beta^3_{I^1_3}(B),\; u^1((D, \varrho^1_{I^2_1}(\beta^{-I^2_1},\tilde\beta))\Diamond I^2_1)=6\beta^3_{I^1_3}(A),\\

u^2((I, \varrho^2_{I^1_2}(\beta^{-I^1_2},\tilde\beta))\Diamond I^1_2) = 9(\beta^3_{I^1_3}(A)\beta^1_{I^2_1}(D)
+ \beta^3_{I^1_3}(B)\beta^1_{I^2_1}(C)),\\

u^2((O, \varrho^2_{I^1_2}(\beta^{-I^1_2},\tilde\beta))\Diamond I^1_2) = 5,\\

u^3((A, \varrho^3_{I^1_3}(\beta^{-I^1_3},\tilde\beta))\Diamond I^1_3) = 6\beta^1_{I^2_1}(D),\; u^3((B, \varrho^3_{I^1_3}(\beta^{-I^1_3},\tilde\beta))\Diamond I^1_3) = 6\beta^1_{I^2_1}(C),\\

u^1(N, \varrho^1_{I^1_1}(\beta^{-I^1_1},\tilde\beta), \mu| I^1_1)=5,\\

 u^1(Y, \varrho^1_{I^1_1}(\beta^{-I^1_1},\tilde\beta), \mu| I^1_1)=6\beta^2_{I^1_2}(I)(\beta^3_{I^1_3}(A)\tilde\beta^1_{I^2_1}(D)
+ \beta^3_{I^1_3}(B)\tilde\beta^1_{I^2_1}(C))+7\beta^2_{I^1_2}(O),\\

u^1(C, \varrho^1_{I^2_1}(\beta^{-I^2_1},\tilde\beta), \mu| I^2_1)=6\mu^1_{I^2_1}(\langle Y,I,B\rangle),\; u^1(D, \varrho^1_{I^2_1}(\beta^{-I^2_1},\tilde\beta), \mu| I^2_1)=6\mu^1_{I^2_1}(\langle Y,I,A\rangle),\\

u^2(I, \varrho^2_{I^1_2}(\beta^{-I^1_2},\tilde\beta), \mu| I^1_2) = 9(\beta^3_{I^1_3}(A)\beta^1_{I^2_1}(D)
+ \beta^3_{I^1_3}(B)\beta^1_{I^2_1}(C)),\\

u^2(O, \varrho^2_{I^1_2}(\beta^{-I^1_2},\tilde\beta), \mu| I^1_2) = 5,\\

u^3(A, \varrho^3_{I^1_3}(\beta^{-I^1_3},\tilde\beta), \mu| I^1_3) = 6\beta^1_{I^2_1}(D),\; u^3(B, \varrho^3_{I^1_3}(\beta^{-I^1_3},\tilde\beta), \mu| I^1_3) = 6\beta^1_{I^2_1}(C).
\end{array}\]}

\noindent {\bf Case (1)}. Suppose that $u^1((C, \varrho^1_{I^2_1}(\beta^{-I^2_1},\tilde\beta))\Diamond I^2_1)>u^1((D, \varrho^1_{I^2_1}(\beta^{-I^2_1},\tilde\beta))\Diamond I^2_1)$. Then, $\beta^1_{I^2_1}(D)=0$ and $\beta^3_{I^1_3}(B)>\beta^3_{I^1_3}(A)$. Thus, $u^3((B, \varrho^3_{I^1_3}(\beta^{-I^1_3},\tilde\beta))\Diamond I^1_3)>u^3((A, \varrho^3_{I^1_3}(\beta^{-I^1_3},\tilde\beta))\Diamond I^1_3)$ and consequently, $\beta^3_{I^1_3}(A)=0$. Therefore, $u^1(C, \varrho^1_{I^2_1}(\beta^{-I^2_1},\tilde\beta), \mu| I^2_1)>u^1(D, \varrho^1_{I^2_1}(\beta^{-I^2_1},\tilde\beta), \mu| I^2_1)$ and $u^2((I, \varrho^2_{I^1_2}(\beta^{-I^1_2}$, $\tilde\beta))\Diamond I^1_2)>u^2((O, \varrho^2_{I^1_2}(\beta^{-I^1_2},\tilde\beta))\Diamond I^1_2)$, which lead to $\tilde\beta^1_{I^2_1}(D)=0$ and 
$\beta^2_{I^1_2}(O)=0$. Hence, $u^1((Y, \varrho^1_{I^1_1}(\beta^{-I^1_1}$, $\tilde\beta))\Diamond I^1_1)>u^1((N, \varrho^1_{I^1_1}(\beta^{-I^1_1},\tilde\beta))\Diamond I^1_1)$ and accordingly, $\beta^1_{I^1_1}(N)=0$. The game has a subgame perfect equilibrium given by $(Y, C, I, B)$. 
\newline
{\bf Case (2)}. Suppose that $u^1((D, \varrho^1_{I^2_1}(\beta^{-I^2_1},\tilde\beta))\Diamond I^2_1)>u^1((C, \varrho^1_{I^2_1}(\beta^{-I^2_1},\tilde\beta))\Diamond I^2_1)$. Then, $\beta^1_{I^2_1}(C)=0$ and $\beta^3_{I^1_3}(A)>\beta^3_{I^1_3}(B)$. Thus, $u^3((A, \varrho^3_{I^1_3}(\beta^{-I^1_3},\tilde\beta))\Diamond I^1_3)>u^3((B, \varrho^3_{I^1_3}(\beta^{-I^1_3},\tilde\beta))\Diamond I^1_3)$ and consequently, $\beta^3_{I^1_3}(B)=0$. Therefore, $u^1(D, \varrho^1_{I^2_1}(\beta^{-I^2_1},\tilde\beta), \mu| I^2_1)>u^1(C, \varrho^1_{I^2_1}(\beta^{-I^2_1},\tilde\beta), \mu| I^2_1)$ and $u^2((I, \varrho^2_{I^1_2}(\beta^{-I^1_2}$, $\tilde\beta))\Diamond I^1_2)>u^2((O, \varrho^2_{I^1_2}(\beta^{-I^1_2},\tilde\beta))\Diamond I^1_2)$, which lead to $\tilde\beta^1_{I^2_1}(A)=0$ and 
$\beta^2_{I^1_2}(O)=0$. Hence, $u^1((Y, \varrho^1_{I^1_1}(\beta^{-I^1_1}$, $\tilde\beta))\Diamond I^1_1)>u^1((N, \varrho^1_{I^1_1}(\beta^{-I^1_1},\tilde\beta))\Diamond I^1_1)$ and accordingly, $\beta^1_{I^1_1}(N)=0$. The game has a subgame perfect equilibrium given by $(Y, D, I, A)$. 
\newline
{\bf Case (3)}. Suppose that $u^1((C, \varrho^1_{I^2_1}(\beta^{-I^2_1},\tilde\beta))\Diamond I^2_1)=u^1((D, \varrho^1_{I^2_1}(\beta^{-I^2_1},\tilde\beta))\Diamond I^2_1)$. Then, $\beta^3_{I^1_3}(A)=\beta^3_{I^1_3}(B)$. Thus, $u^3((A, \varrho^3_{I^1_3}(\beta^{-I^1_3},\tilde\beta))\Diamond I^1_3)=u^3((B, \varrho^3_{I^1_3}(\beta^{-I^1_3},\tilde\beta))\Diamond I^1_3)$ and consequently, $\beta^1_{I^2_1}(C)=\beta^1_{I^2_1}(D)$. Therefore, $u^2((O, \varrho^2_{I^1_2}(\beta^{-I^1_2},\tilde\beta))\Diamond I^1_2)>u^2((I, \varrho^2_{I^1_2}(\beta^{-I^1_2},\tilde\beta))\Diamond I^1_2)$, which leads to  
$\beta^2_{I^1_2}(I)=0$. Hence, $u^1((Y, \varrho^1_{I^1_1}(\beta^{-I^1_1},\tilde\beta))\Diamond I^1_1)>u^1((N, \varrho^1_{I^1_1}(\beta^{-I^1_1},\tilde\beta))\Diamond I^1_1)$ and accordingly, $\beta^1_{I^1_1}(N)=0$. The game has a subgame perfect equilibrium given by $(Y, (\frac{1}{2},\frac{1}{2}), O, (\frac{1}{2},\frac{1}{2}))$. 

The cases (1)-(3) together bring us that the game has three subgame perfect equilibria given by $(Y, C, I, B)$, $(Y, D, I, A)$,  and $(Y, (\frac{1}{2},\frac{1}{2}), O, (\frac{1}{2},\frac{1}{2}))$. 
}\end{example}
}

\section{\large Semi-Sequential Equilibrium and Subgame Perfect Semi-Sequential Equilibrium}

The characterization of NashEBS consists of two behavioral strategy profiles, the original behavioral strategy profile and an extra behavioral strategy profile. When these two behavioral strategy profiles become identical, we attain a strict refinement of NashEBS, which is named as semi-sequential equilibrium.
\begin{definition}[\bf Semi-Sequential Equilibrium]
{\em An assessment $(\beta^*,\mu^*)$ is a semi-sequential equilibrium if $\beta^{*i}_{I^j_i}(a')=0$ for any $i\in N$, $j\in M_i$, and $a',a''\in A(I^j_i)$ with $u^i(a'', \beta^{*-I^j_i}, \mu^*|I^j_i)>u^i(a', \beta^{*-I^j_i}, \mu^*|I^j_i)$, where $\mu^*$ is a solution to the system~(\ref{bsne1}).
}
\end{definition}
When two behavioral strategy profiles in the characterization of subgame perfect equilibrium become identical, we get a strict refinement of subgame perfect equilibrium, which is named as subgame perfect semi-sequential equilibrium. 
\begin{definition}[\bf Subgame Perfect Semi-Sequential Equilibrium]
{\em An assessment $(\beta^*,\mu^*)$ is a subgame perfect semi-sequential equilibrium if $\beta^{*i}_{I^j_i}(a')=0$ for any $i\in N$, $j\in M_i$, and $a',a''\in A(I^j_i)$ with $u^i(a'', \beta^{*-I^j_i}, \mu^*|I^j_i)>u^i(a', \beta^{*-I^j_i}, \mu^*|I^j_i)$, where $\mu^*$ is a solution to the system~(\ref{sgbseq1}).
}
\end{definition}
One can show that there always exists a semi-sequential equilibrium and a subgame perfect semi-sequential equilibrium in a finite extensive-form game with perfect recall. 

\begin{figure}[H]
    \centering
    \begin{minipage}{0.49\textwidth}
        \centering
        \includegraphics[width=0.80\textwidth, height=0.12\textheight]{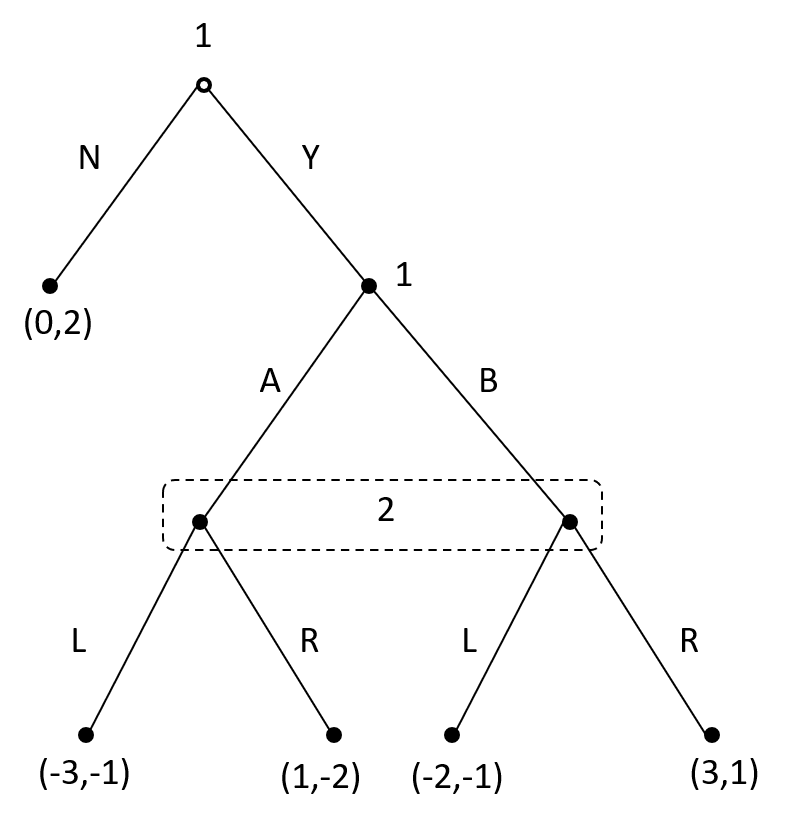}
        \caption{\label{TFig3}\scriptsize An Extensive-Form Game from Mas-Colell et al.~\cite{Mas-Colell et al. (1995)}}
\end{minipage}\hfill
    \begin{minipage}{0.49\textwidth}
        \centering
        \includegraphics[width=0.80\textwidth, height=0.12\textheight]{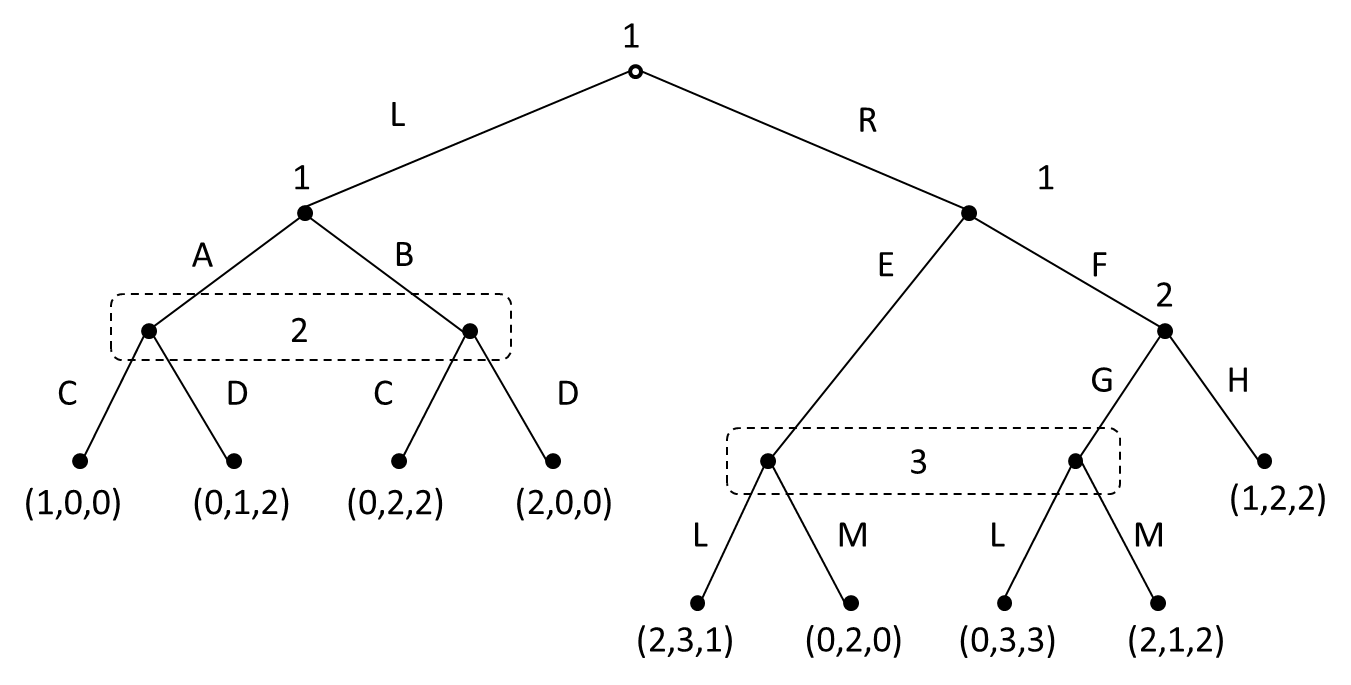}
        \caption{\label{TFig4}\scriptsize An Extensive-Form Game from Bonanno~\cite{Bonanno (2018)}}
\end{minipage}
  \end{figure}

Consider the game in Fig.~\ref{TFig3}, which has two classes of NashEBS given by (1). $(Y, B, R)$; and (2). $(N, (\beta^1_{I^2_1}(A),1- \beta^1_{I^2_1}(A)), (\beta^2_{I^1_2}(L), 1-\beta^2_{I^1_2}(L)))$ with $0\le \beta^1_{I^2_1}(A)\le 1$ and $\beta^2_{I^1_2}(L)\ge\frac{3}{5}$. The game has three classes of semi-sequential equilibria given by (1).
$(Y, B, R)$ with $\mu^2_{I^1_2}(\langle Y, A\rangle)=0$; (2). $(N, B, L)$ with $\mu^2_{I^1_2}(\langle Y, A\rangle)>\frac{2}{3}$; and
(3). $(N, B, (\beta^2_{I^1_2}(L),1-\beta^2_{I^1_2}(L)))$ with $\beta^2_{I^1_2}(L)\ge \frac{3}{5}$ and $\mu^2_{I^1_2}(\langle Y, A\rangle)=\frac{2}{3}$. Comparing the set of NashEBSs with the set of semi-sequential equilibria for the game in Fig.~\ref{TFig3}, one can draw the conclusion that semi-sequential equilibrium is indeed a strict refinement of NashEBS.

 Consider the game in Fig.~\ref{TFig4}, which has infinitely many subgame perfect equilibria given by $(R, (\frac{2}{3},\frac{1}{3}), E, (\frac{2}{3},\frac{1}{3}), (\beta^1_{I^1_1}(G), 1-\beta^1_{I^1_1}(G)), L)$ with $0\le \beta^1_{I^1_1}(G)\le 1$. The game has a unique subgame perfect semi-sequential equilibrium given by
 $(R, (\frac{2}{3},\frac{1}{3}), E, (\frac{2}{3},\frac{1}{3}), G, L)$ with $\mu^3_{I^1_3}(\langle R, E\rangle)=1$.  Comparing the set of NashEBSs with the set of semi-sequential equilibria for the game in Fig.~\ref{TFig4}, one can draw the conclusion that subgame perfect semi-sequential equilibrium is indeed a strict refinement of subgame perfect equilibrium.

\section{\large Differentiable Path-Following Methods for Computing a Nash Equilibrium and a Subgame Perfect Equilibrium\label{dpm}}

To further demonstrate the applications of Definition~\ref{edne1} and Definition~\ref{edsgpe1}, we will exploit in this section the systems~(\ref{nsrne1}) and~(\ref{nscsgpe1}) to develop differentiable path-following methods to compute a NashEBS and a subgame perfect equilibrium.\footnote{A general framework for establishing such a differentiable path-following method can be described as follows.
Step 1: Constitute with an extra variable $t\in (0,1]$ an artificial extensive-form game $\Gamma(t)$ in which each player at each of his information sets solves a convex optimization problem. The artificial game should continuously deform from a trivial game to the target game as $t$ descends from one to zero. $\Gamma(1)$ should have a unique equilibrium, which can be easily computed, and every convergent sequence of equilibria of $\Gamma(t_k)$, $k=1,2,\ldots$, with $\lim\limits_{k\to\infty}t_k=0$ should yield a desired equilibrium at its limit.
Step 2: Apply the optimality conditions to the convex optimization problems in the artificial game to acquire from the equilibrium condition an equilibrium system.
Step 3: Verify that the closure of the set of solutions of the equilibrium system contains a path-connected component that intersects both the levels of $t=1$ and $t=0$.  
Step 4: Ensure through an application of the Transversality Theorem in Eaves and Schmedders~\cite{Eaves and Schmedders (1999)} the existence of a smooth path that starts from the unique equilibrium at $t=1$ and approaches a desired equilibrium as $t\to 0$.
Step 5: Adopt a standard predictor-corrector method for numerically tracing the smooth path to a desired equilibrium.
} 
 Let $\eta^0=(\eta^{0i}_{I^j_i}:i\in N,j\in M_i)$  be a given vector with $\eta^{0i}_{I^j_i}=(\eta^{0i}_{I^j_i}(a):a\in A(I^j_i))^\top$ such that $0<\eta^{0i}_{I^j_i}(a)$ and $\tau^i_{I^j_i}(\eta^0)=\sum\limits_{a\in A(I^j_i)}\eta^{0i}_{I^j_i}(a)< 1$. For $t\in [0,1]$,
let $\varpi(\beta,t)=(\varpi(\beta^q_{I^l_q},t):q\in N,l\in M_q)$ with $\varpi(\beta^q_{I^l_q},t)=(\varpi(\beta^q_{I^l_q}(a),t):a\in A(I^l_q))^\top$, where \[\setlength{\abovedisplayskip}{1.2pt}
\setlength{\belowdisplayskip}{1.2pt}\varpi(\beta^q_{I^l_q}(a),t)=(1-t^2(1-t^2)\tau^q_{I^l_q}(\eta^0))\beta^q_{I^l_q}(a)+t^2(1-t^2)\eta^{0q}_{I^l_q}(a).\]
Let $\beta^0$ and $\tilde\beta^0$ be two given totally mixed behavioral strategy profiles. Let $\xi^0=(\xi^{0i}_{I^j_i}(h):i\in N, j\in M_i,h\in I^j_i)$ with $\xi^{0i}_{I^j_i}(h)={\cal S}^i(h|\beta^0)/{\cal S}^i(I^j_i|\beta^0)$ and $\mu^0=\xi^0$. These vectors will be employed to characterize the starting point of smooth paths in the following developments. 

\subsection{Logarithmic-Barrier Smooth Paths}
For $t\in (0,1]$, we constitute with $\varpi(\beta,t)$ a logarithmic-barrier extensive-form game $\Gamma_L(t)$ in which player $i$ at his information set $I^j_i$ solves against a given $(\hat\beta,\hat{\tilde\beta},\hat\mu)$ the strictly convex optimization problem,
{\footnotesize
\begin{equation}\setlength{\abovedisplayskip}{1.2pt}
\setlength{\belowdisplayskip}{1.2pt}
\label{logbop1}
\begin{array}{rl}
\max\limits_{\beta^i_{I^j_i},\;\tilde\beta^i_{I^j_i},\;\mu^i_{I^j_i},\;\xi^i_{I^j_i}} & (1-t)\sum\limits_{a\in A(I^j_i)}(\beta^i_{I^j_i}(a)u^i((a,\varrho^i_{I^j_i}(\varpi(\hat\beta^{-I^j_i},t),\hat{\tilde\beta}))\land I^j_i)\\
& +\tilde\beta^i_{I^j_i}(a)u^i(a,\varrho^i_{I^j_i}(\varpi(\hat\beta^{-I^j_i},t),\hat{\tilde\beta}),\hat\mu| I^j_i))\\
& +t\sum\limits_{a\in A(I^j_i)}(\beta^{0i}_{I^j_i}(a)\ln\beta^i_{I^j_i}(a)+\tilde\beta^{0i}_{I^j_i}(a)\ln\tilde\beta^i_{I^j_i}(a))+t\sum\limits_{h\in I^j_i}\xi^{0i}_{I^j_i}(h)\ln\xi^i_{I^j_i}(h)\\
\text{s.t.} & \sum\limits_{a\in A(I^j_i)}\beta^i_{I^j_i}(a)=1,\;\sum\limits_{a\in A(I^j_i)}\tilde\beta^i_{I^j_i}(a)=1,\;\sum\limits_{h\in I^j_i}\mu^i_{I^j_i}(h)=1,\\
& ((1-t){\cal S}^i(I^j_i|\varpi(\hat\beta,t)+t)\mu^i_{I^j_i}(h)-\xi^i_{I^j_i}(h)=(1-t){\cal S}^i(h|\varpi(\hat\beta,t)),\;h\in I^j_i.
\end{array}
\end{equation}
}An application of the optimality conditions to the problem~(\ref{logbop1}) together with the equilibrium condition of $(\beta,\tilde\beta,\mu)=(\hat\beta,\hat{\tilde\beta},\hat\mu)$ yields the equilibrium system of $\Gamma_L(t)$,
{\footnotesize
\begin{equation}\setlength{\abovedisplayskip}{1.2pt}
\setlength{\belowdisplayskip}{1.2pt}
\label{Mlogbes1}
\begin{array}{l}
(1-t)u^i((a,\varrho^i_{I^j_i}(\varpi(\beta^{-I^j_i},t),\tilde\beta))\land I^j_i) +t\beta^{0i}_{I^j_i}(a)/\beta^i_{I^j_i}(a)-\zeta^i_{I^j_i}=0,\;i\in N,j\in M_i,a\in A(I^j_i),\\
(1-t)u^i(a,\varrho^i_{I^j_i}(\varpi(\beta^{-I^j_i},t),\tilde\beta),\mu|I^j_i) +t\tilde\beta^{0i}_{I^j_i}(a)/\tilde\beta^i_{I^j_i}(a)-\tilde\zeta^i_{I^j_i}=0,\;i\in N,j\in M_i,a\in A(I^j_i),\\
t\frac{\xi^{0i}_{I^j_i}(h)}{\xi^i_{I^j_i}(h)}((1-t)\omega^i(I^j_i|\varpi(\beta,t))+t)-\sigma^i_{I^j_i}=0,\;i\in N,j\in M_i,h\in I^j_i,\\
  ((1-t){\cal S}^i(I^j_i|\varpi(\beta,t))+t)\mu^i_{I^j_i}(h)-\xi^i_{I^j_i}(h)-(1-t){\cal S}^i(h|\varpi(\beta,t))=0,\;i\in N,j\in M_i,h\in I^j_i,\\
\sum\limits_{a\in A(I^j_i)}\beta^i_{I^j_i}(a)=1,\;\sum\limits_{a\in A(I^j_i)}\tilde\beta^i_{I^j_i}(a)=1,\; \sum\limits_{h\in I^j_i}\mu^i_{I^j_i}(h)=1,\;i\in N,j\in M_i,\\
0<\beta^i_{I^j_i}(a),\;0<\tilde\beta^i_{I^j_i}(a),\;0<\xi^i_{I^j_i}(h),\;i\in N, j\in M_i, a\in A(I^j_i), h\in I^j_i.
\end{array}
\end{equation}}Multiplying $\beta^i_{I^j_i}(a)$, $\tilde\beta^i_{I^j_i}(a)$, and $\xi^i_{I^j_i}(h)$ to the equations in the first, second, and third groups in the system~(\ref{Mlogbes1}), respectively, we come to a polynomial system,  
{\footnotesize
\begin{equation}\setlength{\abovedisplayskip}{1.2pt}
\setlength{\belowdisplayskip}{1.2pt}
\label{Mlogbes2}
\begin{array}{l}
(1-t)\beta^i_{I^j_i}(a)u^i((a,\varrho^i_{I^j_i}(\varpi(\beta^{-I^j_i},t),\tilde\beta))\land I^j_i) +t\beta^{0i}_{I^j_i}(a)-\beta^i_{I^j_i}(a)\zeta^i_{I^j_i}=0,\;i\in N,j\in M_i,a\in A(I^j_i),\\
(1-t)\tilde\beta^i_{I^j_i}(a)u^i(a,\varrho^i_{I^j_i}(\varpi(\beta^{-I^j_i},t),\tilde\beta),\mu|I^j_i) +t\tilde\beta^{0i}_{I^j_i}(a)-\tilde\beta^i_{I^j_i}(a)\tilde\zeta^i_{I^j_i}=0,\;i\in N,j\in M_i,a\in A(I^j_i),\\
t\xi^{0i}_{I^j_i}(h)((1-t){\cal S}^i(I^j_i|\varpi(\beta,t))+t)-\nu^i_{I^j_i}\xi^i_{I^j_i}(h)=0,\;i\in N,j\in M_i,h\in I^j_i,\\
  ((1-t){\cal S}^i(I^j_i|\varpi(\beta,t))+t)\mu^i_{I^j_i}(h)-\xi^i_{I^j_i}(h)-(1-t){\cal S}^i(h|\varpi(\beta,t))=0,\;i\in N,j\in M_i,h\in I^j_i,\\
\sum\limits_{a\in A(I^j_i)}\beta^i_{I^j_i}(a)=1,\;\sum\limits_{a\in A(I^j_i)}\tilde\beta^i_{I^j_i}(a)=1,\; \sum\limits_{h\in I^j_i}\mu^i_{I^j_i}(h)=1,\;i\in N,j\in M_i,\\
0<\beta^i_{I^j_i}(a),\;0<\tilde\beta^i_{I^j_i}(a),\;0<\xi^i_{I^j_i}(h),\;i\in N, j\in M_i, a\in A(I^j_i), h\in I^j_i.
\end{array}
\end{equation}}Taking the sum of equations in the first and second groups of the system~(\ref{Mlogbes2}) over $A(I^j_i)$ and the sum of equations in the third group of the system~(\ref{Mlogbes2}) over $I^j_i$, respectively, we acquire from a division operation the system, {\footnotesize
\begin{equation}\setlength{\abovedisplayskip}{1.2pt}
\setlength{\belowdisplayskip}{1.2pt}
\label{Mlogbes3}
\begin{array}{l}
\zeta^i_{I^j_i}=\frac{1}{\sum\limits_{a'\in A(I^j_i)}\beta^i_{I^j_i}(a')}((1-t)\sum\limits_{a'\in A(I^j_i)}\beta^i_{I^j_i}(a')u^i((a',\varrho^i_{I^j_i}(\varpi(\beta^{-I^j_i},t),\tilde\beta))\land I^j_i) +t),\;i\in N,j\in M_i,\\

\tilde\zeta^i_{I^j_i}=\frac{1}{\sum\limits_{a'\in A(I^j_i)}\tilde\beta^i_{I^j_i}(a')}((1-t)\sum\limits_{a'\in A(I^j_i)}\tilde\beta^i_{I^j_i}(a')u^i(a',\varrho^i_{I^j_i}(\varpi(\beta^{-I^j_i},t),\tilde\beta),\mu|I^j_i) +t),\;i\in N,j\in M_i,\\

\nu^i_{I^j_i}=\frac{1}{\sum\limits_{h'\in I^j_i}\xi^i_{I^j_i}(h')}t((1-t){\cal S}^i(I^j_i|\varpi(\beta,t))+t)\sum\limits_{h'\in I^j_i}\xi^{0i}_{I^j_i}(h'),\;i\in N,j\in M_i.
\end{array}
\end{equation}}Let $a^0_{I^j_i}$ be a given reference action in $A(I^j_i)$ and $h^0_{I^j_i}$ a given reference history of $I^j_i$. Substituting $\zeta^i_{I^j_i}$, $\tilde\zeta^i_{I^j_i}$, and $\nu^i_{I^j_i}$ of the system~(\ref{Mlogbes3}) into the system~(\ref{Mlogbes2}), we arrive at an equivalent system with fewer variables,
{\footnotesize
\begin{equation}\setlength{\abovedisplayskip}{1.2pt}
\setlength{\belowdisplayskip}{1.2pt}
\label{Mlogbes4}
\begin{array}{l}
(1-t)\beta^i_{I^j_i}(a)\sum\limits_{a'\in A(I^j_i)}\beta^i_{I^j_i}(a')(u^i((a,\varrho^i_{I^j_i}(\varpi(\beta^{-I^j_i},t),\tilde\beta))\land I^j_i)-u^i((a',\varrho^i_{I^j_i}(\varpi(\beta^{-I^j_i},t),\tilde\beta))\land I^j_i) )\\
\hspace{4.2cm}+t(\beta^{0i}_{I^j_i}(a)\sum\limits_{a'\in A(I^j_i)}\beta^i_{I^j_i}(a')-\beta^i_{I^j_i}(a))=0,\;
i\in N,j\in M_i,a\in A(I^j_i)\backslash\{a^0_{I^j_i}\},\\
(1-t)\tilde\beta^i_{I^j_i}(a)\sum\limits_{a'\in A(I^j_i)}\tilde\beta^i_{I^j_i}(a')(u^i(a,\varrho^i_{I^j_i}(\varpi(\beta^{-I^j_i},t),\tilde\beta),\mu|I^j_i)-u^i(a',\varrho^i_{I^j_i}(\varpi(\beta^{-I^j_i},t),\tilde\beta),\mu |I^j_i) )\\
\hspace{4.2cm}+t(\tilde\beta^{0i}_{I^j_i}(a)\sum\limits_{a'\in A(I^j_i)}\tilde\beta^i_{I^j_i}(a')-\tilde\beta^i_{I^j_i}(a))=0,\;i\in N,j\in M_i,a\in A(I^j_i)\backslash\{a^0_{I^j_i}\},\\
\xi^{0i}_{I^j_i}(h)\sum\limits_{h'\in I^j_i}\xi^i_{I^j_i}(h')-\xi^i_{I^j_i}(h)\sum\limits_{h'\in I^j_i}\xi^{0i}_{I^j_i}(h')=0,\;i\in N,j\in M_i,h\in I^j_i\backslash\{h^0_{I^j_i}\},\\
  ((1-t){\cal S}^i(I^j_i|\varpi(\beta,t))+t)\mu^i_{I^j_i}(h)-\xi^i_{I^j_i}(h)-(1-t){\cal S}^i(h|\varpi(\beta,t))=0,\;i\in N,j\in M_i,h\in I^j_i,\\
\sum\limits_{a\in A(I^j_i)}\beta^i_{I^j_i}(a)=1,\;\sum\limits_{a\in A(I^j_i)}\tilde\beta^i_{I^j_i}(a)=1,\; \sum\limits_{h\in I^j_i}\mu^i_{I^j_i}(h)=1,\;i\in N,j\in M_i,\\
0<\beta^i_{I^j_i}(a),\;0<\tilde\beta^i_{I^j_i}(a),\;0<\xi^i_{I^j_i}(h),\;i\in N, j\in M_i, a\in A(I^j_i), h\in I^j_i.
\end{array}
\end{equation}}

\noindent When $t=1$, the system~(\ref{Mlogbes4}) has a unique solution given by $(\beta^*(1), \tilde\beta^*(1),\mu^*(1), \xi^*(1))$ with $\beta^{*i}_{I^j_i}(1; a)=\beta^{0i}_{I^j_i}(a)$, $\tilde\beta^{*i}_{I^j_i}(1; a)=\tilde\beta^{0i}_{I^j_i}(a)$, $\mu^{*i}_{I^j_i}(1;h)=\mu^{0i}_{I^j_i}(h)$, and $\xi^{*i}_{I^j_i}(1;h)=\xi^{0i}_{I^j_i}(h)$.

Let $\widetilde{\mathscr{S}}_L$ be the set of all $(\beta,\tilde\beta,\mu,\xi, t)$ satisfying the system~(\ref{Mlogbes4}) with $t>0$ and $\mathscr{S}_L$ the closure of $\widetilde{\mathscr{S}}_L$. One can get that $\mathscr{S}_L$ is a nonempty compact set. We denote by $\{(\beta^k,\tilde\beta^k,\mu^k,\xi^k,  t_k)\in\widetilde{\mathscr{S}}_L,k=1,2,\ldots\}$ a convergent sequence with $0<t_k\le 1$ and $(\beta^*,\tilde\beta^*,\mu^*,\xi^*, 0)=\lim\limits_{k\to\infty}(\beta^k,\tilde\beta^k,\mu^k,\xi^k,t_k)$. Taking the sum of equations in the fourth group of the system~(\ref{Mlogbes4}), we have $\sum\limits_{h\in I^j_i}\xi^i_{I^j_i}(h)=t$ due to the result of $\sum\limits_{h\in I^j_i}\mu^i_{I^j_i}(h)=1$. Then, $\lim\limits_{k\to\infty}\sum\limits_{h\in I^j_i}\xi^{ki}_{I^j_i}(h)=\lim\limits_{k\to\infty}\theta(t_k)=0$. Thus, $\xi^*=0$. Therefore,
as $k\to\infty$, since the system~(\ref{Mlogbes4}) is a polynomial system, we conclude that $(\beta^*,\tilde\beta^*,\mu^*)$ satisfies the system~(\ref{nsrne1}). Hence it follows from Theorem~\ref{nscthm1} that $\beta^*$ is 
 a Nash equilibrium.

Let $\widetilde{\mathscr{E}}_L$ denote the set of all $(\beta,\tilde\beta,\mu,t)$ satisfying the system~(\ref{Mlogbes4}) with $t>0$ and $\mathscr{E}_L$ the closure of $\widetilde{\mathscr{E}}_L$. 
An application of a well-known fixed point theorem in Mas-Colell~\cite{Mas-Colell (1974)} (see also, Herings~\cite{Herings (2000)}) shows that  $\mathscr{E}_L$ contains a unique connected component $\mathscr{E}_L^c$ such that $\mathscr{E}_L^c\cap(\triangle\times\triangle\times\Xi\times\{0\})\ne\emptyset$ and $\mathscr{E}_L^c\cap(\triangle\times\triangle\times\Xi\times\{1\})\ne\emptyset$.
Let $m_0=\sum\limits_{i\in N}\sum\limits_{j\in M_i}|A(I^j_i)|$ and $p_0=\sum\limits_{i\in N,\;j\in M_i}|I^j_i|$. 
Let $g_0(\beta,\tilde\beta,\mu,\xi,t)$ denote the left-hand sides of equations in the system~(\ref{Mlogbes4}). Subtracting a perturbation term of $t(1-t)\alpha\in\mathbb{R}^{2m_0+2p_0}$ from $g_0(\beta,\tilde\beta,\mu,\xi,t)$, we arrive at the system,
\(g_0(\beta,\tilde\beta,\mu,\xi,t)-t(1-t)\alpha=0.\)
Let $g(\beta,\tilde\beta,\mu,\xi,t;\alpha)=g_0(\beta,\tilde\beta,\mu,\xi,t)-t(1-t)\alpha$. For any given $\alpha\in\mathbb{R}^{2m_0+2p_0}$, we denote $g_\alpha(\beta,\tilde\beta,\mu,\xi,t)=g(\beta,\tilde\beta,\mu,\xi,t;\alpha)$.
Let $\tilde {\cal G}_\alpha=\{(\beta,\tilde\beta,\mu,\xi,t)|g_\alpha(\beta,\tilde\beta,\mu,\xi,t)=0\text{ with }0< t\le 1\}$ and ${\cal G}_\alpha$ the closure of $\tilde {\cal G}_\alpha$.
One can obtain that ${\cal G}_\alpha$ is a compact set and $g(\beta,\tilde\beta,\mu,\xi,t;\alpha)$ is continuously differentiable on $\mathbb{R}^{2m_0}\times\mathbb{R}^{2p_0}\times (0,1)\times \mathbb{R}^{2m_0+2p_0}$ with $D_{\alpha}g(\beta,\tilde\beta,\mu,\xi,t;\alpha)=t(1-t)I_{2m_0+2p_0}$, where $I_{2m_0+2p_0}$ is an identity matrix. It is easy to see that, as $0<t<1$, $D_{\alpha}g(\beta,\tilde\beta,\mu,\xi,t;\alpha)$ is nonsingular. Furthermore, when $t=1$, $D_{\beta,\tilde\beta,\mu,\xi}g_0(\beta,\tilde\beta,\mu,\xi,1)$ is nonsingular. When $\|\alpha\|$ is sufficiently small,  the continuity of $g(\beta,\tilde\beta,\mu,\xi,t;\alpha)$ ensures us that there is a unique connected component in ${\cal G}_\alpha$ intersecting both $\mathbb{R}^{2m_0}\times\mathbb{R}^{2p_0}\times\{1\}$ and $\mathbb{R}^{2m_0}\times\mathbb{R}^{2p_0}\times\{0\}$.
These results together with the Transversality Theorem in Eaves and Schmedders~\cite{Eaves and Schmedders (1999)} lead us to the following conclusion.
 For generic choice of $\alpha$ with sufficiently small $\|\alpha\|$, there exists a smooth path  $P_{\alpha}\subseteq {\cal G}_{\alpha}$ that starts from the unique solution $(\beta^*(1),\tilde\beta^*(1),\mu^*(1),\xi^*(1),1)$ on the level of $t=1$ and ends at a Nash equilibrium of $\Gamma$ on the target level of $t=0$.

As an application of Definition~\ref{edsgpe1}, we secure in a similar way to the above developments an equilibrium system that specifies a logarithmic-barrier smooth path to a subgame perfect equilibrium, which is as follows, 
 {\footnotesize
\begin{equation}\setlength{\abovedisplayskip}{1.2pt}
\setlength{\belowdisplayskip}{1.2pt}
\label{Msglogbes1}
\begin{array}{l}
(1-t)\beta^i_{I^j_i}(a)\sum\limits_{a'\in A(I^j_i)}\beta^i_{I^j_i}(a')(u^i((a,\varrho^i_{I^j_i}(\varpi(\beta^{-I^j_i},t),\tilde\beta))\Diamond I^j_i)\\
\hspace{0.8cm}-\sum\limits_{a'\in A(I^j_i)}\beta^i_{I^j_i}(a')u^i((a',\varrho^i_{I^j_i}(\varpi(\beta^{-I^j_i},t),\tilde\beta))\Diamond I^j_i) )+t(\beta^{0i}_{I^j_i}(a)\sum\limits_{a'\in A(I^j_i)}\beta^i_{I^j_i}(a')-\beta^i_{I^j_i}(a))=0,\\
\hspace{10.4cm}i\in N,j\in M_i,a\in A(I^j_i)\backslash\{a^0_{I^j_i}\},\\
(1-t)\tilde\beta^i_{I^j_i}(a)\sum\limits_{a'\in A(I^j_i)}\tilde\beta^i_{I^j_i}(a')(u^i(a,\varrho^i_{I^j_i}(\varpi(\beta^{-I^j_i},t),\tilde\beta),\mu|I^j_i)\\
\hspace{0.8cm}-\sum\limits_{a'\in A(I^j_i)}\tilde\beta^i_{I^j_i}(a')u^i(a',\varrho^i_{I^j_i}(\varpi(\beta^{-I^j_i},t),\tilde\beta),\mu |I^j_i))+t(\tilde\beta^{0i}_{I^j_i}(a)\sum\limits_{a'\in A(I^j_i)}\tilde\beta^i_{I^j_i}(a')-\tilde\beta^i_{I^j_i}(a))=0,\\
\hspace{10.4cm}i\in N,j\in M_i,a\in A(I^j_i)\backslash\{a^0_{I^j_i}\},\\
\xi^{0i}_{I^j_i}(h)\sum\limits_{h'\in I^j_i}\xi^i_{I^j_i}(h')-\xi^i_{I^j_i}(h)\sum\limits_{h'\in I^j_i}\xi^{0i}_{I^j_i}(h')=0,\;i\in N,j\in M_i,h\in I^j_i\backslash\{h^0_{I^j_i}\},\\
  ((1-t){\cal Y}^i(I^j_i|\varpi(\beta,t))+t)\mu^i_{I^j_i}(h)-\xi^i_{I^j_i}(h)-(1-t){\cal Y}^i(h|\varpi(\beta,t))=0,\;i\in N,j\in M_i,h\in I^j_i,\\
\sum\limits_{a\in A(I^j_i)}\beta^i_{I^j_i}(a)=1,\;\sum\limits_{a\in A(I^j_i)}\tilde\beta^i_{I^j_i}(a)=1,\; \sum\limits_{h\in I^j_i}\mu^i_{I^j_i}(h)=1,\;i\in N,j\in M_i,\\
0<\beta^i_{I^j_i}(a),\;0<\tilde\beta^i_{I^j_i}(a),\;0<\xi^i_{I^j_i}(h),\;i\in N, j\in M_i, a\in A(I^j_i), h\in I^j_i.
\end{array}
\end{equation}
}

One can adopt a standard predictor-corrector method as outlined in Eaves and Schmedders~\cite{Eaves and Schmedders (1999)}  for numerically tracing the smooth paths specified by the system~(\ref{Mlogbes4}) and the system~(\ref{Msglogbes1}) to a NashEBS and a subgame perfect equilibrium, respectively. The smooth paths specified by the system~(\ref{Mlogbes4}) and the system~(\ref{Msglogbes1}) for the games in Figs.~\ref{fexm1}-\ref{fexm2} are given in Figs.~\ref{fexm1MlogbA}-\ref{fexm2sgMlogbB}.
 \begin{figure}[H]
    \centering
    \begin{minipage}{0.49\textwidth}
        \centering
        \includegraphics[width=0.90\textwidth, height=0.12\textheight]{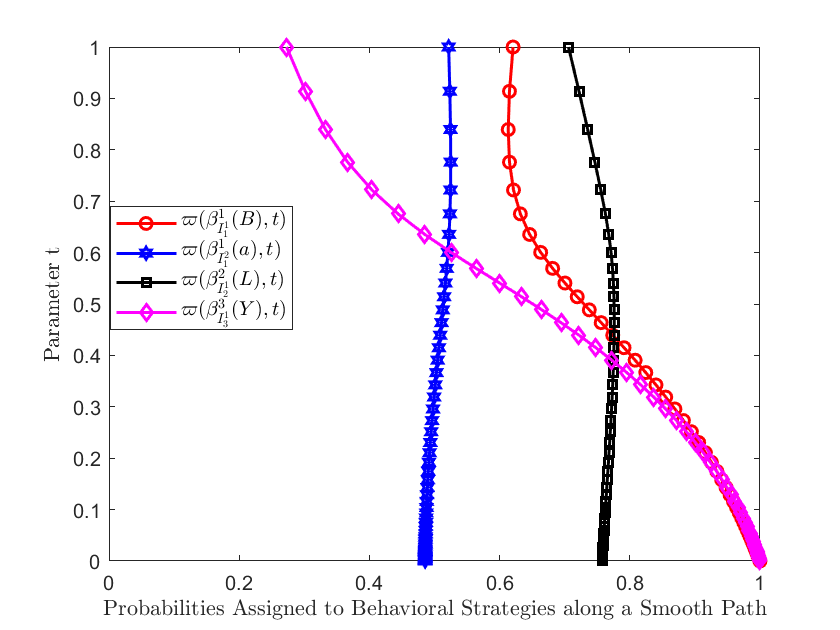}
        \caption{\label{fexm1MlogbA}\scriptsize The Smooth Path of $\varpi(\beta,t)$ Specified by the System~(\ref{Mlogbes4}) for the Game in Fig.~\ref{fexm1}}
\end{minipage}\hfill
    \begin{minipage}{0.49\textwidth}
        \centering
        \includegraphics[width=0.90\textwidth, height=0.12\textheight]{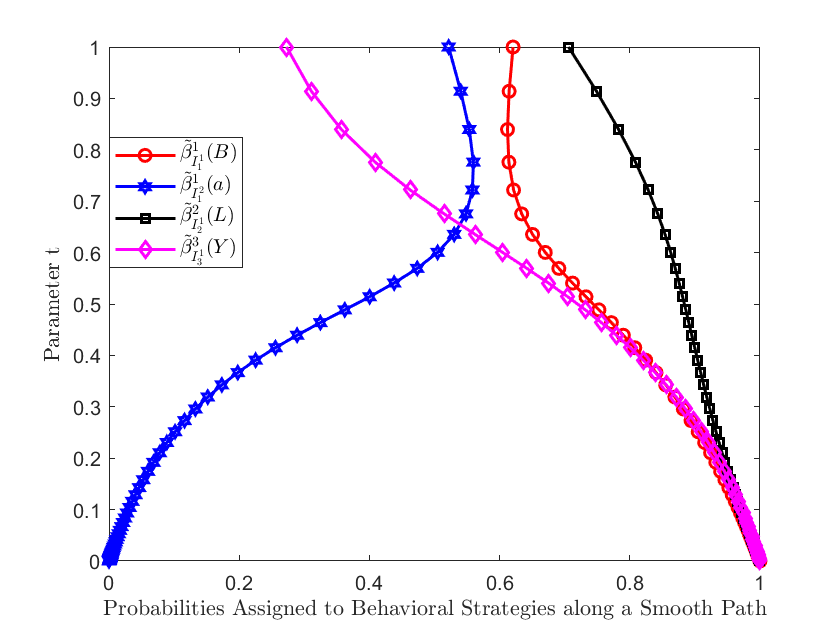}
        \caption{\label{fexm1MlogbB}\scriptsize The Smooth Path of $\tilde\beta$ Specified by the System~(\ref{Mlogbes4}) for the Game in Fig.~\ref{fexm1}}
\end{minipage}
  \end{figure}
  
  \begin{figure}[H]
    \centering
    \begin{minipage}{0.49\textwidth}
        \centering
        \includegraphics[width=0.90\textwidth, height=0.12\textheight]{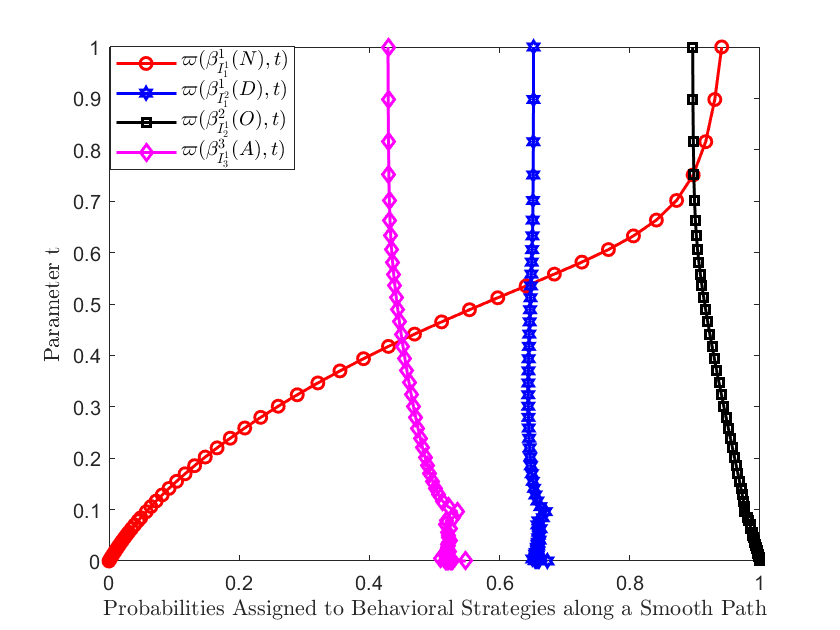}
        \caption{\label{fexm2MlogbA}\scriptsize The Smooth Path of $\varpi(\beta,t)$ Specified by the System~(\ref{Mlogbes4}) for the Game in Fig.~\ref{fexm2}}
               \end{minipage}\hfill
    \begin{minipage}{0.49\textwidth}
        \centering
        \includegraphics[width=0.90\textwidth, height=0.12\textheight]{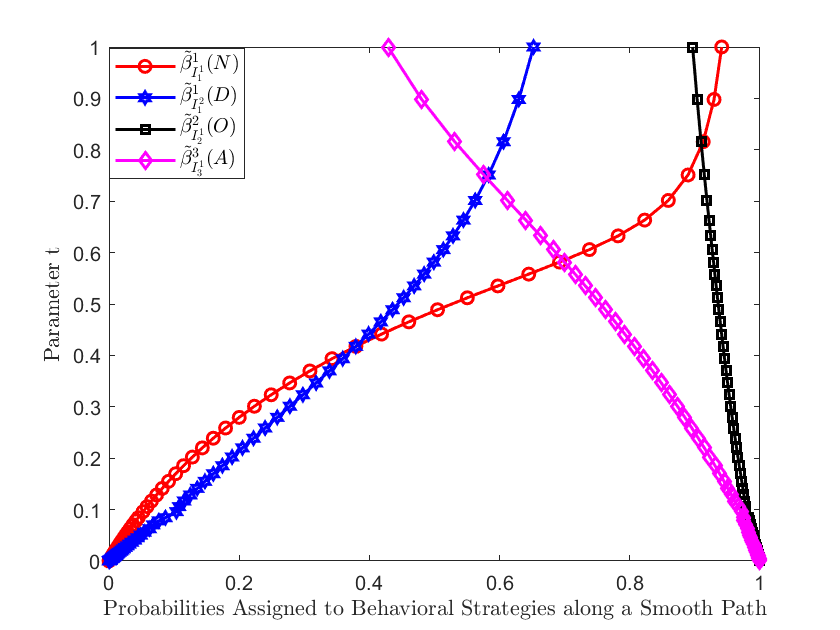}
        \caption{\label{fexm2MlogbB}\scriptsize The Smooth Path of $\tilde\beta$ Specified by the System~(\ref{Mlogbes4}) for the Game in Fig.~\ref{fexm2}}
\end{minipage}
\end{figure}

  \begin{figure}[H]
    \centering
    \begin{minipage}{0.49\textwidth}
        \centering
        \includegraphics[width=0.90\textwidth, height=0.12\textheight]{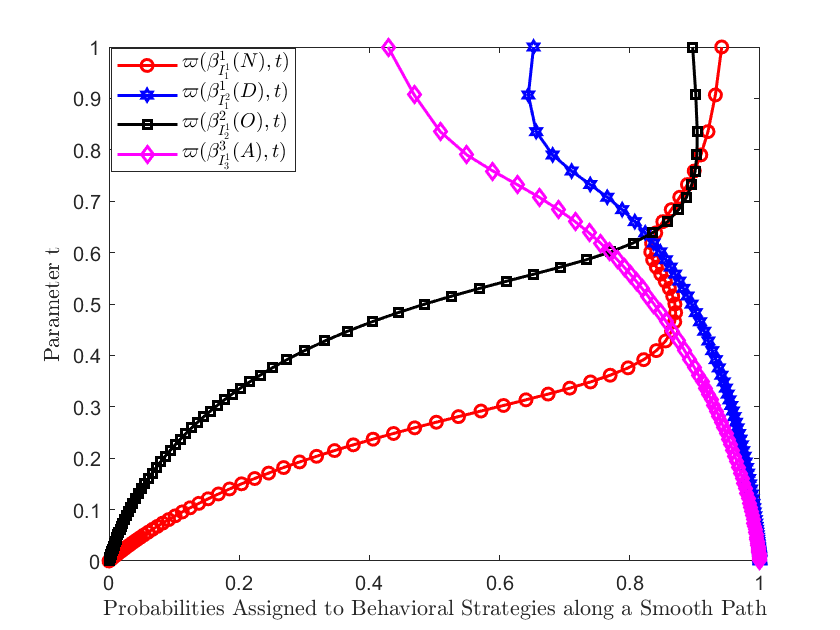}
        \caption{\label{fexm2sgMlogbA}\scriptsize The Smooth Path of $\varpi(\beta,t)$ Specified by the System~(\ref{Msglogbes1}) for the Game in Fig.~\ref{fexm2}}
               \end{minipage}\hfill
    \begin{minipage}{0.49\textwidth}
        \centering
        \includegraphics[width=0.90\textwidth, height=0.12\textheight]{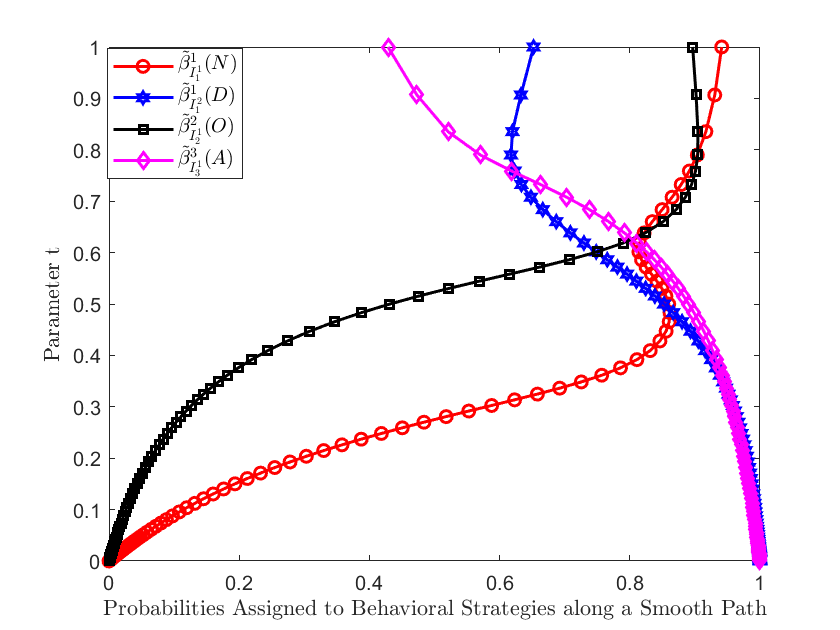}
        \caption{\label{fexm2sgMlogbB}\scriptsize The Smooth Path of $\tilde\beta$ Specified by the System~(\ref{Msglogbes1}) for the Game in Fig.~\ref{fexm2}}
\end{minipage}
\end{figure}

One can observe from Figs.~\ref{fexm1MlogbA} and~\ref{fexm1MlogbB} and from Figs.~\ref{fexm2MlogbA} and~\ref{fexm2MlogbB} distinct trends in the two smooth paths of $\varpi(\beta,t)$ and $\tilde\beta$, which lead to diverse solutions. 
Figs.~\ref{fexm2MlogbA} and~\ref{fexm2sgMlogbA} demonstrate that the paths of $\varpi(\beta,t)$ specified by the system~(\ref{Mlogbes4}) and the system~(\ref{Msglogbes1}) start from the same point but lead to different solutions: The path in Fig.~\ref{fexm2MlogbA} ends at a NashEBS but not subgame perfect, while that in Fig.~\ref{fexm2sgMlogbA} reaches a subgame perfect equilibrium.

\subsection{Convex-Quadratic-Penalty Smooth Paths}

In this subsection, we present alternative schemes to construct smooth paths to a NashEBS and a subgame perfect equilibrium, which can be regarded as exterior-point approaches. For $t\in (0,1]$, we constitute with $\varpi(\beta,t)$ a convex-quadratic-penalty extensive-form game $\Gamma_C(t)$ in which player $i$ at his information set $I^j_i$ solves against a given $(\hat\beta,\hat{\tilde\beta},\hat\mu)$ the strictly convex optimization problem,
{\footnotesize
\begin{equation}\setlength{\abovedisplayskip}{1.2pt}
\setlength{\belowdisplayskip}{1.2pt}
\label{cqupop1}
\begin{array}{rl}
\max\limits_{\beta^i_{I^j_i},\;\tilde\beta^i_{I^j_i},\;\mu^i_{I^j_i},\xi^i_{I^j_i}} & (1-t)\sum\limits_{a\in A(I^j_i)}(\beta^i_{I^j_i}(a)u^i((a,\varrho^i_{I^j_i}(\varpi(\hat\beta^{-I^j_i},t),\hat{\tilde\beta}))\land I^j_i)\\
& +\tilde\beta^i_{I^j_i}(a)u^i(a,\varrho^i_{I^j_i}(\varpi(\hat\beta^{-I^j_i},t),\hat{\tilde\beta}),\hat\mu| I^j_i))\\
& -\frac{1}{2}t\sum\limits_{a\in A(I^j_i)}((\beta^i_{I^j_i}(a)-\beta^{0i}_{I^j_i}(a))^2+(\tilde\beta^i_{I^j_i}(a)-\tilde\beta^{0i}_{I^j_i}(a))^2)
 -\frac{1}{2}t\sum\limits_{h\in I^j_i}(\xi^i_{I^j_i}(h)-\xi^{0i}_{I^j_i}(h))^2\\
\text{s.t.} & \sum\limits_{a\in A(I^j_i)}\beta^i_{I^j_i}(a)=1,\;\sum\limits_{a\in A(I^j_i)}\tilde\beta^i_{I^j_i}(a)=1,\;\sum\limits_{h\in I^j_i}\mu^i_{I^j_i}(h)=1,\\
& ((1-t){\cal S}^i(I^j_i|\varpi(\hat\beta,t))+t)\mu^i_{I^j_i}(h)-\xi^i_{I^j_i}(h)=(1-t){\cal S}^i(h|\varpi(\hat\beta,t)),\;h\in I^j_i,\\
& 0\le \beta^i_{I^j_i}(a),\;0\le\tilde\beta^i_{I^j_i}(a),\;0\le \xi^i_{I^j_i}(h),\;a\in A(I^j_i),h\in I^j_i.
\end{array}
\end{equation}
}An application of the optimality conditions to the problem~(\ref{cqupop1}) together with the equilibrium condition of $(\beta,\tilde\beta,\mu)=(\hat\beta,\hat{\tilde\beta},\hat\mu)$ yields the equilibrium system of $\Gamma_C(t)$,
{\footnotesize\begin{equation}\setlength{\abovedisplayskip}{1.2pt}
\setlength{\belowdisplayskip}{1.2pt}
\label{cqupes1}
\begin{array}{l}
(1-t)u^i((a,\varrho^i_{I^j_i}(\varpi(\beta^{-I^j_i},t),\tilde\beta))\land I^j_i)+ \lambda^i_{I^j_i}(a)-t(\beta^i_{I^j_i}(a)-\beta^{0i}_{I^j_i}(a))-\zeta^i_{I^j_i}=0,\\
\hspace{10.6cm}i\in N,j\in M_i,a\in A(I^j_i),\\
(1-t)u^i(a,\varrho^i_{I^j_i}(\varpi(\beta^{-I^j_i},t),\tilde\beta),\mu|I^j_i)+\tilde\lambda^i_{I^j_i}(a) -t(\tilde\beta^i_{I^j_i}(a)-\tilde\beta^{0i}_{I^j_i}(a))-\tilde\zeta^i_{I^j_i}=0,\\
\hspace{10.6cm}i\in N,j\in M_i,a\in A(I^j_i),\\
(\rho^i_{I^j_i}(h)-t(\xi^i_{I^j_i}(h)-\xi^{0i}_{I^j_i}(h)))((1-t){\cal S}^i(I^j_i|\varpi(\beta,t))+t)-\sigma^i_{I^j_i}=0,\;i\in N,j\in M_i,h\in I^j_i,\\
  ((1-t){\cal S}^i(I^j_i|\varpi(\beta,t))+t)\mu^i_{I^j_i}(h)-\xi^i_{I^j_i}(h)-(1-t){\cal S}^i(h|\varpi(\beta,t))=0,\;i\in N,j\in M_i,h\in I^j_i,\\
\sum\limits_{a\in A(I^j_i)}\beta^i_{I^j_i}(a)=1,\;\sum\limits_{a\in A(I^j_i)}\tilde\beta^i_{I^j_i}(a)=1,\; \sum\limits_{h\in I^j_i}\mu^i_{I^j_i}(h)-1=0,\;i\in N,j\in M_i,\\
 \beta^i_{I^j_i}(a)\lambda^i_{I^j_i}(a)=0,\;\tilde\beta^i_{I^j_i}(a)\tilde\lambda^i_{I^j_i}(a)=0,\;\rho^i_{I^j_i}(h)\xi^i_{I^j_i}(h)=0,\;
 0\le \beta^i_{I^j_i}(a),\;0\le\tilde\beta^i_{I^j_i}(a),\;0\le \xi^i_{I^j_i}(h),\\
 
 0\le \lambda^i_{I^j_i}(a),\;0\le\tilde\lambda^i_{I^j_i}(a),\;0\le \rho^i_{I^j_i}(h),\;
 i\in N,j\in M_i,
 a\in A(I^j_i),h\in I^j_i.
\end{array}
\end{equation}}To secure from the system~(\ref{cqupes1}) a smooth path to a Nash equilibrium, we need to eliminate those complementarity equations and inequalities in the system~(\ref{cqupes1}) through the variable transformations outlined in Herings and Peeters~\cite{Herings and Peeters (2001)}.
Let $\phi_1(v)=(\frac{v+\sqrt{v^2}}{2})^2$ and $\phi_2(v)=(\frac{v-\sqrt{v^2}}{2})^2$. Then, $\phi_1(v)\phi_2(v)=0$. Let $\beta^i_{I^j_i}(z;a)=\phi_1(z^i_{I^j_i}(a))$, $\lambda^i_{I^j_i}(z;a)=\phi_2(z^i_{I^j_i}(a))$,  $\tilde\beta^i_{I^j_i}(\tilde z;a)=\phi_1(\tilde z^i_{I^j_i}(a))$,  $\tilde\lambda^i_{I^j_i}(\tilde z;a)=\phi_2(\tilde z^i_{I^j_i}(a))$, $\xi^i_{I^j_i}(w; h)=\phi_1(w^i_{I^j_i}(h))$, and $\rho^i_{I^j_i}(w; h)=\phi_2(w^i_{I^j_i}(h))$. Substituting $\beta^i_{I^j_i}(z;a)$, $\lambda^i_{I^j_i}(z;a)$, $\tilde\beta^i_{I^j_i}(\tilde z;a)$, $\tilde\lambda^i_{I^j_i}(\tilde z;a)$, $\xi^i_{I^j_i}(w; h)$, and $\rho^i_{I^j_i}(w; h)$ into the system~(\ref{cqupes1}) for  $\beta^i_{I^j_i}(a)$, $\lambda^i_{I^j_i}(a)$, $\tilde\beta^i_{I^j_i}(a)$, $\tilde\lambda^i_{I^j_i}(a)$, $\xi^i_{I^j_i}(h)$, and $\rho^i_{I^j_i}(h)$, we come to the system,
{\footnotesize
\begin{equation}\setlength{\abovedisplayskip}{1.2pt}
\setlength{\belowdisplayskip}{1.2pt}
\label{cqupes2a}
\begin{array}{l}
(1-t)u^i((a,\varrho^i_{I^j_i}(\varpi(\beta^{-I^j_i}(z),t),\tilde\beta(\tilde z)))\land I^j_i)+ \lambda^i_{I^j_i}(z;a)-t(\beta^i_{I^j_i}(z;a)-\beta^{0i}_{I^j_i}(a))-\zeta^i_{I^j_i}=0,\\
\hspace{11.6cm}i\in N,j\in M_i,a\in A(I^j_i),\\
(1-t)u^i(a,\varrho^i_{I^j_i}(\varpi(\beta^{-I^j_i}(z),t),\tilde\beta(\tilde z)),\mu|I^j_i)+\tilde\lambda^i_{I^j_i}(\tilde z; a) -t(\tilde\beta^i_{I^j_i}(\tilde z; a)-\tilde\beta^{0i}_{I^j_i}(a))-\tilde\zeta^i_{I^j_i}=0,\\
\hspace{11.6cm}\;i\in N,j\in M_i,a\in A(I^j_i),\\
(\rho^i_{I^j_i}(w; h)-t(\xi^i_{I^j_i}(w; h)-\xi^{0i}_{I^j_i}(h)))((1-t){\cal S}^i(I^j_i|\varpi(\beta(z),t))+t)-\sigma^i_{I^j_i}=0,\;i\in N,j\in M_i,h\in I^j_i,\\
  ((1-t){\cal S}^i(I^j_i|\varpi(\beta(z),t))+t)\mu^i_{I^j_i}(h)-\xi^i_{I^j_i}(w; h)-(1-t){\cal S}^i(h|\varpi(\beta(z),t))=0,\;i\in N,j\in M_i,h\in I^j_i,\\
\sum\limits_{a\in A(I^j_i)}\beta^i_{I^j_i}(z;a)=1,\;\sum\limits_{a\in A(I^j_i)}\tilde\beta^i_{I^j_i}(\tilde z; a)=1,\; \sum\limits_{h\in I^j_i}\mu^i_{I^j_i}(h)-1=0,\;i\in N,j\in M_i.
\end{array}
\end{equation}
}

\noindent Let $a^0_{I^j_i}$ be a given reference action in $A(I^j_i)$ and $h^0_{I^j_i}$ a given reference history in $I^j_i$. As a result of subtractions, we get from the system~(\ref{cqupes2a}) an equivalent system with fewer variables,
{\footnotesize
\begin{equation}\setlength{\abovedisplayskip}{1.2pt}
\setlength{\belowdisplayskip}{1.2pt}
\label{cqupes2}
\begin{array}{l}
(1-t)(u^i((a,\varrho^i_{I^j_i}(\varpi(\beta^{-I^j_i}(z),t),\tilde\beta(\tilde z)))\land I^j_i)-u^i((a^0_{I^j_i},\varrho^i_{I^j_i}(\varpi(\beta^{-I^j_i}(z),t),\tilde\beta(\tilde z)))\land I^j_i))\\
\hspace{3cm}
+ \lambda^i_{I^j_i}(z;a)-\lambda^i_{I^j_i}(z;a^0_{I^j_i})-t(\beta^i_{I^j_i}(z;a)-\beta^i_{I^j_i}(z;a^0_{I^j_i})-(\beta^{0i}_{I^j_i}(a)-\beta^{0i}_{I^j_i}(a^0_{I^j_i}))=0,\\
\hspace{10.6cm}i\in N,j\in M_i,a\in A(I^j_i)\backslash\{a^0_{I^j_i}\},\\

(1-t)(u^i(a,\varrho^i_{I^j_i}(\varpi(\beta^{-I^j_i}(z),t),\tilde\beta(\tilde z)),\mu|I^j_i)-u^i(a^0_{I^j_i},\varrho^i_{I^j_i}(\varpi(\beta^{-I^j_i}(z),t),\tilde\beta(\tilde z)),\mu|I^j_i))
\\
\hspace{3cm}+\tilde\lambda^i_{I^j_i}(\tilde z; a)-\tilde\lambda^i_{I^j_i}(\tilde z; a^0_{I^j_i})
 -t(\tilde\beta^i_{I^j_i}(\tilde z; a)-\tilde\beta^i_{I^j_i}(\tilde z; a^0_{I^j_i})-(\tilde\beta^{0i}_{I^j_i}(a)- \tilde\beta^{0i}_{I^j_i}(a^0_{I^j_i})))=0,\\
\hspace{10.6cm}i\in N,j\in M_i,a\in A(I^j_i)\backslash\{a^0_{I^j_i}\},\\
\rho^i_{I^j_i}(w; h)-\rho^i_{I^j_i}(w; h^0_{I^j_i})-t(\xi^i_{I^j_i}(w; h)-\xi^i_{I^j_i}(w; h^0_{I^j_i})-(\xi^{0i}_{I^j_i}(h)-\xi^{0i}_{I^j_i}(h^0_{I^j_i})))=0,\\
\hspace{10.6cm}i\in N,j\in M_i,h\in I^j_i\backslash\{h^0_{I^j_i}\},\\
  ((1-t)\omega^i(I^j_i|\varpi(\beta(z),t))+t)\mu^i_{I^j_i}(h)-\xi^i_{I^j_i}(w; h)-\omega^i(h|\varpi(\beta(z),t))=0,\; i\in N,j\in M_i,h\in I^j_i,\\
\sum\limits_{a\in A(I^j_i)}\beta^i_{I^j_i}(z;a)=1,\;\sum\limits_{a\in A(I^j_i)}\tilde\beta^i_{I^j_i}(\tilde z; a)=1,\; \sum\limits_{h\in I^j_i}\mu^i_{I^j_i}(h)-1=0,\;i\in N,j\in M_i.
\end{array}
\end{equation}
}When $t=1$, the system~(\ref{cqupes2}) has a unique solution given by $(z^*(1),\tilde z^*(1),\mu^*(1),w^*(1))$ with $z^{*i}_{I^j_i}(1; a)=\sqrt{\beta^{0i}_{I^j_i}(a)}$, $\tilde z^{*i}_{I^j_i}(1; a)=\sqrt{\tilde\beta^{0i}_{I^j_i}(a)}$, $\mu^{*i}_{I^j_i}(1; h)=\mu^{0i}_{I^j_i}(h)$, and $w^{*i}_{I^j_i}(1; h)=\sqrt{\xi^{0i}_{I^j_i}(h)}$.
Following a similar argument as that in Section 4.1, one can attain that the system~(\ref{cqupes2}) specifies a unique smooth path that starts from $(z^*(1),\tilde z^*(1),\mu^*(1),w^*(1))$ on the level of $t=1$ and ends at a Nash equilibrium on the level of $t=0$.

As a result of Definition~\ref{edsgpe1}, we attain in a similar way to the above development an equilibrium system that specifies a convex-quadratic-penalty smooth path to a subgame perfect equilibrium,
{\footnotesize
\begin{equation}\setlength{\abovedisplayskip}{1.2pt}
\setlength{\belowdisplayskip}{1.2pt}
\label{sgcqupes1}
\begin{array}{l}
(1-t)(u^i((a,\varrho^i_{I^j_i}(\varpi(\beta^{-I^j_i}(z),t),\tilde\beta(\tilde z)))\Diamond I^j_i)-u^i((a^0_{I^j_i},\varrho^i_{I^j_i}(\varpi(\beta^{-I^j_i}(z),t),\tilde\beta(\tilde z)))\Diamond I^j_i))\\
\hspace{3cm}
+ \lambda^i_{I^j_i}(z;a)-\lambda^i_{I^j_i}(z;a^0_{I^j_i})-t(\beta^i_{I^j_i}(z;a)-\beta^i_{I^j_i}(z;a^0_{I^j_i})-(\beta^{0i}_{I^j_i}(a)-\beta^{0i}_{I^j_i}(a^0_{I^j_i}))=0,\\
\hspace{10.6cm}i\in N,j\in M_i,a\in A(I^j_i)\backslash\{a^0_{I^j_i}\},\\

(1-t)(u^i(a,\varrho^i_{I^j_i}(\varpi(\beta^{-I^j_i}(z),t),\tilde\beta(\tilde z)),\mu|I^j_i)-u^i(a^0_{I^j_i},\varrho^i_{I^j_i}(\varpi(\beta^{-I^j_i}(z),t),\tilde\beta(\tilde z)),\mu|I^j_i))
\\
\hspace{3cm}+\tilde\lambda^i_{I^j_i}(\tilde z; a)-\tilde\lambda^i_{I^j_i}(\tilde z; a^0_{I^j_i})-t(\tilde\beta^i_{I^j_i}(\tilde z; a)-\tilde\beta^i_{I^j_i}(\tilde z; a^0_{I^j_i})-(\tilde\beta^{0i}_{I^j_i}(a)- \tilde\beta^{0i}_{I^j_i}(a^0_{I^j_i})))=0,\\
\hspace{10.6cm}i\in N,j\in M_i,a\in A(I^j_i)\backslash\{a^0_{I^j_i}\},\\
\rho^i_{I^j_i}(w; h)-\rho^i_{I^j_i}(w; h^0_{I^j_i})-t(\xi^i_{I^j_i}(w; h)-\xi^i_{I^j_i}(w; h^0_{I^j_i})-(\xi^{0i}_{I^j_i}(h)-\xi^{0i}_{I^j_i}(h^0_{I^j_i})))=0,\\
\hspace{10.6cm}i\in N,j\in M_i,h\in I^j_i\backslash\{h^0_{I^j_i}\},\\
  ((1-t){\cal Y}(I^j_i|\varpi(\beta(z),t))+t)\mu^i_{I^j_i}(h)-\xi^i_{I^j_i}(w; h)-(1-t)\varsigma^i(h|\varpi(\beta(z),t))=0,\\
  \hspace{11cm} i\in N,j\in M_i,h\in I^j_i,\\
\sum\limits_{a\in A(I^j_i)}\beta^i_{I^j_i}(z;a)=1,\;\sum\limits_{a\in A(I^j_i)}\tilde\beta^i_{I^j_i}(\tilde z; a)=1,\; \sum\limits_{h\in I^j_i}\mu^i_{I^j_i}(h)-1=0,\;i\in N,j\in M_i.
\end{array}
\end{equation}
}

A standard predictor-corrector method can be utilized to numerically follow the smooth paths specified by the system~(\ref{cqupes2}) and the system~(\ref{sgcqupes1}) to a NashEBS and a subgame perfect equilibrium, respectively. The smooth paths specified by the system~(\ref{cqupes2}) and the system~(\ref{sgcqupes1}) for the games in Figs.~\ref{fexm1}-\ref{fexm2} are given in Figs.~\ref{fexm1cqupA}-\ref{fexm2sgcqupB}.
 \begin{figure}[H]
    \centering
    \begin{minipage}{0.49\textwidth}
        \centering
        \includegraphics[width=0.90\textwidth, height=0.12\textheight]{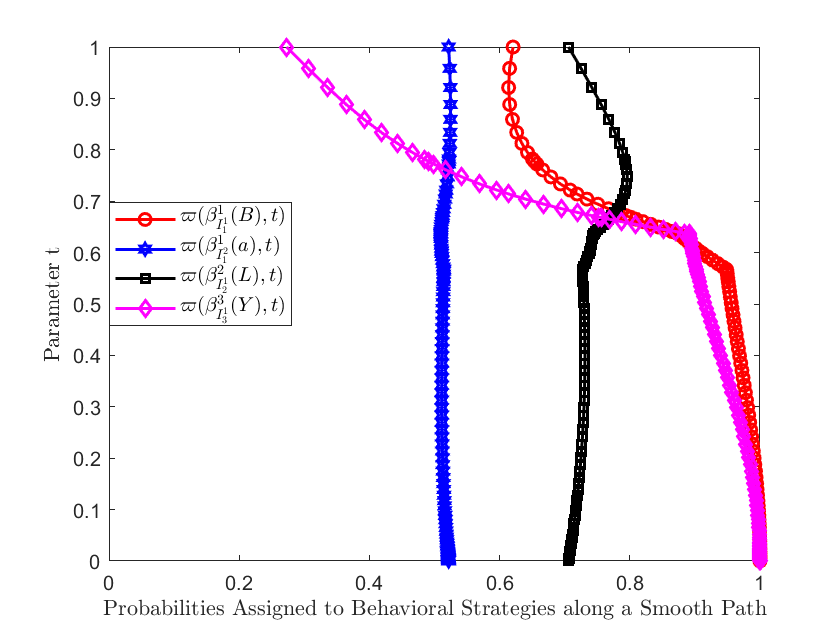}
        \caption{\label{fexm1cqupA}\scriptsize The Smooth Path of $\varpi(\beta,t)$ Specified by the System~(\ref{cqupes2}) for the Game in Fig.~\ref{fexm1}}
\end{minipage}\hfill
    \begin{minipage}{0.49\textwidth}
        \centering
        \includegraphics[width=0.90\textwidth, height=0.12\textheight]{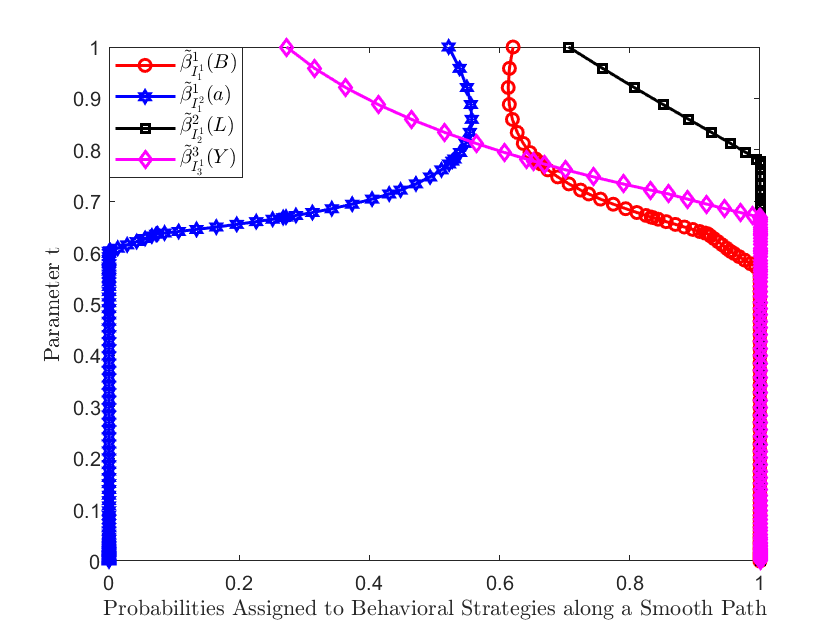}
        \caption{\label{fexm1cqupB}\scriptsize The Smooth Path of $\tilde\beta$ Specified by the System~(\ref{cqupes2}) for the Game in Fig.~\ref{fexm1}}
\end{minipage}
  \end{figure}
  
  \begin{figure}[H]
    \centering
    \begin{minipage}{0.49\textwidth}
        \centering
        \includegraphics[width=0.90\textwidth, height=0.12\textheight]{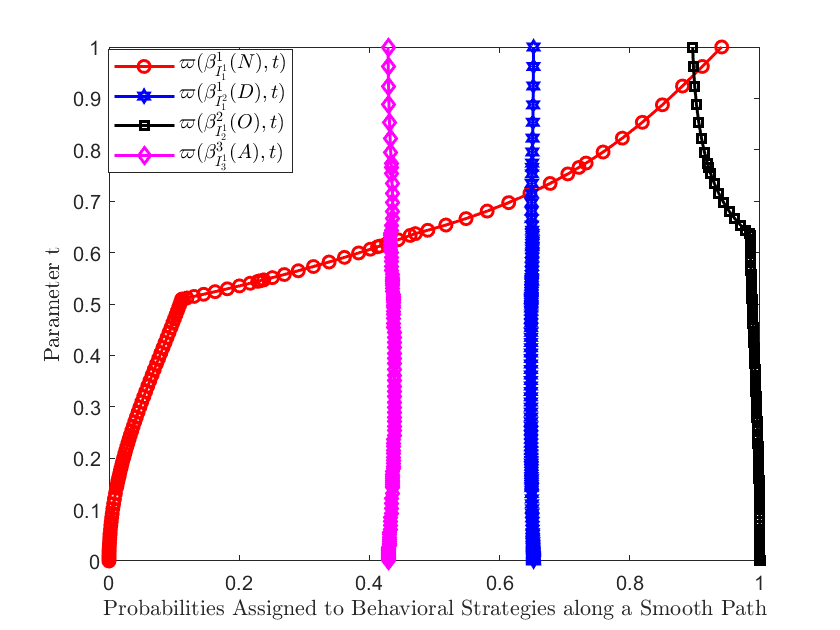}
        \caption{\label{fexm2cqupA}\scriptsize The Smooth Path of $\varpi(\beta,t)$ Specified by the System~(\ref{cqupes2}) for the Game in Fig.~\ref{fexm2}}
               \end{minipage}\hfill
    \begin{minipage}{0.49\textwidth}
        \centering
        \includegraphics[width=0.90\textwidth, height=0.12\textheight]{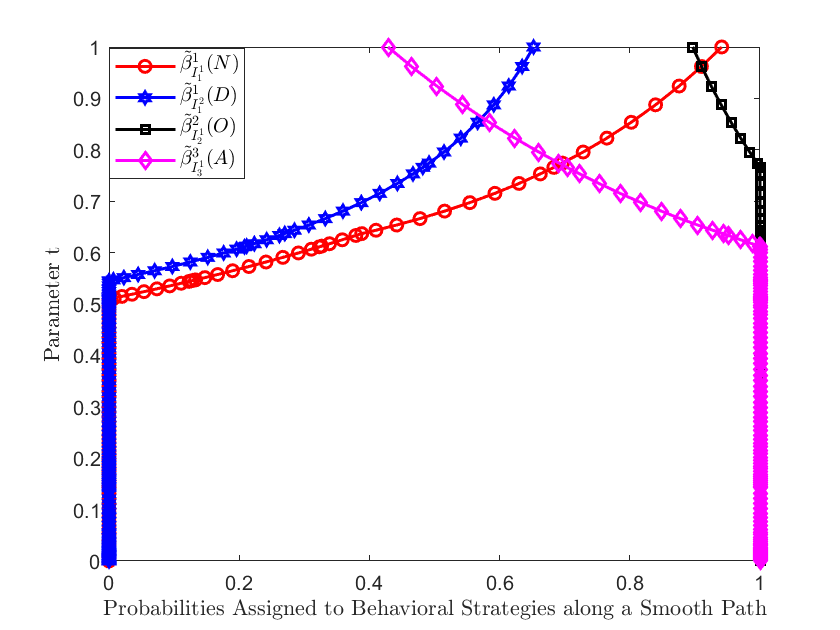}
        \caption{\label{fexm2cqupB}\scriptsize The Smooth Path of $\tilde\beta$ Specified by the System~(\ref{cqupes2}) for the Game in Fig.~\ref{fexm2}}
\end{minipage}
\end{figure}

  \begin{figure}[H]
    \centering
    \begin{minipage}{0.49\textwidth}
        \centering
        \includegraphics[width=0.90\textwidth, height=0.12\textheight]{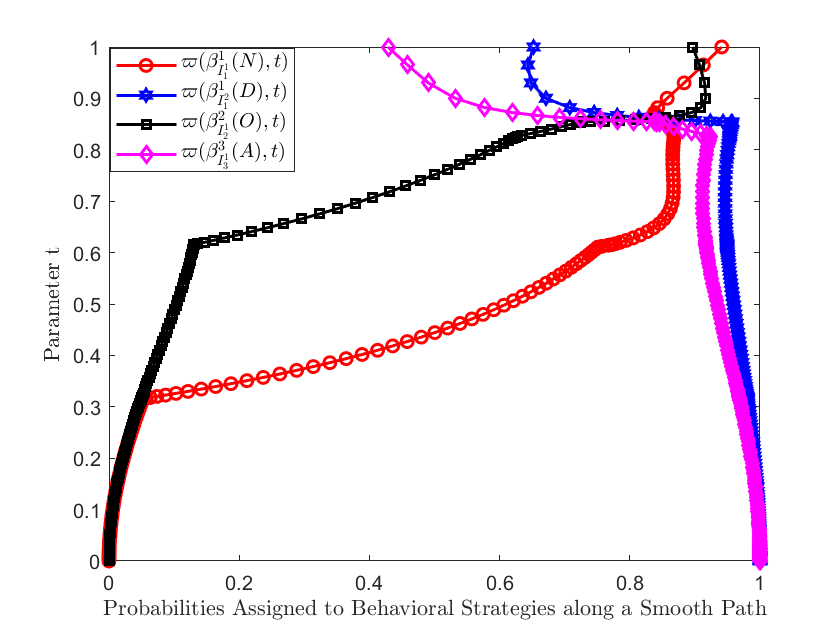}
        \caption{\label{fexm2sgcqupA}\scriptsize The Smooth Path of $\varpi(\beta,t)$ Specified by the System~(\ref{sgcqupes1}) for the Game in Fig.~\ref{fexm2}}
               \end{minipage}\hfill
    \begin{minipage}{0.49\textwidth}
        \centering
        \includegraphics[width=0.90\textwidth, height=0.12\textheight]{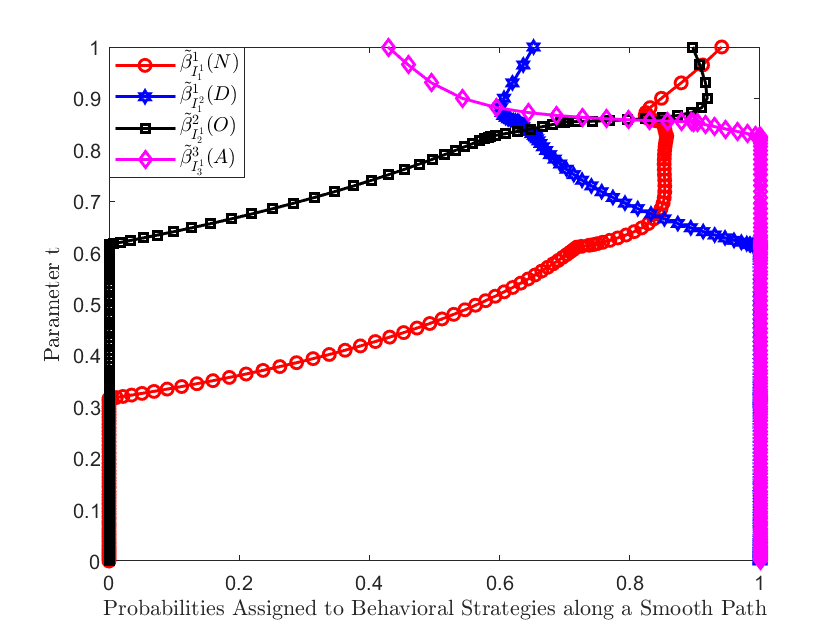}
        \caption{\label{fexm2sgcqupB}\scriptsize The Smooth Path of $\tilde\beta$ Specified by the System~(\ref{sgcqupes1}) for the Game in Fig.~\ref{fexm2}}
\end{minipage}
\end{figure}

One can observe from Figs.~\ref{fexm1cqupA} and~\ref{fexm1cqupB} and from Figs.~\ref{fexm2cqupA} and~\ref{fexm2cqupB} distinct trends in the two smooth paths of $\varpi(\beta,t)$ and $\tilde\beta$, which lead to diverse solutions. 
Figs.~\ref{fexm2cqupA} and~\ref{fexm2sgcqupA} demonstrate that the paths of $\varpi(\beta,t)$ specified by the system~(\ref{cqupes2}) and the system~(\ref{sgcqupes1}) start from the same point but lead to different solutions: The path in Fig.~\ref{fexm2cqupA} ends at a NashEBS but not subgame perfect, while that in Fig.~\ref{fexm2sgcqupA} reaches a subgame perfect equilibrium.

\section{Numerical Experiments}

We have adapted a standard predictor-corrector method as outlined in Eaves and Schmedders~\cite{Eaves and Schmedders (1999)} for numerically tracing the smooth paths specified by the systems~(\ref{Mlogbes4}) and~(\ref{cqupes2}) in this paper. The predictor-corrector methods have been coded in MATLAB.\footnote{The parameter values in the method are set as follows: the predictor step size $=$ $0.1*10^{0.2\ln{t}}$ and accuracy of a starting point for a corrector step $=$  $0.1*10^{0.5\ln{t}}$. The methods terminate as soon as the criterion of $t<10^{-5}$ is met.} The computation was carried out on Windows Server for Intel(R) Xeon(R) Gold 6426Y  2.50GHz  (2 processors) RAM 512GB.


To show the efficiency of the smooth paths specified by the systems~(\ref{Mlogbes4}) and~(\ref{cqupes2}), we have employed the methods to compute NashEBSs for three types of randomly generated large games, namely Type A, Type B, and Type C.\footnote{ In games with structure A, both players $1$ and $2$ have one information set. The number of information sets for player $i$ with $i\ge 3$ equals $\prod_{k=1}^{i-2}|A(I_{k})|$. In games with structure B, player $1$ has one information set. The number of information sets for player $i$ with $i\ge 2$ equals $|A(I_{i-1})|$. A Type C game shares a similar structure with a Type A game, with the distinction that Type C games consist of multiple layers, denoted as $L$. Within each layer, actions are sequentially taken by players from player $1$ to player $n$.} 
In our numerical experiments, each player has the same number of actions in all of his information sets, and each payoff along a terminal history is an integer uniformly drawn from $-10$ to $10$ and is assigned to zero with a random probability between $0$ and $50\%$. Within each game characterized by an explicit choice of $(n,m_i,|A(I_i)|)$ and $(n,m_i,|A(I_i)|, L)$, 10 examples were generated and solved. We denote by LOGM the method with the system~(\ref{Mlogbes4}) and by CQPM the method with the system~(\ref{cqupes2}). The efficiency of the methods is measured with the number of iterations and computational time.
The computational time (in seconds) and the number of iterations are presented in Tables~\ref{Table0}-\ref{Table2}.\footnote{The character ``-" in the tables denotes the failure of a method to find a Nash equilibrium of a game within $4.32\times 10^4$ seconds or $10^5$ iterations.} The numerical outcomes presented in Tables~\ref{Table0}-\ref{Table2} indicate that the LOGM surpasses the CQPM in terms of both computational time and number of iterations. Furthermore, the LOGM exhibits superior performance and can be effectively employed to compute NashEBSs in large-scale $n$-person extensive-form games.

\begin{table}[H]
\linespread{1} 
\scriptsize
\centering
\caption{{\footnotesize Numerical Performances of LOGM and CQPM for Type A Extensive-Form Games}}\label{Table0}
\begin{tabular*}{\hsize}{@{}@{\extracolsep{\fill}}cc|cc|cc@{}}
\hline
\multicolumn{2}{c}{} & \multicolumn{2}{c}{Computational Time} & \multicolumn{2}{c}{Number of Iterations}\\
 $n,m_i,|A(I_i)|$ & & LOGM & CQPM & LOGM & CQPM\\
\hline
  $3,(1,1,2),(2,10,10)$ & avg & 157.02 & 2298.92 & 112 & 1365 \\
 & min & 123.88 & 855.50 & 96 & 535\\
 & max & 273.85 & 5448.19 & 151 & 3948\\
 $3,(1,1,2),(2,15,15)$ & avg & 338.11 & 10768.05 & 732 & 2104\\
 & min & 196.99 & 1406.75 & 95 & 574\\
 & max & 668.95 & 27890.59 & 4800 & 4920\\
 $3,(1,1,2),(2,20,20)$ & avg & 503.87 & 11129.05 & 141 & 6196\\
 & min & 324.70 & 2238.75 & 96 & 514\\
 & max & 1006.04 & 20213.94 & 240 & 12916\\
 $3,(1,1,2),(2,25,25)$ & avg & 1042.86 & 13149.27 & 190 & 2341\\
 & min & 462.41 & 3866.67 & 93 & 544\\
 & max & 2711.59 & 30341.14 & 622 & 7639\\
\hline
\end{tabular*}
\end{table}

\begin{table}[H]
\linespread{1} 
\scriptsize
\centering
\caption{{\footnotesize Numerical Performances of LOGM and CQPM for Type B Extensive-Form Games}}\label{Table1}
\begin{tabular*}{\hsize}{@{}@{\extracolsep{\fill}}cc|cc|cc@{}}
\hline
\multicolumn{2}{c}{} & \multicolumn{2}{c}{Computational Time} & \multicolumn{2}{c}{Number of Iterations}\\
 $n,m_i,|A(I_i)|$ & & LOGM & CQPM & LOGM & CQPM\\
\hline
 $4,(1,2,2,5),(2,2,5,3)$ & avg & 230.39 & 3066.12 & 141 & 1424 \\
 & min & 141.07 & 715.89 & 106 & 500\\
 & max & 477.66 & 18296.47 & 195 & 5977\\
 $5,(1,2,2,2,5),(2,2,2,5,3)$ & avg & 1618.79 & 19275.86 & 320 & 3745\\
 & min & 692.54 & 4006.13 & 119 & 637\\
 & max & 2706.27 & 30889.78 & 625 & 7450\\
 $6,(1,2,2,2,2,5),(2,2,2,2,5,3)$ & avg & 5948.34 & - & 237 & -\\
 & min & 3287.39 & - & 144 & -\\
 & max & 12020.56 & - & 436 & -\\
 $7,(1,2,2,2,2,2,5),(2,2,2,2,2,5,3)$ & avg & 21856.21 & - & 255 & -\\
 & min & 18343.27 & - & 166 & -\\
 & max & 30634.37 & - & 414 & -\\
\hline
\end{tabular*}
\end{table}

\begin{table}[H]
\linespread{1} 
\scriptsize
\centering
\caption{{\footnotesize Numerical Performances of LOGM and CQPM for Type C Extensive-Form Games}}\label{Table2}
\begin{tabular*}{\hsize}{@{}@{\extracolsep{\fill}}cc|cc|cc@{}}
\hline
\multicolumn{2}{c}{} & \multicolumn{2}{c}{Computational Time} & \multicolumn{2}{c}{Number of Iterations}\\
 $n,m_i,|A(I_i)|,L$ & & LOGM & CQPM & LOGM & CQPM\\
\hline
 $2,(11,21),(2,2),3$ & avg & 469.29 & 5316.91 & 164 & 1364 \\
 & min & 367.55 & 1967.63 & 133 & 585\\
 & max & 670.59 & 17324.77 & 244 & 4312\\
 $2,(43,85),(2,2),4$ & avg & 7644.55 & - & 283 & -\\
 & min & 4192.12 & - & 184 & -\\
 & max & 14800.82 & - & 504 & -\\
 $2,(4,10),(3,3),2$ & avg & 2645.92 & 14742.55 & 249 & 1471\\
 & min & 1833.45 & 5516.90 & 182 & 457\\
 & max & 4458.17 & 30331.55 & 426 & 2889\\
 $3,(5,9,18),(2,2,2),2$ & avg & 2483.75 & 11852.13 & 237 & 1030\\
 & min & 1802.18 & 6478.21 & 174 & 520\\
 & max & 3441.77 & 17018.64 & 335 & 2255\\
\hline
\end{tabular*}
\end{table}


\section{Conclusions}

In this paper, we have presented a characterization of Nash equilibrium in behavioral strategies by introducing an extra behavioral strategy profile and beliefs, which meets the principles of local sequential rationality and self-independent belief systems. Our characterization offers a necessary and sufficient condition for identifying a NashEBS in the form of a polynomial system. The advantages of our approach are notable, as it enables the analytical determination of all NashEBSs in various small-scale extensive-form games. Additionally, a characterization of subgame perfect equilibrium in behavioral strategies is achieved. Moreover, to further boost practical applications of NashEBS and subgame perfect equilibrium, we have capitalized on our characterizations to develop differentiable path-following methods to compute such equilibria. Comprehensive numerical results have been provided to demonstrate the efficiency of the methods. 
The idea of this paper can be exploited to characterize Reny's~\cite{Reny (1992)} weak sequential rationality through the introduction of an extra behavioral strategy profile and an application of local sequential rationality.


{\footnotesize
\begin{thebibliography}{00}
\setlength{\itemsep}{-1mm}
\begin{spacing}{1.0}

\bibitem{Battigalli (1996)} Battigalli P (1996) Strategic independence and perfect Bayesian equilibria. Journal of Economic Theory, 70(1): 201-234.

\bibitem{Battigalli and Siniscalchi (2002)} Battigalli P \& Siniscalchi M (2002) Strong belief and forward induction reasoning. Journal of Economic Theory, 106(2): 356-391.

\bibitem{Ben-Porath (1997)}  Ben-Porath E (1997) Rationality, Nash equilibrium and backwards induction in perfect-information games. The Review of Economic Studies, 64(1): 23-46.

\bibitem{Bonanno (2018)} Bonanno G (2018) Game Theory: Volume 1: Basic Concepts. CreateSpace Independent Publishing Platform.

\bibitem{Eaves and Schmedders (1999)} Eaves BC \& Schmedders K  (1999) General equilibrium models and homotopy methods. Journal of Economic Dynamics \& Control, 23: 1249-1279.

\bibitem{Fudenberg and Tirole (1991)} Fudenberg D \& Tirole J (1991) Perfect Bayesian equilibrium and sequential equilibrium. Journal of Economic Theory, 53(2): 236-260.

\bibitem{Harsanyi and Selten (1988)} Harsanyi JC \& Selten R (1988) A General Theory of Equilibrium Selection in Games. MIT Press.


\bibitem{Herings (2000)} Herings PJJ  (2000) Two simple proofs of the feasibility of the linear tracing procedure. Economic Theory, 15: 485-490.

\bibitem{Herings and Peeters (2001)} Herings PJJ \& Peeters RJAP (2001) A differentiable homotopy to compute Nash equilibria of $n$-person games. Economic Theory, 18(1): 159-185.


\bibitem{Herings and Peeters (2010)}  Herings PJJ \& Peeters RJAP (2010) Homotopy methods to compute equilibria in game theory. Economic Theory, 42(1): 119-156.

\bibitem{Koller et al. (1996)} Koller D, Megiddo N \& von Stengel B (1996) Efficient computation of equilibria for extensive two-person games. Games and economic behavior, 4(2): 528-552.

\bibitem{Koller and Megiddo (1992)} Koller D \& Megiddo N (1992) The complexity of two-person zero-sum games in extensive form. Games and economic behavior, 4(4): 528-552.

\bibitem{Kreps and Wilson (1982)} Kreps DM \& Wilson R (1982) Sequential equilibria. Econometrica, 50(4): 863-894.

\bibitem{Kuhn (1953)} Kuhn HW (1953) Extensive games and the problem of information. Annals of Mathematics Studies, 28: 193-216.

\bibitem{Mas-Colell (1974)} Mas-Colell A (1974) A note on a theorem of F. Browder. Mathematical Programming, 6: 229-233.

\bibitem{Mas-Colell et al. (1995)} Mas-Colell A, Whinston MD \& Green JR (1995) Microeconomic Theory. New York: Oxford University Press.

\bibitem{Myerson (1986)} Myerson RB (1986) Multistage games with communication. Econometrica, 54(2): 323-358.

\bibitem{Myerson (1991)} Myerson RB (1991) Game Theory: Analysis of Conflict. Harvard University Press.

\bibitem{Nash (1951)} Nash JrJF (1951) Noncooperative games. Annals of Math, 54: 289-295.

\bibitem{Osborne and Rubinstein (1994)} Osborne MJ \& Rubinstein A (1994) A Course in Game Theory. MIT Press.

\bibitem{Reny (1992)}  Reny PJ (1992) Backward induction, normal form perfection and explicable equilibria. Econometrica: Journal of the Econometric Society, 627-649.

\bibitem{Selten (1965)} Selten R (1965) Spieltheoretische behandlung eines oligopolmodells mit nachfragetragheit. Zeitschrift fur die Gesamte Staatswissenschaft. Journal of Institutional and Theoretical Economics, 121(3): 301-324.

\bibitem{Selten (1975)} Selten R (1975) Reexamination of the perfectness concept for equilibrium points in extensive games. International Journal of Game Theory, 4: 25-55.

\bibitem{von Stengel (1996)} von Stengel B (1996) Efficient computation of behavior strategies. Games and economic behavior, 14(2): 220-246.

\bibitem{Wilson (1972)} Wilson R (1972) Computing equilibria of two-person games from the extensive form. Management Science, 18(7): 448-460.

\end{spacing}
\end{thebibliography}
}
\end{document}